\documentclass[letterpaper, 10pt]{IEEEtran}  

\IEEEoverridecommandlockouts                              

\usepackage{amsmath}
\usepackage{mathrsfs}
\usepackage{amssymb}
\usepackage{amsfonts}
\usepackage{amsthm}
\usepackage{cite}
\usepackage{gauss}
\usepackage{dsfont}
\usepackage{graphicx}
\usepackage{epstopdf}
\usepackage{hyperref}
\usepackage{enumerate}
\usepackage{listings}
\usepackage{lipsum}
\usepackage{multirow}
\usepackage{chemarrow}
\usepackage{extarrows}
\usepackage{bbm}
\usepackage{bm}
\usepackage{url}
\usepackage{mathtools}
\usepackage{etoolbox}
\usepackage{color}
\usepackage{tikz}
\usepackage{verbatim}

\newcommand{\E}{\mathbb{E}}
\renewcommand{\Pr}{\mathbb{P}}

\newcommand{\bX}{\mathbf{X}}

\newcommand{\bY}{\mathbf{Y}}
\newcommand{\bU}{\mathbf{U}}

\newcommand{\prescription}{\varGamma}

\newtheorem{thm}{Theorem}
\newtheorem{lemma}{Lemma}
\newtheorem{prop}{Proposition}

\theoremstyle{definition}
\newtheorem{assump}{Assumption}
\newtheorem{defn}{Definition}
\newtheorem{example}{Example}
\newtheorem{remark}{Remark}

\mathtoolsset{showonlyrefs}  
\allowdisplaybreaks  

\title{
Dynamic Games among Teams with Delayed Intra-Team Information Sharing
}

\author{Dengwang Tang, Hamidreza Tavafoghi, Vijay Subramanian, Ashutosh Nayyar, Demosthenis Teneketzis
\thanks{This work is supported by NSF Grant No. ECCS 1750041, ECCS 2038416, ECCS 1608361, CCF 2008130, ARO Award No. W911NF-17-1-0232, and MIDAS Sponsorship Funds by General Dynamics.}
\thanks{D. Tang, V. Subramanian, and D. Teneketzis are with Electrical and Computer Engineering,
        University of Michigan, Ann Arbor, MI, 48109, USA. E-mail:
        {\tt\small dwtang@umich.edu, vgsubram@umich.edu, teneket@umich.edu}.}%
\thanks{H. Tavafoghi is with Mechanical Engineering, University of California, Berkeley, CA, 94720, USA. E-mail: {\tt\small tavaf@berkeley.edu}. }%
\thanks{A. Nayyar is with the Ming Hsieh Department of Electrical and Computer Engineering, University of Southern California, Los Angeles, CA, 90089, USA. E-mail: {\tt\small ashutosn@usc.edu}.}%
}

\begin{document}

\maketitle
\thispagestyle{empty}
\pagestyle{empty}

\begin{abstract}
We analyze a class of stochastic dynamic games among teams with asymmetric information, where members of a team share their observations internally with a delay of $d$. Each team is associated with a controlled Markov Chain, whose dynamics are coupled through the players' actions. These games exhibit challenges in both theory and practice due to the presence of signaling and the increasing domain of information over time. We develop a general approach to characterize a subset of Nash Equilibria where the agents can use a compressed version of their information, instead of the full information, to choose their actions. We identify two subclasses of strategies: Sufficient Private Information Based (SPIB) strategies, which only compress private information, and Compressed Information Based (CIB) strategies, which compress both common and private information. 
We show that while SPIB-strategy-based equilibria always exist, the same is not true for CIB-strategy-based equilibria.
We develop a backward inductive sequential procedure, whose solution (if it exists) provides a CIB strategy-based equilibrium. We identify some instances where we can guarantee the existence of a solution to the above procedure. Our results highlight the tension among compression of information, existence of (compression based) equilibria, and backward inductive sequential computation of such equilibria in stochastic dynamic games with asymmetric information.
\end{abstract}

\section{INTRODUCTION}


Dynamic games with asymmetric information appear in many socioeconomic contexts. In these games, multiple agents/decision makers interact repeatedly in a changing environment. Agents have different information and seek to optimize their respective long-term payoffs. For example, multiple companies may compete with each other in a market over time and each company attempts to optimize its own long-term benefits \cite{maskin1988theory1,maskin1988theory2,doganoglu2003dynamic,bergemann2006pricecompetition,cabral2011dynamic}; the market is also changing over time driven by the actions the companies take. Another instance of such games arises in cyberphysical systems \cite{hamidbookchapter,amin2012cyber,amin2015game,zhu2015game,shelar2016security}; at each time, attackers make decisions on which hosts to attack, and the system administrators/defenders choose actions to defend against the attackers, for example, by isolating some hosts from the rest of the system \cite{hamidbookchapter}; the system's state changes over time as a result of the attackers' and defenders' actions. In all instances of these games, when an agent takes an action, she needs to consider not only how the action will affect her current payoff but also how it will influence the system's evolution and the future actions of all agents, and hence her future payoffs.

In some settings, agents can form groups, or teams \cite{colombino2017mutually,summers2017information}. The agents in the same group share a common goal but may have different information available to them. This information asymmetry among teammates appears in many engineering applications. In most of these applications, the state of the system changes fast, and agents have to make real-time decisions. Moreover, the communication between agents is either costly, or restricted by bandwidth or delay. Examples of our settings include competing fleets of automated cars from rival companies \cite{hancock2019future} and the DARPA Spectrum Challenge \cite{spectrumchallenge}. In the DARPA Spectrum challenge setup, individual transceivers work in teams to maximize the sum throughput of their networks. Teams compete with other teams, and members of the same team need to coordinate and evolve their responses over time. In these settings, agents in the same team aim to choose their strategy jointly to achieve team optimality (i.e. to choose the joint strategy profile that maximizes the expected utility of the team over all joint strategy profiles) rather than just person-by-person optimality (a team strategy is person-by-person optimal, PBPO, when each team member's strategy is an optimal response to other team members' strategy profile). We study a stylized model of such settings in this paper.

It is worth stating that the games among teams problems we focus on in this paper are different from cooperative games in economics research (e.g. see  \cite{myerson2013game} Chapters 8-10). In cooperative game theory, the goal is to study the group formation process among agents with different objectives. In our setting, groups are assumed to be fixed and given, and we focus instead on determining the optimal actions and payoffs for each group. A unilateral deviation in our problems means one or more agents in one group deviates, but the community structure of the agents stays the same.

There are three main challenges that need to be addressed when studying dynamic games among individual players: (i) the agents' decisions and information are interdependent over time. In particular, \emph{signaling} is present in these games, i.e. agents actively infer other agents' private information based on their actions and their strategy; (ii) the domain of the agents' strategy grows over time; (iii) signaling in games is more challenging and subtle than in team problems due to the diverging incentives of the agents. Games of teams inherit all the above challenges. Moreover, we have the additional challenge of coordination within asymmetrically informed team members to achieve team optimality instead of person-by-person optimality.

In this paper we propose a general approach to characterize a subset of equilibrium strategies of dynamic games among teams with the following goals: (i) to determine appropriate compression of information for each agent to base their decision on; (ii) to develop a sequential decomposition of the game. In addition we would like to determine conditions sufficient to guarantee the existence of such equilibrium strategies.

\subsection{Related Literature}
To understand games among teams, we first examine a team's best-response strategy when other teams' strategies are fixed. Team problems, or decentralized control problems, 
have been extensively studied in the control literature.
Researchers have developed various methodologies/approaches to decentralized control problems to determine team optimal strategies or PBPO strategies, and to determine structural results/properties for the above mentioned strategies. These methodologies include: (i) the person-by-person approach \cite{witsenhausen1979structure,walrand1983optimal,teneketzis2006structure,nayyar2011structure,kaspi2010structure,tenney1981detection,tsitsiklis1989decentralized,teneketzis1987decentralized,veeravalli1993decentralized,veeravalli1994decentralized,nayyar2011sequential,teneketzis1984decentralized,veeravalli2001decentralized,varaiya1983causal,mahajan2009optimal} (ii) the designer's approach \cite{witsenhausen1973standard,mahajan2008sequential} (iii) the coordinator's approach \cite{nayyar2010optimal,nayyar2013decentralized,mahajan2013controlsharing,tavafoghi2018unified}.
The person-by-person approach has been used to determine qualitative/structural properties of team optimal or PBPO strategies. In this approach, the strategies of all team members/agents except one, say agent $i$, are assumed to be arbitrary but fixed; then the qualitative properties of agent $i$'s best response strategy are determined. These properties are then valid for all possible (fixed) strategies of the other agents.
The designer's approach investigates the decentralized control/team problem from the point of view of a designer who knows the system model and the joint probability distributions of the primitive random variables (the system's initial state, the noise driving the system, and the noise in the agents' observations). The designer chooses the strategies of all team members at time 0 by solving an open-loop stochastic control problem, where her decision at each time is the strategy/control law for all the team members/agents. Applying stochastic control results, the designer can obtain a dynamic programming decomposition.
The methodology developed in this paper is inspired by the coordinator's approach used in \cite{nayyar2010optimal,nayyar2013decentralized,tavafoghi2018unified}.
Similar to the designer's approach, the coordinator's approach assumes that a fictitious agent, called the coordinator, assigns instructions to agents. However, unlike the designer's approach, the coordinator is assumed to know the common information of all agents, and assigns partial strategies (prescriptions) instead of full strategies to agents. The partial strategies tell an agent how to utilize her private information to generate actions. Both the designer's approach and the coordinator's approach lead to the determination of globally optimal team strategy profiles. 

Research on dynamic games roughly consists of two directions: One direction focuses on repeated games or multi-stage games, where the instantaneous payoffs at each stage is only affected by actions in this stage but not by the actions in the previous stages. In these games, researchers investigated long term interactions among agents (e.g. punishment and reward strategies) and characterized the set of equilibrium payoffs (e.g. see \cite{mailath2006repeated} or \cite{myerson2013game} Chapter 7). The other direction focuses on games with an underlying dynamic system, in other words, games where instantaneous payoffs can be affected by previous actions. In this more complicated setting, researchers attempted to develop methodologies for the determination of equilibria with either a general structure or a specialized structure. In this paper we focus on the latter direction.

Games of individual agents (i.e. agents do not form teams) with an underlying dynamic system have been studied in both the economics and the control literature. Dynamic games with symmetric information have been studied extensively \cite{bacsar1998dynamic,filar2012competitive}. In \cite{maskin2001markov}, the authors propose the concept of Markov Perfect Equilibrium (MPE) for the case where the state of the system and agents' actions are perfectly observable. The research on dynamic games with asymmetric information can be classified into two categories: zero-sum games and general (i.e. not necessarily zero-sum) games. Zero-sum games are analyzed in  \cite{renault2006value,renault2012value,zheng2013decomposition,gensbittel2015value,li2014lp,li2017solving,cardaliaguet2016markov,kartik2020upper}. In these works, the authors take advantage of many properties of zero-sum games, such as having a unique value and the interchangeability of equilibrium strategies. These properties do not extend to general non-zero-sum games. The literature on general dynamic games includes \cite{maskin2013youtube,nayyar2012dynamic,gupta2014common,gupta2016dynamic,ouyang2015oligopoly,nayyar2013common,ouyang2016dynamic,tavafoghi2016stochastic,hamidthesis,vasal2019spbe}.
In \cite{nayyar2013common}, the authors extend the MPE concept in \cite{maskin2001markov} to the case where the underlying dynamics is only partially observable. Under the crucial assumption that the common information based (CIB) belief is strategy-independent, the authors prove that there exist equilibria where agents play \emph{CIB strategies}, i.e. the agents choose their actions based on CIB belief and private information instead of full information. Furthermore, such equilibria can be found through a sequential decomposition of the game. 
In our setup the system state is not perfectly observed, thus our model is distinctly different from that of \cite{maskin2001markov}. Furthermore, in contrast to \cite{nayyar2013common}, the CIB belief in our model is strategy-dependent. 

The closest work to our paper in terms of both model and approach is \cite{ouyang2016dynamic}. In \cite{ouyang2016dynamic}, the authors consider a game model where, in contrast to \cite{nayyar2013common}, the CIB beliefs are strategy-dependent. 
They propose the concept of Common Information Based Perfect Bayesian Equilibrium (CIB-PBE) as a solution concept for this game model and prove that CIB-PBE can be found through a sequential decomposition whenever this decomposition has a solution. The game model of \cite{ouyang2016dynamic} has multiple features that prevent us from directly applying their results in our analysis in Section \ref{sec: 4}. We will make a more detailed comparison in Section \ref{sec: goc}. Our work is also close in spirit to \cite{maskin2013youtube}. In \cite{maskin2013youtube}, the authors extend their work in \cite{maskin2001markov} by considering games where actions are observable but each agent has a fixed, private utility type.  
They propose Markov Sequential Equilibrium (MSE) as a solution concept for these games, where the agents choose their actions based on a compression of their information along with their beliefs on the types of other agents. The authors show by example that MSE do not necessarily exist. As an alternative to MSE they propose a new concept obtained from limits of $\varepsilon$-MSE as $\varepsilon$ goes to 0. 

Unlike either team problems or dynamic games among individual agents, games among teams (in particular, ones with an underlying dynamic system) have not been systematically studied in the literature. There are only a few works on special models of games among teams. In \cite{farina2018ex} and \cite{zhang2020computing}, the authors proposed algorithms to compute equilibria for zero-sum multiplayer extensive form games, where a team of players plays against an adversary.
In \cite{anantharam2007common} the authors provide an example of a zero-sum game which involves a team. However the players in this team have symmetrical information, hence the team is equivalent to an individual player with vector-valued actions. 
In \cite{nayyar2012dynamic} the authors briefly extend their results in \cite{nayyar2013common} to games among teams for a specialized model where the CIB belief is strategy independent. In both \cite{colombino2017mutually} and \cite{summers2017information} the authors solve a two-team zero-sum linear quadratic stochastic dynamic game. In \cite{bhattacharya2012multi} the authors formulate and solve a game between two teams of mobile agents. The model and information structure of \cite{bhattacharya2012multi} are different from ours.
Additionally, games among teams have been the subject of empirical research (see, for example, \cite{cox2018strategic,cooper2005two}). In our work, we study analytically a model of non zero-sum dynamic stochastic games among teams where the CIB belief is strategy dependent.
\subsection{Contribution}

In this paper, we consider a model of dynamic games among teams with asymmetric information. We assume that each team is associated with a dynamical system that has Markovian dynamics driven by the actions of all agents of all teams. The state of each dynamical system is assumed to be vector-valued, where each component represents an agent's local state. Agents can observe their own local states perfectly and communicate them within their respective teams with a delay of $d$. All actions are public, i.e., observable by every agent in every team. We also assume the presence of public noisy observations of the system's state. The instantaneous reward of a team depends on the states and actions of all teams. Our model is a generalization of the model in \cite{ouyang2016dynamic} to competing teams.

Our contributions are as follows: 
\begin{itemize}
	\item We identify appropriate compression of information for each agent. The compression is achieved in two steps: (i) the compression of team-private information that depends only on the team strategy; (ii) the compression of common information that depends on the strategy of all agents. The compression steps induce two special classes of strategies: (i) Sufficient Private Information Based (SPIB) strategies, where agents only apply the first step of compression  (ii) Compressed Information Based (CIB) strategies, where agents apply both steps of compression.
	
	\item We develop a sequential decomposition of the game where agents play CIB strategies. We show that any solution of the sequential decomposition forms a Nash Equilibrium of the game.  
	
	\item We show that SPIB-strategy-based Nash Equilibria always exist, while CIB-strategy-based Nash Equilibria do not always exist. We identify some simple instances where CIB-strategy-based equilibria are guaranteed to exist.
\end{itemize}

In a broader context, our results highlight the conflicts between compression of information, sequential decomposition, and existence of equilibria that occur in a wide range of dynamic games with asymmetric information, reiterating the message in \cite{maskin2013youtube}: In general, compression can hurt the ability to sustain equilibria, since the full history can allow for a finer calibration of the agents' strategies.

\subsection{Organization}

We organize the rest of the paper as follows: In Section \ref{sec: problemformulation} we formally present our model and problem. In Section \ref{sec: goc} we transform the game among teams into an equivalent game among coordinators where each coordinator represents a team. In Section \ref{sec: spi} we introduce our first step of compression of information and SPIB strategies, and we show the existence of SPIB-strategy-based equilibria. In Section \ref{sec: 4} we introduce the second step of compression and CIB strategies, and we provide a sequential decomposition of the game. We also show the general non-existence of CIB-strategy-based equilibria and provide some conditions for existence. We present some extensions and special cases of our results in Section \ref{sec: add}. Then we discuss our results in Section \ref{sec: dis}. We conclude in Section \ref{sec: concl}. Proof details are provided in the Appendix.

\subsection{Notation}
We use capital letters to represent random variables, bold capital letters to denote random vectors, and lower case letters to represent realizations. We use superscripts to indicate teams and agents, and subscripts to indicate time. We use $i$ to represent a typical team, and $-i$ represents all teams other than $i$. We use $t_1:t_2$ to indicate the collection of timestamps $(t_1, t_1+1, \cdots, t_2)$. For example $X_{5:8}^1$ stands for the random vector $(X_5^1, X_6^1, X_7^1, X_8^1)$. For random variables or random vectors, we use the corresponding script capital letters (italic capital letters for greek letters) to denote the space of values these random vectors can take. For example, $\mathcal{H}_t^i$ denotes the space of values the random vector $H_t^i$ can take. The products of sets in this paper are Cartesian products. We use $\Pr(\cdot)$ and $\E[\cdot]$ to denote probabilities and expectations, respectively. We use $\Delta(\varOmega)$ to denote the set of probability distributions on a finite set $\varOmega$. When writing probabilities, we will omit the random variables when the lower case letters that represent the realizations clearly indicates the random variable it represents. For example, we will use $\Pr(y_t^i|x_t, u_t)$ as a shorthand for $\Pr(Y_t^i = y_t^i|\bX_t = x_t, \bU_t=u_t)$. When $\lambda$ is a function from $\varOmega_1$ to $\Delta(\varOmega_2)$, with some abuse of notation we write $\lambda(\omega_2|\omega_1):=(\lambda(\omega_1))(\omega_2)$ as if $\lambda$ is a conditional distribution. We use $\bm{1}_A$ to denote the indicator random variable of an event $A$.

In general, probability distributions of random variables in a dynamic system are only well defined after a complete strategy profile is specified. We specify the strategy profile that defines the distribution in superscripts, e.g. $\Pr^g(x_t^i|h_t^0)$. When the conditional probability is independent of a certain part of the strategy $(g_t^i)_{(i, t)\in \varOmega}$, we may omit this part of the strategy in the notation, e.g. $\Pr^{g_{1:t-1}}(x_{t}|y_{1:t-1}, u_{1:t-1})$, $\Pr^{g^{i}}(x_{t}^i|h_t^0)$ or $\Pr(x_{t+1}|x_{t}, u_{t})$. We say that a realization of some random vector (for example $h_t^0$) is \emph{admissible} under a partially specified strategy profile (for example $g^{-i}$) if the realization has strictly positive probability under some completion of the partially specified strategy profile (In this example, that means $\Pr^{g^i, g^{-i}}(h_t^0) > 0$ for some $g^i$). Whenever we write a conditional probability or conditional expectation, we implicitly assume that the condition has non-zero probability under the specified strategy profile. When only part of the strategy profile is specified in the superscript, we implicitly assume that the condition is admissible under the specified partial strategy profile.

\section{PROBLEM FORMULATION}\label{sec: problemformulation}
\subsection{System Model and Information Structure}
We consider a finite horizon dynamic game among finitely many teams each consisting of a finite number of agents, where agents have asymmetric information. Let $\mathcal{I}=\{1, \cdots, I\}$ denote the set of teams and $\mathcal{T}=\{1,\cdots, T\}$ denote the set of time indices. We use a tuple $(i, j)$ to indicate the $j$-th member of team $i$. For a team $i\in\mathcal{I}$, let $\mathcal{N}_i=\{(i, 1), \cdots, (i, N_i) \}$ denote team $i$'s members. Let $\mathcal{N} = \bigcup_{i\in\mathcal{I}}\mathcal{N}_i$ denote the set of all agents. At each time $t\in\mathcal{T}$, each agent $(i, j)$ selects an action $U_t^{i, j}\in\mathcal{U}_t^{i, j}$, where $\mathcal{U}_t^{i, j}$ denotes the action space of agent $(i, j)$ at time $t$. Each team is associated with a vector-valued dynamical system $\bX_t^{i} = (X_t^{i, j})_{(i, j)\in\mathcal{N}_i}$ which evolves according to
\begin{equation}\label{eq: system}
\bX_{t+1}^{i} = f_t^{i}(\bX_t^{i}, \mathbf{U}_t, W_t^{i, X}),\quad i\in\mathcal{I},
\end{equation}
where $\bU_t=(U_t^{k, j})_{(k, j)\in\mathcal{N}}$, and $(W_t^{i, X})_{i\in\mathcal{I}, t\in\mathcal{T}}$ is the noise in the dynamical system. We assume that $X_t^{i, j} \in\mathcal{X}_t^{i, j}$ for $(i, j)\in \mathcal{N}$. 

We assume that the actions of all agents are publicly observed. Further, at time $t$, after all the agents take actions, a public observation of team $i$'s state is generated according to
\begin{equation}\label{eq: obs}
Y_t^i = \ell_t^i(\bX_t^i, \mathbf{U}_t, W_t^{i, Y}),\quad i\in\mathcal{I},
\end{equation}
where $Y_t^i\in\mathcal{Y}_t^i$, and $(W_t^{i, Y})_{ i\in\mathcal{I}, t\in\mathcal{T}}$ are the observation noises. 

The order of events occuring between time steps $t$ and $t+1$ is shown in the figure below:
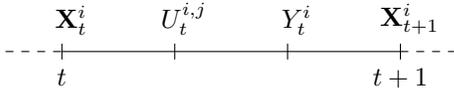
\begin{center}
\begin{tikzpicture}
\draw[dashed] (-0.75, 0) -- (0, 0);
\draw (0,0) -- (4.5,0);
\draw[dashed] (4.5, 0) -- (5.25, 0);

\foreach \x in {0,1.5,3,4.5}
\draw (\x cm,3pt) -- (\x cm,-3pt);

\draw (0,0) node[below=3pt] {$ t $} node[above=3pt] {$~~\bX_t^i $};
\draw (1.5,0) node[below=3pt] {$  $} node[above=3pt] {$~~U_t^{i, j} $};
\draw (3,0) node[below=3pt] {$  $} node[above=3pt] {$~~Y_t^{i} $};
\draw (4.5,0) node[below=3pt] {$ t+1 $} node[above=3pt] {$~~\bX_{t+1}^{i} $};
\end{tikzpicture}
\end{center}

We assume that the functions $(f_t^i)_{ i\in\mathcal{I}, t\in\mathcal{T} }, (\ell_t^i)_{ i\in\mathcal{I}, t\in\mathcal{T} }$ are common knowledge among all agents. We further assume that $(\bX_1^i)_{i\in\mathcal{I}}, (W_t^{i, X} )_{ i\in\mathcal{I}, t\in\mathcal{T} }$, and $(W_t^{i, Y} )_{ i\in\mathcal{I}, t\in\mathcal{T} }$ are mutually independent primitive random variables whose distributions are also common knowledge among all agents. As a result, the teams' dynamics $(\bX_t^i)_{t\in \mathcal{T}}, i\in\mathcal{I}$ are conditionally independent given the actions, and the public observations of different teams' systems are conditionally independent given the states and actions of all teams.

At each time $t$, the following information is available to all agents:
\begin{equation}
H_t^0 = (\bY_{1:t-1}, \bU_{1:t-1}),
\end{equation} 
where $\bY_t=(Y_t^i)_{i\in\mathcal{I}}, \bU_t=(U_t^{i, j})_{(i, j)\in\mathcal{N}}$.
We refer to $H_t^0$ as the common information among teams.

We assume that each agent $(i, j)$ observes her own state $X_t^{i, j}$. Further, agents in the same team share their states with each other with a time delay $d\geq 1$. Thus, at time $t$, all agents in team $i$ have access to $H_t^i$, given by
\begin{equation}
H_t^i = (\bY_{1:t-1}, \bU_{1:t-1}, \bX_{1:t-d}^i), \quad i\in\mathcal{I}.
\end{equation}
We call $H_t^i$ the common information within team $i$.

Finally, the information available to agent $(i, j)$ at time $t$, denoted by $H_t^{i, j}$, is
\begin{equation}
H_t^{i, j} = (\bY_{1:t-1}, \bU_{1:t-1}, \bX_{1:t-d}^i, X_{t-d+1:t}^{i, j}),\quad (i, j)\in\mathcal{N}.
\end{equation}

This model captures the hierarchy of information asymmetry among teams and team members. {It is an abstract representation of dynamic oligopoly games \cite{ouyang2015oligopoly,ouyang2016dynamic} where each member of the oligopoly is a team.}

\begin{remark}\label{remark: onepersonteam}
	Our model also captures the scenarios where a team has only one member. Such a team can be incorporated in our framework by adding a dummy agent to it and assuming a suitable internal communication delay $d$. If all teams are single-member teams, then $d$ can be arbitrarily chosen.
\end{remark}

To illustrate the key ideas of the paper without dealing with technical difficulties arising from continuum spaces, we assume that all the system random variables (i.e. all states, actions, and observations) take values in finite sets.
\begin{assump}
	$\mathcal{X}_t^{i, j}, \mathcal{Y}_t^i, \mathcal{U}_t^{i, j}$ are finite sets for all $(i, j)\in\mathcal{N}, t\in\mathcal{T}$.
\end{assump}

\subsection{Strategies and Reward Functions}\label{sec: rew}
For games among teams, there are three possible types of team strategies one could consider: (1) pure strategies, i.e. deterministic strategies; and (2) randomized strategies where team members independently randomize; (3) randomized strategies where team members jointly randomize.

A pure strategy profile of a team is a collection of functions $\mu^i=(\mu_t^{i, j})_{ (i, j)\in\mathcal{N}_i, t\in\mathcal{T}}$, where ${\mu}_t^{i, j}: \mathcal{H}_t^{i, j} \mapsto \mathcal{U}_t^{i, j}$. Define $\mathcal{M}_t^{i, j}$ as the space of functions from $\mathcal{H}_t^{i, j}$ to $\mathcal{U}_t^{i, j}$. Let $\mathcal{M}^i=\prod_{t\in\mathcal{T}}\prod_{(i, j)\in\mathcal{N}_i}\mathcal{M}_t^{i, j}$. Any randomized strategy of a team, either of type 2 or type 3, can be described through a mixed strategy $\sigma^i\in \Delta(\mathcal{M}^i)$. In particular, if team members independently randomize, the mixed strategy $\sigma^i$ being used to describe the strategy profile will be a product of measures on $\mathcal{M}^{i, j}=\prod_{t\in\mathcal{T}}\mathcal{M}_t^{i, j}$ for $(i, j)\in \mathcal{N}_i$.

Team $i$'s total reward under a pure strategy profile $\mu = (\mu_t^{i, j})_{ (i, j)\in\mathcal{N}, t\in\mathcal{T}}$ is
\begin{equation}
J^i({\mu}) = \E^\mu\left[\sum_{t\in\mathcal{T}} r_t^i(\bX_t, \bU_t)\right],
\end{equation}
where the functions $(r_t^i)_{ i\in\mathcal{I}, t\in\mathcal{T} }, r_t^i: \mathcal{X}_t\times \mathcal{U}_t \mapsto \mathbb{R}$, representing the instantaneous rewards, are common knowledge among all agents. Team $i$'s total reward under a mixed strategy profile $\sigma=(\sigma^{i})_{i\in\mathcal{I}}, \sigma^{i}\in \Delta(\mathcal{M}^i)$, is then an average of the total rewards under pure strategy profiles, i.e.
\begin{equation}
	J^i(\sigma) = \sum_{\mu\in \mathcal{M}} \left(\prod_{i\in\mathcal{I}} \sigma^i(\mu^i)\right) J^i(\mu).
\end{equation}

Note that while members of the same team may jointly randomize their strategies, the randomizations of different teams are independent of each other.

\begin{remark}\label{remark: initial}
	For convenience of notation and proofs, for $t\in \{-(d-1), \cdots, -1, 0\}$, we define $\mathcal{X}_t^{i, j} = \mathcal{U}_t^{i, j} = \mathcal{Y}_t^i = \{0\}$ and $r_t^i(\bX_t, \bU_t) = 0$ for all $i\in\mathcal{N}$ and $(i, j)\in \mathcal{N}$.
\end{remark}

\subsection{Solution Concept}\label{sec: solconcept}
In this work, a team refers to a group of agents that have asymmetric information and the same objective. Because of the shared objective, members of the same team can jointly decide on the strategy to use before the start of the game for the collective benefit of the team. Hence, we can assume that every member of the team knows the strategy of the others in the team. Therefore, when considering an equilibrium concept, we should consider team deviations rather than individual deviations, i.e. multiple members of the same team may decide to play a different strategy than the equilibrium strategy.
We consider randomized strategies where team members jointly randomize. Example \ref{ex: 1} of Section \ref{sec: example1} illustrates why such strategies must be considered when we study games among teams.


The above discussion motivates the definition of a Team Nash Equilibrium.\footnote{We first focus on Nash Equilibrium since it is the simplest and broadest solution concept for games. We will show (in Section \ref{sec: refinement}) that in some cases, our solution satisfy some notion of sequential rationality for games among teams.}
\begin{defn}[Team Nash Equilibrium]
	A mixed strategy profile $\sigma^*=(\sigma^{*i})_{i\in\mathcal{I}}, \sigma^{*i}\in\Delta(\mathcal{M}^i)$, is said to form a Team Nash Equilibrium (TNE) if
	\begin{equation}
		J^i(\sigma^{*i}, \sigma^{*-i}) \geq J^i(\tilde{\sigma}^{i}, \sigma^{*-i})
	\end{equation}
	for any mixed strategy profile $\tilde{\sigma}^i \in\Delta(\mathcal{M}^i)$ for all $i\in\mathcal{I}$.
\end{defn}

To implement an arbitrary mixed strategy, a team can choose their strategy profile jointly before the game starts. For example, the team can jointly choose a random strategy profile out of an ensemble of profiles according to some distribution. Alternatively, they can agree on a protocol to utilize a commonly observed randomness source to randomize their strategies in a correlated manner in real time. 

The primary objective of this paper is to characterize a subclass of Team NE and to devise a backward inductive sequential computation procedure to determine these Team NE.

\subsubsection{A Motivating Example}\label{sec: example1}
The following example illustrates the importance of considering jointly randomized mixed strategies when we study games among teams. Similar to the role mixed strategies play in games among individual players, the space of jointly randomized mixed strategies contains the minimum richness of strategies that ensures an equilibrium exists in games among teams. In particular, if we restrict the teams to use independently randomized strategies, i.e. type 1 and type 2 strategies described in Section \ref{sec: rew}, then an equilibrium may not exist. This example is similar to the examples in \cite{farina2018ex,zhang2020computing,anantharam2007common} in spirit, despite the fact that in our example the players in the same team have asymmetric information.

\begin{example}[Guessing Game]\label{ex: 1}
	Consider a two-stage game (i.e. $\mathcal{T}=\{1, 2\}$) of two teams $\mathcal{I} = \{A, B\}$, each consisting of two players. The set of all agents is given by $\mathcal{N} = \{(A, 1), (A, 2), (B,1), (B, 2) \}$. Let $\bX_t^A=(X_t^{A,1}, X_t^{A,2})\in \{-1, 1\}^2$ and Team B does not have a state, i.e. $\bX_t^B = \varnothing$. Assume $\mathcal{U}_t^{i, j} = \{-1, 1\}$ for $t=1, i=A$ or $t=2, i=B$ and $\mathcal{U}_t^{i, j}=\varnothing$ otherwise, i.e. Team A moves at time $1$, and Team B moves at time $2$. At time $1$, $X_1^{A, 1}$ and $X_1^{A, 2}$ are independently uniformly distributed on $\{-1, 1\}$. Team A's system is assumed to be static, i.e. $\bX_{2}^{A} = \bX_{1}^{A}$.
	
	The rewards of Team A are given by
	\begin{align*}
	r_1^A(\bX_1, \bU_1) &= \bm{1}_{\{ X_{1}^{A, 1} U_1^{A, 1} X_{1}^{A, 2}  U_1^{A, 2}=-1\} },\\
	r_2^A(\bX_2, \bU_2) &= - \bm{1}_{\{ X_2^{A, 1} = U_2^{B, 1}\}} - \bm{1}_{\{ X_2^{A, 2} = U_2^{B, 2}\}},
	\end{align*}
	and the rewards of Team B are given by
	\begin{align*}
	r_1^B(\bX_1, \bU_1) &= 0,\\
	r_2^B(\bX_2, \bU_2) &= \bm{1}_{\{ X_2^{A, 1} = U_2^{B, 1}\}} + \bm{1}_{\{ X_2^{A, 2} = U_2^{B, 2}\}}.
	\end{align*}
	
	Assume that there are no additional common observations other than past actions, i.e. $\bY_t = \varnothing$. We set the delay $d=2$, i.e. agent (A, 1) does not know $X_t^{A,2}$ throughout the game and a similar property is true for agent (A, 2).
	In this game, the task of Team A is to choose actions according to their states at $t=1$ in order to earn a positive reward, while not revealing too much information through their actions to Team B. The task of Team B is to guess Team A's state.
	
	It can be verified (see Appendix \ref{app: ex1} for a detailed derivation) that if we restrict both teams to use independently randomized strategies (including deterministic strategies), then there exists no equilibria. However, there does exist an equilibrium where Team A randomizes in a correlated manner, specifically, the following strategy profile $\sigma^*$: At $t=1$, Team A plays $\gamma^A=(\gamma^{A, 1}, \gamma^{A, 2})$ with probability 1/2, and $\tilde{\gamma}^A=(\tilde{\gamma}^{A, 1}, \tilde{\gamma}^{A, 2})$ with probability 1/2, where 
	\begin{align*}
	\gamma^{A, 1}(x_1^{A, 1}) &= x_1^{A, 1}, \quad \gamma^{A, 2}(x_1^{A, 2}) = - x_1^{A, 2},\\
	\tilde{\gamma}^{A, 1}(x_1^{A, 1}) &= - x_1^{A, 1}, \quad \tilde{\gamma}^{A, 2}(x_1^{A, 2}) = x_1^{A, 2}
	\end{align*}
	and at $t=2$, the two members of Team B choose independent and uniformly distributed actions on $\{-1, 1\}$, independent of their action and observation history. In $\sigma^*$, each agent $(A, j)$ chooses a uniform random action irrespective of their states. It is important to have $(A, 1)$ and $(A, 2)$ choose these actions in a correlated way to ensure that they obtain the full instantaneous reward while not revealing any information.
\end{example}

\section{GAME OF COORDINATORS}\label{sec: goc}
In this section we present a game among individual players that is equivalent to the game among teams formulated in Section \ref{sec: problemformulation}.

We view the agents of a team as being coordinated by a fictitious \emph{coordinator} as in \cite{nayyar2013decentralized}: At each time $t$, team $i$'s coordinator instructs the members of team $i$ how to use their private information $H_t^{i, j}\backslash H_t^{i}$, based on $H_t^i$ and her past instructions up to time $t-1$ (see \cite{nayyar2013decentralized}).
Using this vantage point, we can view the games among teams as games among coordinators, where the coordinators' actions are the instructions, or \emph{prescriptions}, provided to individual agents. Notice that unlike agents' actions, coordinators' actions (prescriptions) cannot be publicly observed. To proceed further we formally define coordinators' actions and strategies, and prove Lemma \ref{lem: pure2coord}.

\begin{defn}[Prescription]
	Coordinator $i$'s \emph{prescriptions} at time $t$ is a collection of functions $\gamma_{t}^i=(\gamma_t^{i, j})_{(i, j)\in\mathcal{N}_i}$ where $\gamma_t^{i, j}: \mathcal{X}_{t-d+1:t}^{i,j} \mapsto \mathcal{U}_t^{i, j}$.
\end{defn}
Define $\prescription_t^{i, j}$ to be the space of functions that maps $\mathcal{X}_{t-d+1:t}^{i,j}$ to $\mathcal{U}_t^{i, j}$. Define $\prescription_t^i=\prod_{(i, j)\in\mathcal{N}_i}\prescription_t^{i, j}$.

\begin{defn}[Pure Coordination Strategy]
	Define the augmented team-common information of team $i$ to be $\overline{H}_t^i=(H_t^i, \bm{\Gamma}_{1:t-1}^i)$, where $\bm{\Gamma}_{1:t-1}^i$ are past prescriptions assigned by the coordinator of team $i$.
	A pure coordination strategy of team $i$ is a collection of mappings $\nu^i=(\nu_t^i)_{t\in\mathcal{T}}$ where $\nu_{t}^i: \overline{\mathcal{H}}_t^i \mapsto \prescription_t^i$.
\end{defn}

The next lemma establishes the equivalence between pure coordination strategies and pure strategies of a team.
\begin{lemma}\label{lem: pure2coord}
	For every pure coordination strategy profile $\nu$, there exists a pure strategy profile $\mu$ that yields the same payoffs for all teams and vice versa.
\end{lemma} 

\begin{proof}
	See Appendix \ref{app: pure2coord}.
\end{proof}

Based on the above lemma, we can immediately conclude that a mixed strategy profile is equivalent to a mixed coordination strategy (i.e. a distribution on the space of pure coordination strategy profiles). As a result, Team Nash Equilibria, as defined in Section \ref{sec: solconcept}, will be equivalent to Nash Equilibria of coordinators, where the coordinators can use mixed coordination strategies. 

Therefore, we can transform the games among teams to games among individual players, where each player is a (team) coordinator whose actions are prescriptions. Following the standard approach in game theory, we now consider behavioral strategies of the individuals (i.e. the coordinators) in this lifted game since, unlike mixed strategies, behavioral strategies allow for independent randomizations across time and therefore better facilitate a sequential decomposition of the dynamic game. 

\begin{defn}[Behavioral Coordination Strategy]
	A behavioral coordination strategy of team $i$ is a collection of mappings $g^i=(g_t^i)_{t\in\mathcal{T}}$ where $g_{t}^i: \overline{\mathcal{H}}_t^i \mapsto \Delta(\prescription_t^i)$.
\end{defn}

Given that the coordinators have perfect recall, that is, at any time $t$, the coordinator remembers all her observations up to time $t$, and all her ``actions'' (prescriptions) up to time $t-1$, we can conclude from Kuhn's theorem \cite{kuhn2016extensive} that behavioral coordination strategies are equivalent to mixed coordination strategies in the following sense.

\begin{lemma}\label{lem: bcs2ms}
	For any behavioral coordination strategy profile, there exists a mixed coordination strategy profile with the same expected payoffs and vice versa.
\end{lemma}

Based on this equivalence we can first define Nash Equilibria for the coordinator's game and then restate our objective from Section \ref{sec: solconcept}.
\begin{defn}[Coordinators' Nash Equilibrium]\label{def: cne}
	For any behavioral coordination strategy profile $g$, define
	\begin{equation}
		J^i(g) = \E^{g}\left[ \sum_{t\in\mathcal{T}}r_t^i(\bX_t, \bU_t)\right].
	\end{equation}
	
	A behavioral coordination strategy profile $g^*=$$(g_t^{*i})_{i\in\mathcal{I}, t\in\mathcal{T}}$ where $g_t^{*i}: \overline{\mathcal{H}}_t^i\mapsto \Delta(\prescription_t^i)$ is said to form a Coordinator's Nash Equilibrium (CNE) if
	\begin{equation}
	J^i(g^{*i}, g^{*-i}) \geq J^i(\tilde{g}^{i}, g^{*-i}) 
	\end{equation}
	for any behavioral coordination strategy profile $\tilde{g}^i: \overline{\mathcal{H}}_t^i\mapsto \Delta(\prescription_t^i)$ for any team $i\in\mathcal{I}$, i.e. the behavioral strategies of coordinators form a Bayes-Nash Equilibrium in the game of coordinators. 
\end{defn}

Given that we have lifted the game among teams to a game among coordinators, we adjust the terminology for the information structure accordingly. From now on, we will refer to the common information among all teams (i.e. $H_t^0$) as simply the \emph{common information}, while the information that members of team $i$ share but is not known to other teams (i.e. $\overline{H}_t^i\backslash H_t^0=(\bX_{1:t-d}^i, \bm{\Gamma}_{1:t-1}^i)$) will be referred to as the \emph{private information} of coordinator $i$. The information that is private to an agent (i.e. $X_{t-d+1:t}^{i, j}$) will be referred to as \emph{hidden information}, since none of the coordinators observe this information.

\begin{remark}
The games among coordinators we obtain has a few differences from the game model in \cite{ouyang2016dynamic}:
\begin{itemize}
	\item Actions in \cite{ouyang2016dynamic} are publicly observable. As mentioned before, in our game among coordinators, the ``actions'' (prescriptions) of the coordinators are private information.
	
	\item The local state $X_t^i$ in \cite{ouyang2016dynamic} is perfectly observable by player $i$ without delay. In our game among coordinators, at time $t$, a coordinator can only observe her local state up to time $t-d$.
	
	\item The transitions of local states in \cite{ouyang2016dynamic} are conditionally independent given the actions, i.e. $\Pr(x_{t+1}| x_t, u_t) = \prod_{i} \Pr(x_{t+1}^i| x_t^i, u_t)$. In our game among coordinators, transition of local states are not independent given the prescriptions.
	
	\item The public observation process of local states in \cite{ouyang2016dynamic} is conditionally independent given the actions, i.e. $\Pr(y_{t}| x_t, u_t) = \prod_{i} \Pr(y_{t}^i| x_t^i, u_t)$. In our game among coordinators,  public observations of local states are not independent given the prescriptions and local states.
\end{itemize}

Due to the above differences, we cannot directly apply the results of \cite{ouyang2016dynamic} to the game of coordinators.
\end{remark}

\subsection{An Illustrative Example}
The following example illustrates how to visualize games among teams from the coordinators' viewpoint. 

\begin{example}\label{ex: 2}
	Consider a variant of the Guessing Game in Example \ref{ex: 1} with the same system model and information structure but different action sets and reward functions. In the new game, Team A moves at both $t=1$ and $t=2$, with $\mathcal{U}_t^{A, j} = \{-1, 1\}$ for $t=1,2$ and $j=1,2$. Team B moves only at time $t=2$ as in the original game. The new reward functions are given by
	\begin{align*}
	r_1^A(\bX_1, \bU_1) &= 0,\\
	r_2^A(\bX_2, \bU_2) &= \bm{1}_{\{X_2^{A, 2} = U_2^{A, 1}, X_2^{A, 1} = U_2^{A, 2}\}} + \bm{1}_{\{ \bX_2^{A} \neq \bU_2^{B}\}},\\
	r_1^B(\bX_1, \bU_1) &= 0,\\
	r_2^B(\bX_2, \bU_2) &= \bm{1}_{\{ \bX_2^{A} = \bU_2^{B}\}}.
	\end{align*}
	In this example, Team A's task is to guess its own state after a round of publicly observable communication while not leaking information to Team B. 
	
	A Team Nash Equilibrium $(\sigma^{*A}, \sigma^{*B})$ of this game is as follows: Team A chooses one of the four pure strategy profiles listed below with equal probability:
	\begin{align*}
		\bullet~&\mu_1^{A, 1}(x_1^{A, 1}) = - x_1^{A, 1}, \mu_1^{A, 2}(x_1^{A, 2}) = x_1^{A, 2},\\
		&\mu_2^{A, 1}(\mathbf{u}_1, x_{1:2}^{A, 1}) = u_1^{A, 2}, \mu_2^{A, 2}(\mathbf{u}_1, x_{1:2}^{A, 2}) = -u_1^{A, 1};\\
		\bullet~&\mu_1^{A, 1}(x_1^{A, 1}) = -  x_1^{A, 1}, \mu_1^{A, 2}(x_1^{A, 2}) = - x_1^{A, 2},\\ 
		&\mu_2^{A, 1}(\mathbf{u}_1, x_{1:2}^{A, 1}) = -u_1^{A, 2}, \mu_2^{A, 2}(\mathbf{u}_1, x_{1:2}^{A, 2}) =  -u_1^{A, 1};\\
		\bullet~&\mu_1^{A, 1}(x_1^{A, 1}) = x_1^{A, 1}, \mu_1^{A, 2}(x_1^{A, 2}) = x_1^{A, 2},\\
		&\mu_2^{A, 1}(\mathbf{u}_1, x_{1:2}^{A, 1}) = u_1^{A, 2} , \mu_2^{A, 2}(\mathbf{u}_1, x_{1:2}^{A, 2}) = u_1^{A, 1};\\
		\bullet~&\mu_1^{A, 1}(x_1^{A, 1}) = x_1^{A, 1}, \mu_1^{A, 2}(x_1^{A, 2}) = -x_1^{A, 2},\\ 
		&\mu_2^{A, 1}(\mathbf{u}_1, x_{1:2}^{A, 1}) =  -u_1^{A, 2}, \mu_2^{A, 2}(\mathbf{u}_1, x_{1:2}^{A, 2}) = u_1^{A, 1};
	\end{align*}
	while Team B choose $\bU_2^{B}$ uniformly at random independent of $\bU_1$. In words, from Team B's point of view, Team A chooses $\bU_1^{A}$ to be a uniform random vector independent of $\bX_1^A$. However the randomization is done in a coordinated manner: Before the game starts, both members of team A randomly draw a card from two cards, where one card says ``lie'' and the other says ``tell the truth.'' Both players then tell each other what card they have drawn before the game starts. At time $t=1$, both players in Team A play the strategy indicated by their cards. At time $t=2$, Team A can then perfectly recover $\bX_1^A$ from $\bU_1^A$ and the knowledge about the strategy being used at $t=1$. 
	
  	Now we describe Team A's equilibrium strategy by the equivalent coordinator A's behavioral strategy. Use $\mathbf{ng}$ to denote the prescription that maps $-1$ to $1$ and $1$ to $-1$. Use $\mathbf{id}$ to denote the identity map prescription, i.e. the prescription that maps $-1$ to $-1$ and $1$ to $1$. Use $\mathbf{cp}_{b}$ to denote the constant prescription that always instruct individuals to play $b\in \{-1, 1\}$.
  	The mixed strategy profile $\sigma^{*A}$ is equivalent to the following behavioral coordination strategy: At time $t=1$, $g_1^A(\varnothing)\in \Delta(\prescription_1^{A, 1} \times \prescription_1^{A, 2} )$ satisfies 
  	\begin{align*}
  		&g_1^A(\varnothing)(\gamma_1^{A, 1}, \gamma_1^{A, 2}) = \frac{1}{4}\qquad \forall \gamma_1^{A, 1}, \gamma_1^{A, 2} \in \{\mathbf{ng}, \mathbf{id} \}.
  	\end{align*}
  	At time $t=2$, $g_2^A: \mathcal{U}_1^{A, 1}\times \mathcal{U}_1^{A, 2}\times \prescription_1^{A, 1}\times \prescription_1^{A, 2} \mapsto \Delta(\prescription_2^{A, 1} \times \prescription_2^{A, 2})$ is a deterministic strategy that satisfies
  	\begin{align*}
  		&g_2^A(u^1, u^2, \mathbf{ng}, \mathbf{id} ) = \textsc{dm}(\mathbf{cp}_{u^2}, \mathbf{cp}_{-u^1}),\\
  		&g_2^A(u^1, u^2, \mathbf{ng}, \mathbf{ng} ) = \textsc{dm}(\mathbf{cp}_{-u^2}, \mathbf{cp}_{-u^1}),\\
  		&g_2^A(u^1, u^2, \mathbf{id}, \mathbf{id} ) = \textsc{dm}(\mathbf{cp}_{u^2}, \mathbf{cp}_{u^1}),\\
  		&g_2^A(u^1, u^2, \mathbf{id}, \mathbf{ng} ) = \textsc{dm}(\mathbf{cp}_{-u^2}, \mathbf{cp}_{u^1}),
  	\end{align*}
  	where $\textsc{dm}: \prescription_2^{A, 1}\times \prescription_2^{A, 2} \mapsto \Delta(\prescription_2^{A, 1}\times \prescription_2^{A, 2})$ represents the delta measure. In words, the coordinator of Team A randomly chooses one of all four possible prescription profiles at time $t=1$. At time $t=2$, based on the observed action and the prescriptions chosen before, the coordinator of Team A directly assign actions to agents to instruct them to recover the state from the actions at $t=1$.
  	Note that the behavioral coordination strategy at $t=2$ depends explicitly on the past prescription $\bm{\Gamma}_1^A$ in addition to the realization of past actions. This is because the coordinator needs to remember not only the agents' actions, but also the rationale behind those actions in order to interpret the signals sent through the actions.
\end{example}

\section{COMPRESSION OF PRIVATE INFORMATION}\label{sec: spi}
In this section, we identify a subset of a coordinator's private information that is sufficient for decision-making for the game of coordinators formulated in Section \ref{sec: goc}. We refer to this subset of private information as the Sufficient Private Information (SPI) for this coordinator. We restrict attention to Sufficient Private Information Based (SPIB) strategies, where coordinators choose prescriptions based on their sufficient private information along with the common information. As a result, the coordinators do not need full recall to play SPIB strategies. We show that there always exist a Coordinator's Nash Equilibrium where coordinators play SPIB strategies. As a result, the restriction to SPIB strategies does not hurt the existence of equilibria.

We proceed as follows. We first present a structural result that plays an important role in the subsequent analysis. We then introduce our results in two steps in separate sections based on the value of $d$, the delay in information sharing within the same team. We treat the cases $d=1$ and $d>1$ separately. This is since when $d=1$ the equilibrium strategies we obtain are simpler than those under $d>1$. For $d>1$, we introduce the notion of Partially Realized Prescriptions (PRP) and use them to construct a subset of private information that is sufficient for decision-making. We then define the notion of Sufficient Private Information (SPI) and Sufficient Private Information Based (SPIB) strategies to unify the results for $d=1$ and $d>1$. Finally, we show that CNEs where coordinators play SPIB strategies always exist.

\subsection{A Preliminary Result}
We show that the states and prescriptions of different coordinators are conditionally independent given the common information.

\begin{lemma}[Conditional Independence]\label{lem: condindep}
	Under any behavioral coordination strategy profile $g$ and for each time $t\in\mathcal{T}$, $(\bX_{1:t}^k, \bm{\Gamma}_{1:t}^k)_{k\in\mathcal{I}}$ are conditionally independent given the common information $H_t^0$. Furthermore, the conditional distribution of $(\bX_{1:t}^k, \bm{\Gamma}_{1:t}^k)$ depends on $g$ only through $g^k$.
\end{lemma}

\begin{proof}
	See Appendix \ref{app: condindep}.
\end{proof}

As a result of Lemma \ref{lem: condindep}, coordinator $i$'s estimation of other coordinators' state and prescriptions is independent of her own strategy and private information. In other words, while coordinator $i$ has access to both the common information and her private information, her belief on the other coordinators' private information (history of states and prescription) is solely based on the common information. 

\subsection{Result for $d=1$}
While coordinator $i$'s private information consists of $(\bX_{1:t-1}^i, \bm{\Gamma}_{1:t-1}^i)$, she does not have to use all of it to form a best response.

\begin{lemma}\label{lem: suffprivate}
	Under $d=1$, for any behavioral coordination strategy profile $g^{-i}$ of all coordinators other than $i$, there exists a best response behavioral coordination strategy $g^i$ for coordinator $i$ that chooses randomized prescriptions based solely on $(H_t^0, \bX_{t-1}^i)$.
\end{lemma}

\begin{proof}
	Deferred to the proof of Lemma \ref{lem: suffprivated3}.
\end{proof}

Lemma \ref{lem: suffprivate} shows that the coordinators can ignore much of their private information without compromising their objective.

\subsection{Result for $d>1$}
We now identify a compressed version of private information for $d>1$ case that is sufficient for decision-making.

Recall that coordinator $i$'s information at time $t$ consists of $\overline{H}_t^i=(\bY_{1:t-1}, \bU_{1:t-1}, \bX_{1:t-d}^i, \bm{\Gamma}_{1:t-1}^i)$. To choose her prescriptions at time $t$, coordinator $i$ needs to estimate her hidden information (i.e. $\bX_{t-d+1:t}^i$). When $d=1$, the belief on hidden information is simply constructed using $(\bX_{t-1}^i, \bU_{t-1})$ and the knowledge of the transition probabilities of the underlying system. However, when $d>1$, more information in addition to $(\bX_{t-d}^i, \bU_{t-d:t-1})$ is needed to form the belief.

To illustrate this, we start with the case $d=2$. When $d=2$, the belief of coordinator $i$ on her hidden information would depend on the last prescription $\bm{\Gamma}_{t-1}^i$ in addition to $(\bX_{t-2}^i, \bU_{t-2:t-1})$. This is due to the signaling effect of the action $\bU_{t-1}^i$: since coordinator $i$ knows $\bU_{t-1}^i$, she can infer something about $\bX_{t-1}^i$ through the prescription used to produce these actions (recall that $U_{t-1}^{i, j}=\Gamma_{t-1}^{i, j}(X_{t-2:t-1}^{i, j})$ for $(i, j)\in\mathcal{N}_i$). Hence at time $t$, coordinator $i$ needs to take $\bm{\Gamma}_{t-1}^i$ into account when forming her belief on the hidden information. 

Furthermore, for $d=2$, when making a decision at time $t$, coordinator $i$ can use a compressed version of the prescription $\bm{\Gamma}_{t-1}^i$ instead of $\bm{\Gamma}_{t-1}^i$ itself. This is because at time $t$, coordinator $i$ has learned $\bX_{t-2}^i$ that she didn't know at time $t-1$. The coordinator can then focus on the following essential question: given the knowledge of $\bX_{t-2}^i$, what is the relationship between $\bX_{t-1}^i$ and $\bU_{t-1}^i$?

Similarly, for a general $d>1$, to estimate the hidden information, each coordinator needs to utilize her past $(d-1)$ prescriptions. Again, a coordinator can use a compressed version of the past $(d-1)$ prescriptions, since she can incorporate the additional information she knows at time $t$ that she did not know back when the prescriptions were chosen. Each coordinator can now focus on the relationship between the unknown states and the known actions, given what is already known. This motivates the definition of $(d-1)$-step \emph{partially realized prescriptions} PRPs.

\begin{defn}
	The $(d-1)$-step \emph{partially realized prescriptions}\footnote{The $(d-1)$-step PRPs are the same as the partial functions defined in the second structural result in \cite{nayyar2010optimal}. } (PRPs) for coordinator $i$ at time $t$ is a collection of functions $\bm{\Phi}_{t}^i:=(\Phi_{t-l, l}^{i,j})_{\substack{ (i, j)\in\mathcal{N}_i, 1\leq l\leq d-1 } }$, where 
	\begin{equation}
	\Phi_{t-l, l}^{i,j} = \Gamma_{t-l}^{i, j}(X_{t-l-d+1:t-d}^{i, j}, \cdot)
	\end{equation}
	is a function from $\mathcal{X}_{t-d+1:t-l}^{i, j}$ to $\mathcal{U}_{t-l}^{i, j}$.
\end{defn}

PRPs have smaller dimension than prescriptions. To illustrate this point, consider the case where $d=2$: A prescription $\gamma_{t-1}^{i, j}$ can be represented as a table, where the rows represent $x_{t-2}^{i, j}\in \mathcal{X}_{t-2}^{i, j}$, the columns represent $x_{t-1}^{i, j}\in \mathcal{X}_{t-1}^{i, j}$, and the entries represent the corresponding action $u_{t-1}^{i, j}=\gamma_{t-1}^{i, j}(x_{t-2:t-1}^{i, j})$ to take. On the other hand, the 1-step partially realized prescription $\phi_{t}^{i, j}= \gamma_{t-1}^{i, j}(x_{t-2}^{i, j}, \cdot)$ can be represented by one row of the table of $\gamma_{t-1}^{i, j}$ chosen based on the realization of $X_{t-2}^{i, j}$.

In addition to $(\bX_{t-d}^i, \bU_{t-d:t-1}, \bm{\Phi}_t^i)$, coordinator $i$ also needs to use $Y_{t-d+1:t-1}^i$ to form a belief on her hidden information since $Y_{t-d+1:t-1}^i$ can provide additional insight on $\bX_{t-d+1:t-1}^i$ that $(\bX_{t-d}^i, \bU_{t-d:t-1}, \bm{\Phi}_t^i)$ cannot necessarily provide. The belief coordinator $i$ has on her hidden information is summarized in the following lemma.

\begin{lemma}\label{lem: selfbelief}
	Suppose that the behavioral coordination strategy profile $g=(g^i)_{i\in\mathcal{I}}$ is being played. Then the conditional distribution of $\bX_{t-d+1:t}^i$ given $\overline{H}_t^i$ under $g$ can be expressed as a fixed function of $(Y_{t-d+1:t-1}^i, \bU_{t-d:t-1}, \bX_{t-d}^i, \bm{\Phi}_t^i)$, i.e.
	\begin{equation}
	\begin{split}
	&\quad~\Pr^g(x_{t-d+1:t}^i|\overline{h}_t^i)~~\\
	&= P_t^i(x_{t-d+1:t}^i|y_{t-d+1:t-1}^i, u_{t-d:t-1}, x_{t-d}^i, \phi_{t}^i)~\forall \overline{h}_t^i\in\overline{\mathcal{H}}_t^i~~
	\end{split}\label{eq: belieffunction}
	\end{equation}
	for some function $P_t^i$ that does not depend on $g$.
\end{lemma}

\begin{proof}
	See Appendix \ref{app: selfbelief}.
\end{proof}

\begin{remark}
	The above result can be interpreted in the following way: $\bX_{t-d}^i$ is perfectly observed, hence coordinator $i$ can discard $\bX_{1:t-d-1}^i$ which are irrelevant information due to the Markov property. Since $\bX_{t-d+1:t-1}^i$ are not perfectly observed by coordinator $i$, every public observation and action based upon $\bX_{t-d+1:t-1}^i$ are important to coordinator $i$ since it can help in estimating the state $\bX_{t-d+1:t-1}^i$. Note that $\bm{\Phi}_{t}^i$ encodes the essential information coordinator $i$ needs to remember at time $t$ about her previous signaling strategy: how does $\bX_{t-d+1:t-1}^i$ (unknown) map to $\bU_{t-d+1:t-1}^i$ (known)? With this piece of information, coordinator $i$ can fully interpret the signals sent through $\bU_{t-d+1:t-1}^i$.
\end{remark}

We claim that while coordinator $i$'s private information consists of $(\bX_{1:t-d}^i, \bm{\Gamma}_{1:t-1}^i)$, she only needs to use $(\bX_{t-d}^i, \bm{\Phi}_{t}^i)$ along with the common information to choose prescriptions.

\begin{lemma}\label{lem: suffprivated3}
	Given an arbitrary $d>1$, for any behavioral coordination strategy profile $g^{-i}$ of all coordinators other than $i$, there exists a best response behavioral coordination strategy $g^i$ for coordinator $i$ that chooses randomized prescriptions based solely on $(H_t^0, \bX_{t-d}^i, \bm{\Phi}_{t}^i)$.
\end{lemma}

\begin{proof}
	See Appendix \ref{app: suffprivated3}.
\end{proof}

\begin{remark}
	Lemmas \ref{lem: selfbelief} and \ref{lem: suffprivated3} and their proofs also apply to $d=1$, in which case the $(d-1)$-step PRP $\bm{\Phi}_{t}^i$ is empty by definition. 
\end{remark}

From now on, we unify the results for $d=1$ and $d>1$. We formally define the Sufficient Private Information (SPI) and SPIB strategies which will be used in the rest of the paper.

\begin{defn}[Sufficient Private Information]
	For a given $d>0$, the \emph{Sufficient Private Information} (SPI) for coordinator $i$ at time $t$ is defined as $S_t^i = (\bX_{t-d}^i, \bm{\Phi}_{t}^i)$.
\end{defn}

\begin{defn}[Sufficient Private Information Based Strategy]
	A \emph{Sufficient Private Information Based} (SPIB) strategy for coordinator $i$ is a collection of functions $\rho^i = (\rho_t^i)_{t\in\mathcal{T}}, \rho_t^i: \mathcal{H}_t^0 \times \mathcal{S}_t^i \mapsto \Delta(\prescription_{t}^i)$.
\end{defn}

It can be easily verified that $S_t^i$ can be sequentially updated, i.e., there exists a fixed, strategy-independent function $\iota_t^i$ such that
\begin{equation}\label{eq: spiupdate}
S_{t+1}^i = \iota_t^i(S_t^i, \bX_{t-d+1}^i, \bm{\Gamma}_t^i).
\end{equation}

Therefore, a coordinator does not need full recall to play a SPIB strategy.

\subsection{Coordinators' Nash Equilibrium in SPIB Strategies and its Existence}
Since the coordinators have perfect recall, we know from standard results for dynamic games that a CNE, as defined in Definition \ref{def: cne}, exists (see Chapter 11 of \cite{osborne1994course}, for example). However, in those CNEs, coordinators do not necessarily play SPIB strategies, hence the standard arguments that guarantee the existence of CNE cannot be used to establish the existence of CNE in SPIB strategies. Moreover, SPIB strategies do not feature full recall, hence one cannot directly apply standard arguments to establish the existence of CNE in SPIB strategies.

An SPIB strategy profile $\rho=(\rho_t^i)_{i\in\mathcal{I}, t\in\mathcal{T}}, \rho_t^i: \mathcal{H}_t^0 \times \mathcal{S}_t^i \mapsto \Delta(\prescription_{t}^i)$ is called a \emph{Sufficient Private Information Based Coordinators' Nash Equilibrium} (SPIB-CNE) if $\rho$, seen as a profile of behavioral coordination strategies, forms a Coordinator's Nash Equilibrium.

\begin{thm}\label{thm: spibexist}
	There exists at least one SPIB-CNE for the dynamic game among coordinators.
\end{thm}
\begin{proof}
	See Appendix \ref{app: spibexist}.
\end{proof}

\section{COMPRESSION OF COMMON INFORMATION AND SEQUENTIAL DECOMPOSITION}\label{sec: 4}
The SPIB strategies defined in the previous section use sufficient private information instead of the entire private information for each coordinator. If the sets $\mathcal{X}_t, \mathcal{Y}_t, \mathcal{U}_t$ are time-invariant, the set of possible values of sufficient private information used in SPIB strategies is also time-invariant. However, the common information still increases with time and this means that the domain of SPIB strategies keeps increasing with time. In order to limit the growing domain of SPIB strategies, we introduce a subclass of SPIB strategies, named Compressed Information Based (CIB) strategies, where the coordinators use a compressed version of common information instead of the entire common information. We show that this new class of strategies satisfies a key best-response/closedness property. Based on this property we provide a backward inductive procedure that identifies an equilibrium in this subclass of strategies if each step of this procedure has a solution. While equilibria in CIB strategies may not exist in general (see example in Section \ref{sec: existence}), we identify classes of games among teams where such equilibria do exist.


\subsection{Compressed Common Information and CIB Strategy}\label{sec: ccicibd1}
In decentralized control problems \cite{nayyar2013decentralized,tavafoghi2018unified} and games among individuals \cite{ouyang2016dynamic,tavafoghi2016stochastic}, agents can compress their common information into beliefs on hidden and (sufficient) private information for the purpose of decision-making. Similarly, we would like to consider a subclass of SPIB strategies where each coordinator compresses the common information $H_t^0$ to a belief on sufficient private information and hidden information, i.e. $\Pr(\bX_{t-d:t}^k=\cdot, \bm{\Phi}_{t}^k=\cdot|H_t^0)$ for $k\in\mathcal{I}$. Due to Lemma \ref{lem: selfbelief}, these beliefs can be constructed from $\Pr(\bX_{t-d}^k=\cdot, \bm{\Phi}_{t}^k=\cdot|H_t^0)$ and $(Y_{t-d+1:t-1}^k, \bU_{t-d:t-1})$. Therefore, we will consider strategies where coordinators use common information based beliefs on the sufficient private information $S_t^k=(\bX_{t-d}^k, \bm{\Phi}_t^k)_{k\in\mathcal{I}}$ along with the uncompressed values of $(\bY_{t-d+1:t-1}, \bU_{t-d:t-1})$, instead of the whole $H_t^0$.

We formalize the above discussion in the rest of this subsection.

\begin{defn}[Belief Generation System]\label{def: bgs2}
	A \emph{Belief Generation System} for coordinator $i$ consists of a sequence of functions $\psi^i=(\psi_t^{i, k})_{k\in\mathcal{I}, t\in\mathcal{T}}$ where $\psi_t^{i, k}: \left(\prod_{l\in \mathcal{I}}\Delta(\mathcal{S}_{t}^l)\right) \times \mathcal{Y}_{t-d+1:t} \times \mathcal{U}_{t-d:t}\mapsto \Delta(\mathcal{S}_{t+1}^k)$
\end{defn}

Coordinator $i$ can use this system to generate common information based beliefs $\Pi_t^{i, k}\in \Delta(\mathcal{S}_{t}^k)$ for all $k\in\mathcal{I}$ as follows:
\begin{itemize}
	\item $\Pi_1^{i, k}$ is the prior distribution of $(\bX_{-(d-1)}^k, \bm{\Phi}_1^k)$, i.e. a measure which assigns probability 1 to the event $(\bX_{-(d-1)}^k=0, \bm{\Phi}_1^k=\hat{\phi}_1^k)$, where $\hat{\phi}_1^k$ is the PRP that always produces actions $u_{t}^{k, j}=0$ for all $(k, j)\in\mathcal{N}_k, t\leq 0$ (see Remark \ref{remark: initial});
	\item $\Pi_{t+1}^{i, k} = \psi_t^{i, k}((\Pi_t^{i, l})_{l\in\mathcal{I}}, \bY_{t-d+1:t}, \bU_{t-d:t}), t\geq 1$.
\end{itemize}
$\Pi_t^{i, k}$ represents coordinator $i$'s subjective belief on coordinator $k$'s sufficient private information $S_t^k$. These beliefs along with $(\bY_{t-d+1:t-1}, \bU_{t-d:t-1})$ will serve as coordinator $i$'s compressed common information.

\begin{defn}[Compressed Common Information]\label{def: cci2}
	We define coordinator $i$'s \emph{Compressed Common Information} (CCI) at time $t$ as
	\begin{equation}
	B_t^i=\left(\left(\Pi_t^{i, l}\right)_{l\in\mathcal{I}}, \bY_{t-d+1:t-1}, \bU_{t-d:t-1}\right),
	\end{equation}
	where $(\Pi_t^{i, l})_{l\in\mathcal{I}}$ are generated using the belief generation system defined in Definition \ref{def: bgs2}. Note that when $d=1$, we have $B_t^i=((\Pi_t^{i, l})_{l\in\mathcal{I}}, \bU_{t-1})$.
\end{defn}

We can write the belief update using $B_t^i$ as $\Pi_{t+1}^{i, k} = \psi_t^{i, k}(B_t^i, \bY_t, \bU_{t})$. With a slight abuse of notation, we use $\psi_t^i$ to represent the collection $(\psi_t^{i, k})_{k\in\mathcal{I}}$ and write the belief updates collectively as $(\Pi_{t+1}^{i, l})_{l\in\mathcal{I}} = \psi_t^{i}(B_t^i, \bY_t, \bU_{t})$.

We now define a subclass of strategies where coordinator $i$ uses her CCI instead of the entire common information.

\begin{defn}[Compressed Information Based Strategy]\label{def: cibstrategy}
	Let $\mathcal{B}_t=\left(\prod_{k\in\mathcal{I}}\Delta(\mathcal{S}_{t}^k)\right) \times \mathcal{Y}_{t-d+1:t-1} \times \mathcal{U}_{t-d:t-1}$.
	A \emph{Compressed Information Based} (CIB) strategy for coordinator $i$ is a pair $(\lambda^i, \psi^i)$, where $\lambda^i=(\lambda_t^i)_{t\in\mathcal{T} }$ is a collection of functions 
	$\lambda_t^i: \mathcal{B}_t\times \mathcal{S}_{t}^i \mapsto \Delta(\prescription_t^i)$, and $\psi^i=(\psi_t^{i, k})_{k\in\mathcal{I}, t\in\mathcal{T} }$, 
	$\psi_t^{i, k}: \mathcal{B}_t\times \mathcal{Y}_t\times \mathcal{U}_t \mapsto \Delta(\mathcal{S}_{t+1}^k)$ is a belief generation system as defined in Definition \ref{def: bgs2}. 
\end{defn}

Under a CIB strategy, coordinator $i$ uses her belief generation system to compress common information into beliefs and then uses these beliefs along with $(\bY_{t-d+1:t-1}, \bU_{t-d:t-1}, S_{t}^i)$ to select a randomized prescription. Thus, a CIB strategy $(\lambda^i, \psi^i)$ is equivalent to an SPIB-strategy
\begin{equation}
\begin{split}
\rho_t^i(h_t^0, s_t^i) = \lambda_t^i\left(\left(\pi_t^{i, k}\right)_{k\in\mathcal{I}}, y_{t-d+1:t-1}, u_{t-d:t-1}, s_{t}^i\right)&\\
\forall h_t^0\in \mathcal{H}_t^0, \forall s_t^i\in \mathcal{S}_t^i&
\end{split}
\end{equation}
where $(\pi_t^{i, k})_{k\in\mathcal{I}}$ is generated from $h_t^0$ through the belief generation system defined in Definition \ref{def: bgs2}.

\begin{remark}
	One advantage of CIB strategies is that at each time coordinator $i$ only needs to use her current CCI rather than the time-increasing full common information (i.e. $H_t^0$). Thus, if the sets $\mathcal{X}_t, \mathcal{Y}_t, \mathcal{U}_t$ are time-invariant, the mappings $\lambda_t^i, \psi_t^i$ in a CIB strategy have a time-invariant domain.
\end{remark}

\begin{remark}
	We have not imposed any restriction on the mapping $\psi_t^i$ in coordinator $i$'s belief generation system (see Definition \ref{def: bgs2}). Intuitively, however, one can imagine that coordinator $i$ has some prediction about others' strategies and is rationally using her prediction about others' strategies to update her beliefs through the mapping $\psi_t^i$. In the following discussion, our focus will be on such ``rational'' $\psi_t^i$ where the notion of rationality will be captured by Bayes' rule.
\end{remark}

Coordinator $i$'s belief generated from $\psi^i$ can be grouped into two parts: $(\Pi_t^{i, -i})_{t\in\mathcal{T}}$ and $(\Pi_t^{i, i})_{t\in\mathcal{T}}$. The first part represents what coordinator $i$ believes about other coordinators' SPI. The second part represents what coordinator $i$ thinks is the other coordinators' belief on her own SPI. 

\subsection{Consistency and Closedness of CIB Strategies}
As mentioned before, our interest in CIB strategies is motivated by the common information belief based strategies that appeared in the solution of decentralized control problems \cite{nayyar2013decentralized,tavafoghi2018unified} or games among individuals \cite{nayyar2013common,ouyang2016dynamic}. The common beliefs used in these prior works are compatible with Bayes' rule (i.e. the beliefs can be obtained using Bayes' rule along with the knowledge of the system model and the strategies being used). Inspired by these observations, we are particularly interested in CIB strategies where the belief generation system is compatible with Bayes' rule, i.e. the beliefs generated by coordinator $i$ using $\psi^i$ agree with those generated using Bayes' rule along with the knowledge of the system model and the strategies being used. 

In the following discussion, we identify a key property of such Bayes' rule compatible CIB strategies. To do so, we use the following technical definition.

\begin{defn}[Consistency]\label{def: consistency}
	Given $\lambda_t^i: \mathcal{B}_t\times \mathcal{S}_{t}^i \mapsto \Delta(\prescription_t^i)$, a belief generation function $\psi_t^{*, i}: \mathcal{B}_t\times \mathcal{Y}_t\times \mathcal{U}_t\mapsto \Delta(\mathcal{S}_{t+1}^i)$ is said to be \emph{consistent} with $\lambda_t^i$ if the following holds: For all $b_t=((\pi_t^l)_{l\in\mathcal{I}}, y_{t-d+1:t-1}, u_{t-d:t-1}) \in\mathcal{B}_t$, $\psi_t^{*, i}(b_t, y_t, u_t)$ is equal to the conditional distribution of $S_{t+1}^i$ given the event $(\bY_t=y_t, \bU_t=u_t)$ found using Bayes rule (whenever Bayes rule applies), assuming that $y_{t-d+1:t-1}$ and $u_{t-d:t-1}$ are the realization of recent observations and actions, $S_t^i$ has prior distribution $\pi_t^{i}$, and given $S_{t}^i=s_{t}^i$, $\bm{\Gamma}_t^i$ has distribution $\lambda_t^i(b_t, s_{t}^i)$. That is,
	\begin{equation}\label{beliefupdate}
	[\psi_t^{*, i}(b_t, y_t, u_t)](s_{t+1}^i)= \dfrac{\Upsilon_t^{ i}(b_t, y_{t}^i, u_{t}, s_{t+1}^i)}{\sum_{\tilde{s}_{t+1}^i} \Upsilon_t^i(b_t, y_{t}^i, u_{t}, \tilde{s}_{t+1}^i)}
	\end{equation}
	whenever the denominator of \eqref{beliefupdate} is non-zero, where
	\begin{align*}
	&~\quad\Upsilon_t^i(b_t, y_{t}^i, u_{t}, s_{t+1}^i )\\&:=\sum_{\tilde{s}_t^i}\sum_{\tilde{x}_{t-d+1:t}^i }\sum_{ \tilde{\gamma}_t^i: \tilde{\gamma}_t^i(\tilde{x}_{t-d+1:t}^i) = u_t^i } \Big[\Pr(y_t^i|\tilde{x}_{t}^i, u_{t}) \times \\
	&\times \bm{1}_{ \{s_{t+1}^i = \iota_t^i(\tilde{s}_t^i, \tilde{x}_{t-d+1}^i, \tilde{\gamma}_t^i) \} } \lambda_{t}^i(\tilde{\gamma}_{t}^i|b_{t}, \tilde{s}_{t}^i) \times\\&\times P_t^i(\tilde{x}_{t-d+1:t}^i| y_{t-d+1:t-1}^i, u_{t-d:t-1}, \tilde{s}_{t}^i) \pi_t^{i}(\tilde{s}_{t}^i)\Big]
	\end{align*}
	for all
	\begin{align*}
	&b_t=((\pi_t^l)_{l\in\mathcal{I}}, y_{t-d+1:t-1}, u_{t-d:t-1}) \in\mathcal{B}_t, y_t^i\in \mathcal{Y}_t^i,\\
	&u_t\in\mathcal{U}_t, s_{t+1}^i\in \mathcal{S}_{t+1}^i,
	\end{align*}
	$\iota_t^i$ is defined in \eqref{eq: spiupdate} and $P_t^i$ is as described in Lemma \ref{lem: selfbelief}. 
	
	For any index set $\varOmega\subset \mathcal{I}\times \mathcal{T}$ We say that $\psi^{*, i}=(\psi_t^{*, i})_{(i, t)\in \varOmega}$ is consistent with $\lambda^i=(\lambda_t^i)_{(i, t)\in \varOmega}$ if $\psi_t^{*, i}$ is consistent with $\lambda_t^i$ for all $(i, t)\in \varOmega$.
\end{defn}

A CIB strategy $(\lambda^i, \psi^i)$ for coordinator $i$ is said to be self-consistent if $\psi^{i, i}$ is consistent with $\lambda^i$. Since self-consistency can be viewed as Bayes' rule compatibility, the beliefs $(\Pi_t^{i, i})_{t\in\mathcal{T}}$ represents true conditional distributions of coordinator $i$'s SPI given the common information under a self-consistent strategy.
\begin{lemma}\label{lem: piistruebelief0}
	Let $(\lambda^i, \psi^i)$ be a self-consistent CIB strategy of coordinator $i$.
	Denote the behavioral strategy generated from $(\lambda^{i}, \psi^i)$ as $g^{i}$. Let $h_t^0\in\mathcal{H}_t^0$ be admissible under $g_{1:t-1}^{i}$, then 
	\begin{align}
	&\quad~\Pr^{g_{1:t-1}^{i}}(s_{t}^i, x_{t-d+1:t}^i|h_t^0)\\
	&=\pi_{t}^{i, i}(s_{t}^i) P_t^i(x_{t-d+1:t}^i|y_{t-d+1:t-1}^i, u_{t-d:t-1}, s_t^i) \\
	&\qquad\qquad  \forall s_{t}^i\in\mathcal{S}_{t}^i~\forall x_{t-d+1:t}^i\in \mathcal{X}_{t-d+1:t}^i
	\end{align}
	where $\pi_t^{i, i}$ is the belief obtained using $\psi^i$ under the realization $h_t^0$ of common information and $P_t^i$ is as described in Lemma \ref{lem: selfbelief}.
\end{lemma}

\begin{proof}
	See Appendix \ref{app: piistruebelief0}.
\end{proof}


Now, consider a game with two coordinators: Suppose that coordinator 1 plays a self-consistent CIB strategy with belief generation system $\psi^1$. Since the belief $\Pi_t^{1, 1}$ generated from $\psi^1$ is a true conditional distribution on coordinator 1's SPI, coordinator 2 can use $\Pi_t^{1, 1}$ as her belief on coordinator 1's SPI. Further, coordinator 2 can use $\psi^1$ to compute coordinator 1's belief about coordinator 2's SPI. This suggests that coordinator 2 should mimic coordinator 1's belief generation system when coordinator 1's strategy is self-consistent. This observation, along with results from Markov decision theory, lead to the following crucial best-response property of CIB strategies.

\begin{lemma}[Closedness of CIB strategies]\label{lem: closenessd1}
	Suppose that all coordinators other than coordinator $i$ are using self-consistent CIB strategies. Let $(\lambda^k, \psi^k)_{k\in \mathcal{I}\backslash\{i \}  }$ be the CIB strategy profile of coordinators other than $i$. Suppose that $\psi^j=\psi^k$ for all $j, k\in\mathcal{I}\backslash\{i\}$.	
	Then, a best-response strategy for coordinator $i$ is a CIB strategy with the same belief generation system as the other coordinators.
\end{lemma}

\begin{proof}
	See Appendix \ref{app: closenessd1}.
\end{proof}

\subsection{Interpretation and Discussion of Consistency and Closedness Property}
Lemma \ref{lem: closenessd1} imposes two conditions on the CIB strategies of coordinators other than $i$, namely (I) they are self-consistent, and (II) their belief generation systems are identical. In order to illustrate the significance of both conditions, we first describe how coordinator $i$ could form her best response when all coordinators other than $i$ are playing some generic CIB strategies that are not necessarily self-consistent or having identical belief generation system.

The problem of finding coordinator $i$'s best response to others' CIB strategies can be thought of as a stochastic control problem with partial observation. This suggest that
in order to form a best response at time $t$, coordinator $i$ needs to compute (or form beliefs on) the data that coordinators $-i$'s CIB strategies use, i.e. the CCI and the SPI of other coordinators. Coordinator $i$ also needs to estimate all the hidden information in order to evaluate the payoffs. Coordinator $i$'s estimation task can be divided into three sub-tasks: (i) to form a belief on her own hidden information $\bX_{t-d+1:t}^i$, (ii) to recover coordinators $-i$'s CCI $(B_t^{k})_{k\in\mathcal{I}\backslash\{i\} }$, and (iii) to form a belief on coordinators $-i$'s SPI and hidden information $\bX_{t-d+1:t}^{-i}$. 

For the first sub-task, coordinator $i$ can compute the belief using $(Y_{t-d+1:t-1}^{i}, \bU_{t-d:t-1}, S_t^i)$ through the function $P_t^i$ defined in Lemma \ref{lem: selfbelief}, without using any belief generation system.
For the second sub-task, recall that $B_t^k$ includes $(Y_{t-d+1:t-1}^{i}, \bU_{t-d:t-1})$, which coordinator $i$ already knows. Thus, to complete the second task, coordinator $i$ can simply use $(\psi^{k})_{k\in\mathcal{I}\backslash\{i\} }$ and the common information $H_t^0$ to compute all the beliefs in $(B_t^{k})_{k\in\mathcal{I}\backslash\{i\} }$. Condition (I), namely that the CIB strategies for coordinators other than $i$ are self-consistent, ensures that coordinator $i$ can also accomplish the third sub-task using the beliefs in $(B_t^{k})_{k\in\mathcal{I}\backslash\{i\} }$ due to Lemma \ref{lem: piistruebelief0}. By using self-consistent CIB strategies, coordinators $-i$ effectively ``invite'' coordinator $i$ to use the same belief generation system as $-i$.

Thus, all of coordinator $i$'s sub-tasks can be done if she keeps track of her own $S_t^i$ and the CCI $(B_t^{k})_{k\in\mathcal{I}\backslash\{i\} }$ used by others. Therefore, coordinator $i$ can form a best response with a strategy that chooses prescriptions based on $(B_t^{k})_{k\in\mathcal{I}\backslash\{i\} }$ and $S_t^i$ at time $t$. Condition (II), namely that the belief generation systems are identical, ensures that $B_t^k$'s are identical for all $k\in\mathcal{I}\backslash\{i\}$ and hence the best response described above becomes a CIB strategy with the same belief generation system as the one used by all coordinators other than $i$.

\begin{remark}
Note the CIB strategy that is a best-response strategy for coordinator $i$ in Lemma \ref{lem: closenessd1} may not necessarily be self-consistent. However, the equilibrium strategies in a CIB-CNE (which we will introduce later) will be self-consistent for all players.
\end{remark}


\subsection{Coordinators' Nash Equilibrium in CIB Strategies and Sequential Decomposition}

The fact that one of coordinator $i$'s best responses to others using CIB strategies (with identical and self-consistent belief generation systems) is itself a CIB strategy (with the same belief generation system as others) suggests the possibility of a Coordinators' Nash Equilibrium (CNE) where all coordinators are using CIB strategies with identical and self-consistent belief generation systems. We refer to such a CNE as a CIB-CNE. More formally, a CIB-CNE is a CIB strategy profile $(\lambda^{*i}, \psi^i)_{i\in\mathcal{I}}$ where (i) all coordinators have the same belief generation system, i.e., for all for all $i\in\mathcal{I}$, $\psi^i=\psi^*$ for some $\psi^*$, (ii) for each $k\in\mathcal{I}$, $\psi^{*, k}$ is consistent with $\lambda^k$, and (iii) for each $i\in\mathcal{I}$, the CIB strategy $(\lambda^{*i}, \psi^i)$ is a best response for coordinator $i$ to $(\lambda^{*k}, \psi^k)_{k\in\mathcal{I}\backslash\{i\}  }$.

Notice that in a CIB-CNE all coordinators are using the same belief generation system, hence the CCI $B_t^i$ (as defined in Definition \ref{def: cci2}) is the same for all coordinators. We denote the identical $B_t^i$ for all coordinators by $B_t$. Furthermore, when all coordinators other than $i$ are using fixed CIB strategies, $(B_t, S_{t}^i)$ can be viewed as an information state for coordinator $i$'s stochastic control problem (see proof of Lemma \ref{lem: closenessd1} for details). Based on this observation, we introduce a backward inductive computation procedure for determining CIB-CNEs where $B_t$ is used as an information state. Our procedure decomposes the game into a collection of one-stage games, one for each time $t$ and each realization of $B_t$. These one-stage games are used to characterize a CIB-CNE in a backward inductive manner. 

\begin{defn}[Stage Game]\label{def: stagegamed2}
	Given the value functions $V_{t+1} = (V_{t+1}^i)_{i\in\mathcal{I}}$, where $V_{t+1}^i: \mathcal{B}_{t+1}\times \mathcal{S}_{t+1}^i \mapsto \mathbb{R}$, a realization of the CCI $b_t=(\bm{\pi}_t, y_{t-d+1:t-1}, u_{t-d:t-1})$ where $\bm{\pi}_t=(\pi_t^i)_{i\in\mathcal{I}}, \pi_t^i\in\Delta(\mathcal{S}_t^i)$, and update functions $\psi_t^*=(\psi_t^{*, i})_{i\in\mathcal{I}}, \psi_t^{*, i}: \mathcal{B}_t\times \mathcal{Y}_t\times \mathcal{U}_t\mapsto \Delta(\mathcal{S}_{t+1}^i)$, we define a stage game for the coordinators dynamic game as follows:
	
	\textbf{Stage Game} $G_t(V_{t+1}, b_t, \psi_t^*)$:
	\begin{itemize}
		\item There are $|\mathcal{I}|$ players, each representing a coordinator.
		\item $(V_{t+1}, b_t, \psi_t^*)$ are commonly known.
		\item Nature chooses $\mathbf{Z}_t=(\mathbf{S}_t, \bX_{t-d+1:t}, \mathbf{W}_t^Y)$\footnote{Since $\mathcal{X}_t,\mathcal{U}_t,\mathcal{Y}_t$ are finite sets, one can assume that $\mathbf{W}_t^Y$ also takes finite values without lost of generality.}, where $\mathbf{S}_t= (S_t^k)_{k\in\mathcal{I}}$.
		\item Player $i$ observes $S_{t}^i=s_{t}^i$.
		\item Player $i$'s belief on $\mathbf{Z}_t$ is given by
		\begin{align}
		&\quad\beta_t^i(\tilde{z}_t|s_{t}^i) = \bm{1}_{\{\tilde{s}_t^i = s_t^i \}  } \prod_{k\neq i} \pi_t^k(\tilde{s}_{t}^k)\times \\
		&\times \prod_{k\in \mathcal{I}} P_t^k(\tilde{x}_{t-d+1:t}^k| y_{t-d+1:t-1}^k, u_{t-d:t-1}, \tilde{s}_{t}^k)\Pr(\tilde{w}_t^{k, Y}),\\
		&\quad \forall \tilde{z}_{t} = (\tilde{s}_t, \tilde{x}_{t-d+1:t}, \tilde{w}_t^Y)\in\mathcal{S}_{t}\times \mathcal{X}_{t-d+1:t}\times \mathcal{W}_t^Y.
		\label{beliefstagegame2}
		\end{align}
		where $P_t^k$ is the belief function defined in Eq. \eqref{eq: belieffunction}.
		\item Player $i$ selects a prescription $\bm{\Gamma}_t^i\in\prescription_t^i$ as her action.
		\item Player $i$ has utility
		\begin{align}\label{Qfunc2}
		Q_t^i(\mathbf{Z}_t, \bm{\Gamma}_{t}) = r_t^i(\bX_t, \bU_t) + V_{t+1}^i(B_{t+1}, S_{t+1}^i),
		\end{align}
		where
		\begin{align*}
		U_t^{k, j} &= \Gamma_t^{k, j}(\bX_{t}^{k, j})\quad\forall (k, j)\in\mathcal{N},\\
		B_{t+1} &= ((\Pi_{t+1}^k)_{k\in\mathcal{I}}, (y_{t-d+2:t-1}, \bY_{t}), (u_{t-d+1:t-1}, \bU_{t})),\\
		\Pi_{t+1}^k &= \psi_t^{*, k}(b_t, \bY_t, \bU_t)\quad\forall k\in\mathcal{I},\\
		Y_t^k &= \ell_t^j(\bX_t^k, \bU_t, W_t^{k, Y})\quad\forall k\in \mathcal{I},\\
		S_{t+1}^i &= \iota_t^i(S_t^i, \bX_{t-d+1}^i, \bm{\Gamma}_t^i)
		\end{align*}
	\end{itemize}
\end{defn}

Given the stage game $G_t(V_{t+1}, b_t, \psi_t^*)$, we define two associated concepts:

\begin{defn}[IBNE Correspondence]\label{def: ibe}
	Given the value functions $V_{t+1} = (V_{t+1}^i)_{i\in\mathcal{I}}$, where $V_{t+1}^i: \mathcal{B}_{t+1}\times \mathcal{S}_{t+1}^i \mapsto \mathbb{R}$ and belief update functions $\psi_t^*=(\psi_t^{*, i})_{i\in\mathcal{I}}, \psi_t^{*, i}: \mathcal{B}_t\times \mathcal{Y}_t\times \mathcal{U}_t\mapsto \Delta(\mathcal{S}_{t+1}^i)$, the \emph{Interim Bayesian Nash Equilibrium correspondence} $\mathrm{IBNE}_t(V_{t+1}, \psi_t^*)$ is defined as the set of all $\lambda_t=(\lambda_t^i)_{i\in\mathcal{I}}, \lambda_t^i:\mathcal{B}_t\times \mathcal{S}_{t}^i \mapsto \Delta(\prescription_{t}^i)$ such that
	\begin{align}
	&\lambda_t^{i}(b_t, s_{t}^i) \in \\
	&\underset{\eta\in \Delta(\prescription_{t}^i)}{\arg\max} \left(\sum_{\tilde{z}_t, \tilde{\gamma}_t} \left[\eta(\tilde{\gamma}_{t}^i) Q_t^i(\tilde{z}_t, \tilde{\gamma}_t)  \beta_t^i(\tilde{z}_t|s_{t}^i)\prod_{k\neq i}\lambda_t^{k}(\tilde{\gamma}_t^k|b_t, \tilde{s}_{t}^k)\right]\right)\\
	&\qquad\forall b_t \in\mathcal{B}_t, s_{t}^i\in\mathcal{S}_{t}^i, \forall i\in\mathcal{I},
	\end{align}
	where $\beta_t^i$ and $Q_t^i$ are defined using $(V_{t+1}^i ,b_t, \psi_t^*)$ in \eqref{beliefstagegame2} and \eqref{Qfunc2} respectively.
\end{defn}

\begin{defn}[DP Operator]\label{def: dpo}
	Given a value function $V_{t+1}^i: \mathcal{B}_{t+1}\times \mathcal{S}_{t+1}^i \mapsto \mathbb{R}$ and a CIB strategy profile $(\lambda_t^*, \psi_t^*)$ at time $t$, where $\lambda_t^*=(\lambda_{t}^{*i})_{i\in\mathcal{I}}, \lambda_{t}^{*i}:\mathcal{B}_t\times \mathcal{S}_{t}^i \mapsto \Delta(\prescription_{t}^i)$ and  $\psi_t^*=(\psi_t^{*, i})_{i\in\mathcal{I}}, \psi_t^{*, i}:\mathcal{B}_t\times \mathcal{Y}_t\times \mathcal{U}_t\mapsto \Delta(\mathcal{S}_{t+1}^i)$, the \emph{dynamic programming operator} $\mathrm{DP}_t^i$ defines the value function at time $t$ through
	\begin{align*}
	&~\quad [\mathrm{DP}_t^i(V_{t+1}^i, \lambda_t^*, \psi_t^*)](b_t, s_{t}^i)\\
	&:=\sum_{\tilde{z}_t, \tilde{\gamma}_t} Q_t^i(\tilde{z}_t, \tilde{\gamma}_t) \beta_t^i(\tilde{z}_t|s_{t}^i) \prod_{k\in \mathcal{I}}\lambda_t^{*k}(\tilde{\gamma}_t^k|b_t, \tilde{s}_{t}^k),
	\end{align*}
	where $\beta_t^i$ and $Q_t^i$ are defined using $(V_{t+1}^i ,b_t, \psi_t^*)$ in \eqref{beliefstagegame2} and \eqref{Qfunc2} respectively.
	
\end{defn}

\begin{thm}[Sequential Decomposition]\label{thm: sd1}
	Let $(\lambda^{*i}, \psi^{*})_{i\in\mathcal{I}}$ be a CIB strategy profile with identical belief generation system $\psi^*$ for all $i\in\mathcal{I}$.
	If this strategy profile satisfies the dynamic program defined below:
	\begin{align*}
	V_{T+1}^i(\cdot, \cdot) &= 0\quad \forall i\in\mathcal{I};
	\end{align*}
	and for $t\in\mathcal{T}$
	\begin{align}
	&\lambda_t^*\in \mathrm{IBNE}_t(V_{t+1}, \psi_t^*);\label{thomyorke}\\
	&\psi_t^*\text{ is consistent with }\lambda_t^*;\label{liannelahavas}\\
	&V_t^i := \mathrm{DP}_t^i(V_{t+1}^i, \lambda_t^*, \psi_t^*)\quad\forall i\in\mathcal{I},
	\end{align}
	then $(\lambda^{*i}, \psi^*)_{i\in\mathcal{I}}$ forms a CIB-CNE.
\end{thm}

\begin{proof}
	See Appendix \ref{app: sd1}.
\end{proof}

\begin{remark}
	Note that \eqref{thomyorke} and \eqref{liannelahavas} can be verified for each realization $b_t\in\mathcal{B}_t$ separately, i.e., one can check that $\lambda_t^*(b_t, \cdot)$ is an IBNE of the stage game game $G_t(V_{t+1}, b_t, \psi_t^*(b_t, \cdot))$, and that $\psi_t^*(b_t, \cdot)$ is consistent with $\lambda_t^*(b_t, \cdot)$ for each $b_t$. 
\end{remark}

\subsection{Existence of CIB-CNE}\label{sec: existence}
We have shown in Theorem \ref{thm: spibexist} that an SPIB-CNE always exists. However, a CIB-CNE does not necessarily exist, even when each team contains only one member (i.e. in games among individuals). We present below one example where CIB-CNEs do not exist. 

\begin{example}\label{ex: nonexistence}
	Consider a 3-stage dynamic game (i.e. $\mathcal{T} = \{1, 2, 3\}$) with two players: Alice (A) and Bob (B). Each player forms a one-person team. Let $X_t^A\in \{-1, 1\}$ and $X_t^B\equiv \varnothing$, i.e. Bob is not associated with a state. Let $\mathbf{Y}_t=\varnothing$, i.e. there is no public observation of the states. The initial state $X_1^A$ is uniformly distributed on $\{-1, 1\}$. At $t=1$, (a) Alice can choose an action $U_1^A\in \{-1, 1\}$ and Bob has no actions to take; (b) the next state is given by $X_2^A=X_1^A\cdot U_1^A$; (c) the instantaneous reward is given by
	\begin{equation}
		r_1^A(\bX_1, \bU_1) = -r_1^B(\bX_1, \bU_1) = \varepsilon \cdot \bm{1}_{\{U_1^A = +1 \} },
	\end{equation}
	where $\varepsilon\in (0, \frac{1}{3})$.
	
	At $t=2$, (a) neither player has any action to take; (b) the state at next time is given by $X_3^A=X_2^A$; (c) the instantaneous rewards are 0 for both players; (This stage is a dummy stage inserted in the game to alter the definition of the CCI at the beginning of the last stage.) 
	
	At $t=3$, (a) Alice has no action to take, and Bob chooses $U_3^B\in \{\mathrm{L}, \mathrm{R} \}$; (b) The instantaneous reward $r_3^A(\bX_3, \bU_3)$ for Alice is given by
	\begin{align*}
	r_3^A(-1, \mathrm{L}) &= 0,\quad r_3^A(-1, \mathrm{R}) = 1\\
	r_3^A(+1, \mathrm{L}) &= 2,\quad
	r_3^A(+1, \mathrm{R}) = 0
	\end{align*}
	and $r_3^B(\bX_3, \bU_3) = - r_3^A(\bX_3, \bU_3)$.
\end{example}

In a game where each team contains only one person, we can assume the delay $d$ to be any number (see Remark \ref{remark: onepersonteam}).	
In the next proposition, we view Example \ref{ex: nonexistence} as a game among teams with internal delay $d=1$. 
	
	\begin{prop}\label{prop: nonexistenceexample}
		There exist no CIB-CNE in the game described in Example \ref{ex: nonexistence}.
	\end{prop}
	
	\begin{proof}
		See Appendix \ref{app: nonexistenceexample}.
	\end{proof}

	\begin{remark}
	One can provide an example for non-existence of CIB-CNE for any $d>0$ by inserting $d-1$ additional dummy stages (analogous to stage 2) into Example \ref{ex: nonexistence}, and viewing it as a game among teams with internal delay $d$. Example \ref{ex: nonexistence} can also be used to show that the CIB-PBE concept defined in \cite{ouyang2016dynamic} for games among individuals does not exist in general, hence the conjecture in \cite{ouyang2016dynamic} that a CIB-PBE always exists is not true.
	\end{remark}

	Intuitively, the reason that a CIB-CNE does not exist in this game is that at $t=3$, a CIB strategy requires Bob to choose his action based only on a compressed version of his information rather than the full information. This compression does not hurt Bob's ability to form a best response. However, in an equilibrium, Bob needs to carefully choose from the set of optimal responses to induce Alice to play the predicted mixed strategy. Being unable to choose different actions under different histories due to information compression makes Bob unable to sustain an equilibrium. In this game, as in the example in \cite{maskin2013youtube}, payoff irrelevant information plays an essential role in sustaining the equilibrium.

In the remainder of this section we present two subclasses of the dynamic games described in Section \ref{sec: problemformulation} where CIB-CNEs exist.

\subsubsection{Signaling-Neutral Teams}
In this subsection we consider $d=1$. One subclass of games where CIB-CNEs exist is when the teams are \emph{signaling-neutral}. In these games, the agents are indifferent in terms of signaling to other teams, i.e. revealing more or less information about their private information to the other teams does not affect their utility. (Note that agents can always actively reveal information to their teammates through their actions.)

We shall now describe the game: 

\begin{defn}\label{def: publicteam}
	A team $i$ whose state $\bX_t^i$ can be recovered from $(\bY_{t}^i, \bU_{t})$ (i.e. for every fixed $u_{t}$, $\ell_t^i(x_t^i, u_{t}, W_t^{i, Y})$ has disjoint support for different $x_t^i\in\mathcal{X}_t^i$) is called a \emph{public} team. Otherwise, it is called \emph{private} team.
\end{defn}

For a public team $i$, the private state $\bX_{t-1}^i$ is effectively part of the common information of all members of all teams.

\begin{defn}[Information Dependency Graph]
	The information dependency graph $\mathcal{G}$ of a dynamic game is a directed graph defined as follows: The vertices represent the teams. A directed edge $i\leftarrow j$ is present if either the state transition, the observation, or the instantaneous reward of team $i$ at some time $t$ depends directly on either the state or the action of team $j$. In other words, there is no directed edge from $j$ to $i$ if and only if $\bX_{t+1}^i = f_t^i(\bX_t^i, \bU_t^{-j}, W_t^{i, X})$, $\bY_{t}^i = \ell_t^i(\bX_t^i, \bU_t^{-j}, W_t^{i, Y})$ and $r_t^i(\bX_t, \bU_t)=r_t^i(\bX_t^{-j}, \bU_t^{-j})$ for some functions $f_t^i, \ell_t^i, r_t^i$ for all $t$. Self loops are not considered in this graph.
\end{defn}

\begin{thm}\label{thm: case1}
	Let $d=1$. If every strongly connected component of the information dependency graph $\mathcal{G}$ of a dynamic game consists of either (I) a single team, or (II) multiple public teams, then a CIB-CNE exists.
\end{thm}

\begin{proof}
	See Appendix \ref{app: case1}.
\end{proof}

\begin{remark}
	The precedence relation among teams considered in Theorem \ref{thm: case1} is similar to the $s$-partition of teams that was presented and analyzed in \cite{Yoshikawa1978}.
\end{remark}

When the condition in Theorem \ref{thm: case1} is satisfied, all teams will be neutral in signaling: When a private team $i$ sends information, this information is only useful to those teams whose actions do not affect team $i$'s utility. Public players are always neutral in signaling since their state history is publicly available. 

Notice that in Example \ref{ex: nonexistence}, Alice (as a one-person team) is a private team while Bob is a public team. The instantaneous reward of Bob at $t=3$ depends on Alice's state $X_2^A$, while Alice's instantaneous reward at $t=3$ depends on Bob's action. Hence Alice and Bob form a strongly connected component in the information dependency graph.

\subsubsection{Signaling-Free Equilibria}
In this section, we introduce another class of games where CIB-CNE exists. These games are games-among-teams extension of \emph{Game M} defined in \cite{ouyang2016dynamic}. We present the result for a general $d>0$.

\begin{example}
	Consider a dynamic game that satisfies the following conditions.
	\begin{itemize}
		\item States are uncontrolled, i.e. $\bX_{t+1}^i = f_t^i(\bX_t^i, W_t^{i, X})$.
		\item Observations are uncontrolled, i.e. $Y_t^i=\ell_t^i(\bX_t, W_t^{i, Y})$.
		\item Instantaneous rewards of team $i$ can be expressed as $r_t^i(\bX_t^{-i}, \bU_t)$.
	\end{itemize}
\end{example}

\begin{thm}\label{thm: signalfreegame}
	A dynamic game that satisfies the above conditions has a CIB-CNE.
\end{thm}

\begin{proof}
	See Appendix \ref{app: signalfreegame} for a direct proof. Alternatively, one can first assume that the teams share information with a delay of $d=0$, then we can view a team as one individual since team members have the same information. Then one can apply results for \emph{Game M} in \cite{ouyang2016dynamic} to obtain an equilibrium where each player/team plays a public strategy (i.e. a strategy that does not use private information), in particular, a strategy where actions are solely based on the common information based belief. Since public strategies can also be played when $d>0$, we conclude that the equilibrium we obtained is also an equilibrium for the original game.
\end{proof}

\section{ADDITIONAL RESULTS}\label{sec: add}
\subsection{Separated Dynamics and Observations among Teammates}
Consider a special case of the model in Section \ref{sec: problemformulation} where both the evolution and observations of the local states of each member of each team are conditionally independent given the actions, i.e. 
\begin{align}
	X_{t+1}^{i, j} &= f_t^{i, j}(X_t^{i, j}, \bU_t, W_t^{i, j}),\\
	\bY_t^i &= (Y_{t}^{i, j})_{(i, j)\in\mathcal{N}_i}\\
	Y_{t}^{i, j} &= \ell_t^{i, j}(X_t^{i, j}, \bU_t, W_t^{i, j, Y}),
\end{align}
where $(W_t^{i, j, X}, W_t^{i, j, Y})_{t\in\mathcal{T}, (i, j)\in\mathcal{N}}$ are mutually independent primitive random variables.\footnote{This is a correction from an earlier version of this paper, where we did not assume that $Y_{t}^{i, j} = \ell_t^{i, j}(X_t^{i, j}, \bU_t, W_t^{i, j, Y})$.}

In this case, we show that the independence among team members' state dynamics enable us to consider equilibria where the coordinators assign prescriptions that map $X_{t}^{i, j}$ to $U_t^{i, j}$ (instead of mapping $X_{t-d+1:t}^{i, j}$ to $U_t^{i, j}$); this is because, given $H_t^i$, the belief of member $(i, j)$ about her teammates' states is independent of $X_{t-d+1:t}^{i, j}$. In other words, one can replace the hidden information $\bX_{t-d+1:t}^i$ with the \emph{sufficient hidden information} $\bX_t^i$.\footnote{The compression of hidden information to sufficient hidden information is similar to the shredding of irrelevant information in \cite{mahajan2013controlsharing}.}

\begin{defn}[Simple Prescriptions]
	A \emph{simple prescription} for coordinator $i$ at time $t$ is a collections of functions $\theta_t^i = (\theta_t^{i, j})_{(i, j)\in\mathcal{N}_i}, \theta_t^{i, j}: \mathcal{X}_{t}^{i, j} \mapsto \mathcal{U}_t^{i, j}$.
\end{defn}

\begin{lemma}\label{lem: suffhidden}
	Suppose that $g^{-i}$ is a behavioral coordination strategy profile for coordinators other than coordinator $i$, then there exists a best response behavioral coordination strategy $g^i$ for coordinator $i$ that chooses randomized simple prescriptions based on $\overline{H}_t^i$.
\end{lemma}

\begin{proof}
	See Appendix \ref{app: suffhidden}.
\end{proof}

Given the above result, one can restrict attention to \emph{sufficient hidden information based} strategies where each coordinator $i$ assigns simple prescriptions based on $\overline{H}_t^i$. 
Consequently, results analogue to that of Sections \ref{sec: spi} and \ref{sec: 4} can be derived considering similar compression of private and common information.

\subsection{Refinement of Coordinators' Nash Equilibrium}\label{sec: refinement}
In the game among coordinators, one can also consider Coordinators' \emph{weak Perfect Bayesian Equilibrium} (wPBE) \cite{mas1995microeconomic} as a refinement of CNE. Coordinator's wPBE provides a refinement of Coordinator's Nash Equilibrium by ruling out equilibrium outcomes that rely on non-credible threats \cite{fudenberg1991game}.\footnote{We refer an interested reader to Chapter 9 of \cite{mas1995microeconomic} for a detailed description of wPBE.}
\begin{defn}[Coordinators' wPBE]\label{def: wpbe}
	Define $\mathcal{H}_t^{*} = \mathcal{X}_{1:t}\times \mathcal{Y}_{1:t-1} \times \mathcal{U}_{1:t-1} \times \varGamma_{1:t-1}$. Let $g$ denote a behavioral coordination strategy profile of all coordinators and $\vartheta = (\vartheta_t^{i})_{i\in \mathcal{I}, t\in\mathcal{T} }, \vartheta_t^{i}: \overline{\mathcal{H}}_t^i \mapsto \Delta(\mathcal{H}_t^{*})$ denote a belief system. The strategy profile $g$ is said to be sequentially rational given $\vartheta$ if
	\begin{align}
	g_{t:T}^i \in \underset{\tilde{g}_{t:T}^i }{\arg\max}~  J_t^i(\tilde{g}_{t:T}^i, g_{t:T}^{-i}; \vartheta_t^i, \overline{h}_t^i)&\\
	\forall \overline{h}_t^i\in \overline{\mathcal{H}}_t^i, \forall i\in\mathcal{I}, \forall t\in\mathcal{T} &
	\end{align}
	where
	\begin{align}
	J_t^i(\tilde{g}_{t:T}; \vartheta_t^i, \overline{h}_t^i) := \sum_{\tilde{h}_{t}^*}\E^{ \tilde{g}_{t:T}}\left[\sum_{\tau=t}^T r_{\tau}^i(\bX_{\tau}, \bU_{\tau})\Big|\tilde{h}_t^*\right]\vartheta_t^i(\tilde{h}_t^*|\overline{h}_t^i);
	\end{align}
	the belief system $\vartheta$ is said to be consistent with $g$ \cite{mas1995microeconomic} if
	\begin{align}
	\Pr^g(\overline{h}_t^i) > 0~\Rightarrow~\vartheta_t^i(\tilde{h}_{t}^*|\overline{h}_t^i) = \frac{\Pr^g(\tilde{h}_{t}^*, \overline{h}_t^i)}{\Pr^g(\overline{h}_t^i)}&\\
	\forall \tilde{h}_{t}^*\in\mathcal{H}_t^*~ \forall \overline{h}_t^i\in\overline{\mathcal{H}}_t^i~ \forall t\in \mathcal{T}~ \forall i\in\mathcal{I}.&
	\end{align}
	
	A pair $(g, \vartheta)$ is called a Coordinators' wPBE if $g$ is sequentially rational given $\vartheta$ and $\vartheta$ is consistent with $g$.
\end{defn}

Let $\rho$ be an SPIB strategy profile and $\vartheta$ to be a belief system. A pair $(\rho, \vartheta)$ is called an SPIB-wPBE if it forms a wPBE.

\begin{prop}
	SPIB-wPBE exists in the game among coordinators.
\end{prop}

\begin{proof}
	The proof follows steps similar to the proof of Theorem \ref{thm: spibexist}.
\end{proof}

As a result of the sequential decomposition of the dynamic game, with some assumptions on the belief generation systems, a CIB-CNE obtained from the sequential decomposition is a wPBE as well, where the beliefs $\vartheta$ can be derived from the CCI. This is formalized in the following proposition.

\begin{defn}
	Define
	\begin{align}
		\hat{\varPi}_t(h_t^0):=\{&\pi_t\in \prod_{k\in\mathcal{I}}\Delta(\mathcal{S}_t^k): \exists g_{1:t-1}, ~\mathrm{s.t.}~\Pr^{g_{1:t-1}}(h_t^0)>0,\\
		 &\prod_{k\in\mathcal{I}}\pi_t^k(s_t^k) = \Pr^{g_{1:t-1}}(s_t|h_t^0)~\forall s_t\in\mathcal{S}_t \}.
	\end{align}
	A belief generation system $\psi^{*}=(\psi_t^{*})_{t\in\mathcal{T}}, \psi_t^{*}: \mathcal{B}_t \times \mathcal{Y}_t\times \mathcal{U}_t \rightarrow \prod_{k\in\mathcal{I}}\Delta(\mathcal{S}_t^k)$ is said to be \emph{regular} if for all $t\in\mathcal{I}$, we have $\psi_t^{*}\left(\pi_t, y_{t-d+1:t}, u_{t-d:t}\right)\in \hat{\varPi}_{t+1}(h_{t+1}^0)$ for all $\pi_t\in \hat{\varPi}_{t}(h_{t}^0)$.
\end{defn}

Intuitively, a belief generation system is regular if it assigns positive probability only to realizations of SPI that are admissible under some strategy profile $g$.

\begin{prop}\label{prop: pbe}
	Let $(\lambda^*, \psi^*)$ be a CIB strategy profile that satisfies the condition of Theorem \ref{thm: sd1}. Assume that $\psi^*$ is regular. Let $g^*$ be the behavioral coordination strategy profile induced from $(\lambda^*, \psi^*)$. Then there exist a belief system $\vartheta^*$ such that $(g^*, \vartheta^*)$ forms a Coordinators' wPBE.
\end{prop}

\begin{proof}
	See Appendix \ref{app: pbe}.
\end{proof}

\section{DISCUSSION}\label{sec: dis}

\subsection{Implementation of Behavioral Coordination Strategies}
One can also interpret behavioral coordination strategies as strategies with coordinated randomization, i.e., the strategies are randomized, but all the team members know exactly how this randomization is done. We note that one can view the main purpose of randomization as to ``confuse'' other teams. As such, it is best to use coordinated randomization where every team member knows what partial mapping their teammate is using; such coordinated randomization is superior to private and independent randomization by each individual member in a team: This is since individual randomization can create information that are unknown to teammates, while the same ``confusion'' effect to other teams can be achieved with coordinated randomization.

To implement behavioral coordination strategies, a team can utilize a correlation device which generates a random seed at each time $t$. Then each member $(i, j)$ of the team $i$ can choose an action based on $H_t^{i, j}$ and present and past random seeds generated by the correlation device, or equivalently, choose an action based on $(H_t^{i, j}, \bm{\Gamma}_{1:t-1}^i)$ where $\bm{\Gamma}_{1:t-1}^i$ is sequentially updated. If the behavioral coordination strategy is a CIB strategy, then member $(i, j)$ need to use $(B_t, \bX_{t-d}^i, \bm{\Phi}_t^i, X_{t-d+1:t}^{i, j})$ and current random seed to chose an action, where $(B_t, \bm{\Phi}_t^i)$ are sequentially updated.

In the absence of correlation devices accessible at every time, a behavioral coordination strategy can also be implemented as its equivalent mixed strategy (recall Lemma \ref{lem: pure2coord} and Lemma \ref{lem: bcs2ms}): Before the beginning of the game, the team can jointly pick a strategy profile in $\mathcal{G}^i$ randomly, according to a distribution induced from the behavioral coordination strategy. 

\subsection{Stage Game: IBNE vs BNE}\label{sec: ibevsbne}

One can observe that the belief of the agents defined in the stage game (Definition \ref{def: stagegamed2}) can be seen as a conditional distribution derived from the common prior
\begin{align}
&\quad~\beta_t(\tilde{z}_t) \\
&=\prod_{k\in \mathcal{I}}\left[ \pi_t^k(\tilde{s}_{t}^k) P_t^k(\tilde{x}_t^k|y_{t-d+1:t-1}^k, u_{t-d:t-1}, \tilde{s}_t^k)\Pr(\tilde{w}_t^{k, Y})\right].~~\label{priorofstagegame}
\end{align}
However, in the aforementioned stage game we focus on the beliefs of agents instead of a common prior, and we use \emph{Interim Bayesian Nash Equilibrium} (IBNE) as the equilibrium concept instead of BNE. This is since, unlike a standard Bayesian game with a common prior, the true prior of the stage game is dependent on the actual strategy played in previous stages. The prior $\beta_t$ described in \eqref{priorofstagegame} may not be a true prior, since some coordinator $i$ may have already deviated from the strategy prediction which $\pi_t^i$'s were relying on. However, coordinator $i$ is always trying to optimize her reward given $(b_t, s_{t}^i)$, no matter $\pi_t^i(s_{t}^i)=0$ or not. Hence in this stage game, we must consider the player's belief and strategy for all possible realizations $s_{t}^i$ under \emph{any} strategy profile, not just those with positive probability under the prior in \eqref{priorofstagegame}. The corresponding equilibrium concept is Interim Bayesian Equilibrium instead of Bayes-Nash Equilibrium. IBNE strengthens BNE by requiring the strategy of an agent to be optimal under \emph{all} private information realizations, including those with zero probability under the common prior.

\subsection{Choice of Compressed Common Information}\label{app: why}
In decentralized control \cite{nayyar2013decentralized} and certain settings of games among individuals \cite{nayyar2013common,ouyang2016dynamic}, a common information based belief $\Pi_t$ on the state is usually enough to serve as an information state, or compression of common information. However, in our setting we use a subset of actions and observations in addition to the CIB belief as the compressed common information. We argue below that this is necessary for our setting.

To illustrate the point, consider the case $d=1$ and assume that all coordinators use the same belief generation system and hence the same CCI (denoted by $B_t^*$). An alternative for the CCI $B_t^*=((\Pi_t^{*, i})_{i\in\mathcal{I}}, \bU_{t-1})$ is the CIB belief $\tilde{\bm{\Pi}}_t^*=(\tilde{\Pi}_t^{*, i})_{i\in\mathcal{I}}, \tilde{\Pi}_t^{*, i}\in \Delta(\mathcal{X}_{t-1:t}^i)$ where $\tilde{\Pi}_t^{*, i}$ represents the belief on $\bX_{t-1:t}^i$ based on common information. One might argue that we can use $\tilde{\bm{\Pi}}_t^*$ instead of $B_t^*$ through the following argument: After we transform the game into games among coordinators, because of the full recall of coordinator $i$, coordinator $i$'s belief (on other coordinator's private information and all hidden information) is independent of her behavioral coordination strategy $\tilde{g}^i$. Hence coordinator $i$ can always form this belief as if she was using the strategy prediction $g^{*i}$ no matter what strategy she is actually using. 

However this argument can run into technical problems: A crucial step for Lemma \ref{lem: closenessd1} is Eq. \eqref{eq: ftid2}, which establishes that coordinator $i$'s belief can be expressed as a function of $(B_t^*, \bX_{t-1}^i)$ for any behavioral coordination strategy $\tilde{g}^i$ coordinator $i$ might use. To use $\tilde{\bm{\Pi}}_t^*$ alone as the information state, one need to argue that coordinator $i$'s belief on her hidden information, $\Pr(X_t^i=\cdot|x_{t-1}^i, u_{t-1})$, can be computed solely through $(\tilde{\pi}_t^{*, i}, x_{t-1}^i)$ without using $u_{t-1}$. Through belief independence of strategy, one may argue that
\begin{align}
\Pr(x_t^i|x_{t-1}^i, u_{t-1}) &= \Pr^{g^{*i}, g^{*-i}}(x_t^i|x_{t-1}^i, u_{t-1})\\
&=\Pr^{g^{*, i}, g^{*,-i}}(x_t^i|x_{t-1}^i, y_{1:t-1}, u_{1:t-1})\\
&=\dfrac{\Pr^{g^{*, i}, g^{*,-i}}(x_t^i, x_{t-1}^i| y_{1:t-1}, u_{1:t-1})}{\Pr^{g^{*, i}, g^{*,-i}}( x_{t-1}^i|y_{1:t-1}, u_{1:t-1})}\\
&=\dfrac{\tilde{\pi}_t^{*, i}(x_{t-1}^i, x_t^i)}{\sum_{\tilde{x}_t^i} \tilde{\pi}_t^{*, i}(x_{t-1}^i, \tilde{x}_t^i)}.\label{trydiv}
\end{align}

However, the above argument is not always valid. It is only valid when the denominator of \eqref{trydiv} is non-zero, but it can be zero. One simple example is as the following: Let $\hat{x}_{t-1}^{i}\in \mathcal{X}_{t-1}^i$ be some fixed state and $\hat{u}_{t-1}^i\in \mathcal{X}_{t-1}^i$ be some fixed action profile. Let $\hat{\prescription}_{t-1}^i$ be the set of prescriptions that maps $\hat{x}_{t-1}^i$ to $\hat{u}_{t-1}^i$.
Suppose that the strategy prediction $g^{*i}$ is a behavioral coordination strategy satisfying the following:
\begin{align}
g^{*i}_{t-1}(\overline{h}_{t-1}^i)(\gamma_{t-1}^i) = 0\qquad \forall \overline{h}_{t-1}^i\in\overline{\mathcal{H}}_{t-1}^i, \gamma_{t-1}^i\in \hat{\prescription}_{t-1}^i,
\end{align}
i.e. under $g^{*i}$, coordinator $i$ never assigns any prescription that maps $\hat{x}_{t-1}^i$ to $\hat{u}_{t-1}^i$. If $\tilde{\pi}_t^{*, i}$ is consistent with the strategy prediction $g^{*i}$, then 
\begin{equation}
\sum_{\tilde{x}_t^i} \tilde{\pi}_t^{*, i}(\hat{x}_{t-1}^i, \tilde{x}_t^i) = \Pr^{g^i, g^{-i}}(\hat{x}_{t-1}^i|h_t^0) = 0
\end{equation}
if $u_{t-1}^i = \hat{u}_{t-1}^i$. When coordinator $i$ use a strategy $\tilde{g}^i$ such that $\bX_{t-1}^i = \hat{x}_{t-1}^i, \bU_{t-1}^i= \hat{u}_{t-1}^i$ could happen with non-zero probability, coordinator $i$ cannot use $\tilde{\pi}_t^{*, i}$ to form her belief on her hidden information. This is contrary to what we need in Eq. \eqref{eq: ftid2} in the proof of Lemma \ref{lem: closenessd1}, which states that the belief function is compatible with \emph{any} behavioral coordination strategy $\tilde{g}^i$.


\subsection{Connection with Sufficient Information Approach}
The compression of private information of coordinators in our model can be seen as an application of Tavafoghi et al.'s \cite{tavafoghi2018unified} sufficient information approach. One can show that our sufficient private information $S_{t}^i=(\bX_{t-d}^i, \bm{\Phi}_t^i)$ satisfy the definition of \emph{sufficient private information} (Definition 4) in \cite{tavafoghi2018unified} (hence we choose to use the same terminology): (i) It can be sequentially updated; (ii) it is sufficient for estimating future private information; (iii) it is sufficient for estimating cost; (iv) it is sufficient for estimating others' information at current time.\cite{tavafoghi2018unified} We note that (i) is true due to Eq. \eqref{eq: spiupdate}; (ii) and (iii) are established and utilized in our Lemma \ref{lem: suffprivated3}; (iv) is true because of conditional independence between the coordinators (our Lemma \ref{lem: condindep}). In \cite{tavafoghi2018unified}, the authors proved that one can compress the common information into the \emph{sufficient common information} (SCI) and consider \emph{sufficient information based} strategies, which choose actions based on SCI and SPI. The SCI is defined to be the common information based belief on the sufficient private information along with the system state, which is $\bX_{t-d+1:t}$ in our case. As we discussed in Section \ref{sec: ccicibd1}, our CCI $B_t$ can be used to create the belief on $\bX_{t-d+1:t}$, hence our CCI $B_t$ (defined in Definition \ref{def: cci2}) is playing the role of the SCI.


\section{CONCLUSION AND FUTURE WORK}\label{sec: concl}
We studied a model of dynamic games among teams with asymmetric information, where agents in each team share their observations with a delay of $d$. Each team is associated with a controlled Markov Chain, whose dynamics are controlled by the actions of all agents. We developed a general approach to characterize a subset of Nash Equilibria with the following feature: At each time, each agent can make their decision based on a compressed version of their information, instead of the full information. We identified two subclasses of strategies: sufficient private information based (SPIB) strategies, which only compresses private information, and compressed information based (CIB) strategies, which compresses both common and private information. 
We showed that while SPIB-strategy-based equilibria always exist, CIB strategy-based equilibria do not always exist.
We developed a backward inductive sequential procedure, whose solution (if it exists) is a CIB strategy-based equilibrium. We characterized certain game environment where the solution exists. Our results highlight the discord between compression of information, existence of (compression based) equilibria, and backward inductive sequential computation of such equilibria in stochastic dynamic games.

Moving forward, there are a few research problems arising from this work: (i) discovering broader conditions for the existence of CIB-CNE in the model of this paper; (ii) developing an efficient algorithm which solves the dynamic program of CIB-CNE (when they exist); (iii) determining minimal additional information needed to be added to the CCI such that CIB-CNE (under the new CCI) is guaranteed to exist; (iv) defining a notion of $\epsilon$-CIB-CNE, analyzing its existence, and developing sequential computation procedures to find them.

Other future research directions include identifying a suitable compression of information and developing a sequential decomposition for other models of games among teams, for example (i) games with continuous state and action spaces (e.g. linear quadratic Gaussian settings), and (ii) general models with non-observable actions.


\bibliographystyle{ieeetran}
\bibliography{mybib}

\appendix
\subsection{Proof of Claim in Example \ref{ex: 1}}\label{app: ex1}
Define two pure strategies $\mu^A$ and $\tilde{\mu}^A$ of Team A as follows:
\begin{align*}
\mu^{A, 1}(x_1^{A, 1}) &= x_1^{A, 1}, \quad \mu^{A, 2}(x_1^{A, 2}) = - x_1^{A, 2},\\
\tilde{\mu}^{A, 1}(x_1^{A, 1}) &= - x_1^{A, 1}, \quad \tilde{\mu}^{A, 2}(x_1^{A, 2}) = x_1^{A, 2}.
\end{align*}

Now, assume that Team A and Team B are restricted to use independently randomized strategies (type 2 strategies defined in Section \ref{sec: rew}). We will show in two steps that there exist no equilibria within this class of strategies.

\textbf{Step 1:} If Team A and Team B's type 2 strategies form an equilibrium, then Team A is playing either $\mu^A$ or $\tilde{\mu}^A$.

Let $p_j(x)$ denote the probability that player (A, $j$) plays $U_1^{A, j} = -x$ given $X_1^{A, j} = x$. Define $$q_j = \frac{1}{2}p_j(-1) + \frac{1}{2}p_j(+1),$$ i.e. the ex-ante probability that player (A, $j$) ``lies''.

Then we have 
$$\E[r_1^A(\bX_1, \bU_1)]=q_1(1-q_2) + q_2(1-q_1).$$

Under an equilibrium, Team B will optimally respond to Team A strategy's described through $(p_1, p_2)$. We can find a lower bound of Team B's reward by fixing a strategy: Consider the ``random guess'' strategy of Team B, where each of $(B, j)$ (for $j=1,2$) chooses $U_2^{B, j}$ uniformly at random irrespective of $\bU_1^A$ and independent of the other team member. Team B can thus guarantee an expected reward of $\frac{1}{2} + \frac{1}{2} = 1$
given any strategy of Team A. Since $r_2^A(\bX_2, \bU_2) = -r_2^B(\bX_2, \bU_2)$, we conclude that Team A's total reward in an equilibrium is upper bounded by 
\begin{align*}
&\quad~ q_1(1-q_2) + q_2(1-q_1) - 1  \\
&= -q_1q_2-(1-q_1)(1-q_2)\leq 0
\end{align*}

Let $\sigma^B$ denote the strategy of Team B. Let $\pi_j(u^1, u^2)$ denote the probability that player $(B, j)$ plays $U_2^{B, j} = -u^j$ given $U_1^{A, 1} = u^{1}, U_1^{A, 2} = u^{2}$ (i.e. the probability that player (B, $j$) believes that (A, $j$) was ``lying'' hence guesses the opposite of what was signaled). If Team A plays $\mu^A$, then the total reward of Team A is
\begin{align*}
	&\quad ~J^A(\mu^A, \sigma^B)\\
	&= 1 - \E[1 -\pi_1(X_1^{A, 1}, -X_1^{A, 2})  + \pi_2(X_1^{A, 1}, -X_1^{A, 2})]\\
	&=\dfrac{1}{4}\sum_{\mathbf{x}\in \{-1, 1\}^2 } (-\pi_1(\mathbf{x}) + \pi_2(\mathbf{x})).
\end{align*}

If Team A plays $\tilde{\mu}^A$, then the total reward of Team A is
\begin{align*}
&\quad ~J^A(\tilde{\mu}^A, \sigma^B)\\
&= 1 - \E[\pi_1(-X_1^{A, 1}, X_1^{A, 2})  + 1 - \pi_2(-X_1^{A, 1}, X_1^{A, 2})]\\
&=\dfrac{1}{4}\sum_{\mathbf{x}\in \{-1, 1\}^2 } (\pi_1(\mathbf{x}) - \pi_2(\mathbf{x})).
\end{align*}

Observe that $J^A(\mu^A, \sigma^B) + J^A(\tilde{\mu}^A, \sigma^B) = 0$. Hence for any $\sigma^B$, either $J^A(\mu^A, \sigma^B)\geq 0$ or $J^A(\tilde{\mu}^A, \sigma^B)\geq 0$. In particular, we can conclude that Team A's total reward is at least 0 in any equilibrium.

We have established both an upper bound and lower bound for Team A's total reward in an equilibrium. Hence we must have
\begin{equation}
	-q_1q_2-(1-q_1)(1-q_2) = 0,
\end{equation}
which implies $q_1=0, q_2=1$ or $q_1=1, q_2=0$. The former case corresponds to Team A playing the pure strategy $\mu^A$, and the latter to playing $\tilde{\mu}^A$.

\textbf{Step 2:} There does not exist equilibria where Team A plays $\mu^A$ or $\tilde{\mu}^A$.

Suppose that Team A plays $\mu^A$. Then the only best response of Team B is to play $U_2^{B, 1} = U_1^{A, 1}, U_2^{B, 2} = -U_1^{A, 2}$. Then, Team A's total reward is $J^A(\mu^A, \sigma^B) = 1 - 1 - 1 = -1$. If Team A deviate to $\tilde{\mu}^A$, then Team A can obtain a total reward of $+1$ (remember that $J^A(\mu^A, \sigma^B) + J^A(\tilde{\mu}^A, \sigma^B) = 0$ for any $\sigma^B$). Hence Team A does not play $\mu^A$ at equilibrium.

Similar arguments apply to $\tilde{\mu}^A$, which completes the proof.

\subsection{Proof of Lemma \ref{lem: pure2coord}}\label{app: pure2coord}
Given a pure strategy profile $\mu$, define a pure coordination strategy profile $\nu$ by
\begin{align*}
	\nu_t^{i}(h_t^i, \gamma_{1:t-1}^i) &= (\mu_t^{i, j}(h_t^i, \cdot))_{(i, j)\in\mathcal{N}_i}\\
	&\quad \forall h_t^i\in \mathcal{H}_t^i, \gamma_{1:t-1}^i\in \prescription_{1:t-1}^i, \forall i\in \mathcal{I}.
\end{align*}

We first prove one side of the result by coupling two systems, i.e. for every pure strategy profile $\mu$, there exist an equivalent coordination strategy profile $\nu$. In one of the systems, we assume that pure strategies are used. In the other system, we assume that the corresponding pure coordination strategies are used. The realizations of primitive random variables (i.e. $(X_1^i)_{i\in\mathcal{I}}, (W_t^{i, X}, W_t^{i, Y})_{i\in\mathcal{I}, t\in\mathcal{T}}$) are assumed to be the same for two systems. We proceed to show that the realizations of all system variables (i.e. $(\bX_t, \bY_t, \bU_t)_{t\in\mathcal{T}}$) will be the same for both systems. As a result, the expected payoffs are the same for both systems. The other direction can be proved analogously.

We prove that the realizations of $(\bX_t, \bY_t, \bU_t)_{t\in\mathcal{T}}$ are the same by induction on time $t$.

\textbf{Induction Base}: At $t=1$, the realizations of $\bX_1$ are the same for two systems by assumption. For the first system we have
\begin{equation}
U_1^{i, j} = \mu_1^{i, j}(X_1^{i, j}),
\end{equation}
and for the second system we have
\begin{align*}
\bm{\Gamma}_1^{i} &= \nu_t^i(H_1^{i}) = (\mu_t^{i, j}(\cdot))_{(i, j)\in\mathcal{N}_i}, \\
U_1^{i, j} &= \Gamma_1^{i, j}(X_1^{i, j}),
\end{align*}
which means that $U_1^{i, j} = \mu_1^i(X_1^{i, j})$ also holds in the second system. 

Since $(W_1^{i, Y})_{i\in\mathcal{I}}$ are the same for both systems, $Y_1^i=\ell_1^i(X_1^i, \bU_1, W_1^{i, Y})$ are the same for both systems.

\textbf{Induction Step}: Suppose that $\bX_s, \bY_s, \bU_s$ are the same for both systems for all $s<t$. Now we prove it for $t$.

First, since the realizations of $\bX_{t-1}^i, \bU_{t-1}, W_{t-1}^{i, X}$ are the same, we have
\begin{equation}
\bX_t^i = f_t^i(\bX_{t-1}^i, \bU_{t-1}, W_{t-1}^{i, X})
\end{equation}
to be the same for both systems.

Consider the actions. For the first system
\begin{equation}
	U_t^{i, j} = \mu_t^{i, j}(H_t^{i, j}) = \mu_t^{i, j}(H_t^{i}, X_{t-d+1:t}^{i, j}).
\end{equation}

In the second system
\begin{align*}
\bm{\Gamma}_t^{i} &= \nu_t^i(H_t^{i}) = (\mu_t^{i, j}(H_t^i, \cdot))_{(i, j)\in\mathcal{N}_i} \\
U_t^{i, j} &= \Gamma_t^{i, j}(X_{t-d+1:t}^{i, j}),
\end{align*}
which means that
\begin{equation}
	U_t^{i, j} = \mu_t^{i, j}(H_t^i, X_{t-d+1:t}^{i, j}).
\end{equation}

We conclude that $\bU_t$ has the same realization for two systems since $(H_t^i, X_{t-d+1:t}^{i, j})$ have the same realization by the induction hypothesis and the argument above. Since $(W_t^{i, Y})_{i\in\mathcal{I}}$ are the same for both systems, $Y_t^i=\ell_t^i(X_t^i, \bU_t, W_t^{i, Y})$ are same for both systems.

Therefore we have established the induction step, proving that for every pure strategy profile $\mu$ there exists an equivalent coordination strategy profile $\nu$.

To complete the other half of the proof, for each given coordination strategy $\nu$ we define
\begin{equation}\label{nu2mu1}
	\mu_t^{i, j}(h_t^{i, j}) = \gamma_{t}^{i, j}(x_{t-d+1:t}^{i, j})\quad\forall h_t^{i, j}\in\mathcal{H}_t^{i, j},
\end{equation}
where $\gamma_{t}^i=(\gamma_{t}^{i, j})_{(i, j) \in \mathcal{N}_i}$ is recursively defined by $\nu_{1:t}^i$ and $h_t^i$ through
\begin{equation}\label{nu2mu2}
	\gamma_{t}^i = \nu_t^i(h_t^i, \gamma_{1:t-1}^i)\quad\forall t\in \mathcal{T}.
\end{equation}

Then using a similar argument we can show that $\mu$ is equivalent to $\nu$.
\subsection{Proof of Lemma \ref{lem: condindep}}\label{app: condindep}
	Induction on time $t$.
	
	\textbf{Induction Base:} At $t=1$, we have $\bX_{1}^k$ to be independent for different $k$ because of the assumption on primitive random variables. Furthermore, since $H_1^k$ is a deterministic random vector (see Remark \ref{remark: initial}) and the randomization of different coordinators are independent, we conclude that $(\bX_{1}^k, \bm{\Gamma}_{1}^k)$ are mutually independent for different $k$. The distribution of $(\bX_{1}^k, \bm{\Gamma}_{1}^k)$ depends on $g$ only through $g^k$.
	
	\textbf{Induction Step:} Suppose that $(\bX_{1:t}^k, \bm{\Gamma}_{1:t}^k)$ are conditionally independent given $H_t^0$ and $\Pr^g(\bX_{1:t}^k, \bm{\Gamma}_{1:t}^k|H_t^0)$ depends on $g$ only through $g^k$. Now, we have
	\begin{align*}
	&\quad~\Pr^g(x_{1:t+1}, \gamma_{1:t+1}|h_{t+1}^0)\\
	&=\Pr^g(x_{t+1}|h_{t+1}^0, x_{1:t}, \gamma_{1:t+1})\times\Pr^g(\gamma_{t+1}|h_{t+1}^0, x_{1:t}, \gamma_{1:t})\times\\
	&\quad \times \Pr^g(x_{1:t}, \gamma_{1:t}|h_{t+1}^0)\\
	&=\left(\prod_{k\in\mathcal{I}} \Pr(x_{t+1}^k|x_t^k, u_t)g_{t+1}^k(\gamma_{t+1}^k|h_{t+1}^0, x_{1:t-d+1}^k, \gamma_{1:t}^k)\right) \times\\
	&\quad\times \Pr^g(x_{1:t}, \gamma_{1:t}|h_{t+1}^0).
	\end{align*}
	
	We then claim that
	\begin{equation}
	\Pr^g(x_{1:t}, \gamma_{1:t}, y_t, u_t|h_{t}^0) = \prod_{k\in \mathcal{I}} F_t^k(x_{1:t}^k, \gamma_{1:t}^k, h_{t+1}^0)
	\end{equation}
	where for each $k\in\mathcal{I}$, $F_t^k$ is a function that depends only on $g^k$.
	
	To establish the claim we note that
	\begin{align*}
	&\quad~\Pr^g(x_{1:t}, \gamma_{1:t}, y_t, u_t|h_{t}^0) \\
	&= \Pr^g(y_t, u_t|h_{t}^0, x_{1:t}, \gamma_{1:t}) \Pr^g(x_{1:t}, \gamma_{1:t}|h_t^0)\\
	&=\left(\prod_{k\in\mathcal{I}}\Pr(y_t^k|x_t^k, u_t) \bm{1}_{\{u_{t}^k = \gamma_{t}^k(x_{t-d+1:t}^k)\}}\right) \Pr^g(x_{1:t}, \gamma_{1:t}|h_t^0)\\
	&=\left(\prod_{k\in\mathcal{I}}\Pr(y_t^k|x_t^k, u_t) \bm{1}_{\{u_{t}^k = \gamma_{t}^k(x_{t-d+1:t}^k)\}}\right)\times \\ &\quad\times \left(\prod_{k\in\mathcal{I}}\Pr^{g_k}(x_{1:t}^k, \gamma_{1:t}^k|h_t^0) \right)\\
	&=\prod_{k\in\mathcal{I}} F_t^k(x_{1:t}^k, \gamma_{1:t}^k, h_{t+1}^0),
	\end{align*}
	where in the third step we have used the induction hypothesis. 
	
	Given the claim, we have
	\begin{align*}
	&\quad~\Pr^g(x_{1:t}, \gamma_{1:t}|h_{t+1}^0)\\
	&=\dfrac{\Pr^g(x_{1:t}, \gamma_{1:t}, y_t, u_t|h_{t}^0)}{\sum_{\tilde{x}_{1:t}, \tilde{\gamma}_{1:t}} \Pr^g(\tilde{x}_{1:t}, \tilde{\gamma}_{1:t}, y_t, u_t|h_{t}^0)}\\
	&=\dfrac{\prod_{k\in \mathcal{I}} F_t^k(x_{1:t}^k, \gamma_{1:t}^k, h_{t+1}^0)}{\sum_{\tilde{x}_{1:t}, \tilde{\gamma}_{1:t}}\prod_{k\in \mathcal{I}} F_t^k(\tilde{x}_{1:t}^k, \tilde{\gamma}_{1:t}^k, h_{t+1}^0)}\\
	&=\dfrac{\prod_{k\in \mathcal{I}} F_t^k(x_{1:t}^k, \gamma_{1:t}^k, h_{t+1}^0)}{\prod_{k\in \mathcal{I}} \left(\sum_{\tilde{x}_{1:t}^k, \tilde{\gamma}_{1:t}^k}F_t^k(\tilde{x}_{1:t}^k, \tilde{\gamma}_{1:t}^k, h_{t+1}^0)\right)}\\
	&=\prod_{k\in\mathcal{I}} \left( \dfrac{ F_t^k(x_{1:t}^k, \gamma_{1:t}^k, h_{t+1}^0)}{\sum_{\tilde{x}_{1:t}^k, \tilde{\gamma}_{1:t}^k}F_t^k(\tilde{x}_{1:t}^k, \tilde{\gamma}_{1:t}^k, h_{t+1}^0)}\right)
	\end{align*}
	and then
	\begin{equation}
	\Pr^g(x_{1:t+1}, \gamma_{1:t+1}|h_{t+1}^0) = \prod_{k\in\mathcal{I}} G_t^k(x_{1:t+1}^k, \gamma_{1:t+1}^k, h_{t+1}^0),
	\end{equation}
	where $G_t^k$ is given by
	\begin{align}
		&\quad~G_t^k(x_{1:t+1}^k, \gamma_{1:t+1}^k, h_{t+1}^0)\\
		&=\Pr(x_{t+1}^k|x_t^k, u_t)g_{t+1}^k(\gamma_{t+1}^k|h_{t+1}^0, x_{1:t-d+1}^k, \gamma_{1:t}^k) \times\\
		&\quad\times\dfrac{ F_t^k(x_{1:t}^k, \gamma_{1:t}^k, h_{t+1}^0)}{\sum_{\tilde{x}_{1:t}^k, \tilde{\gamma}_{1:t}^k}F_t^k(\tilde{x}_{1:t}^k, \tilde{\gamma}_{1:t}^k, h_{t+1}^0)}.
	\end{align}
	
	One can check that $G_t^k$ depends on $g$ only through $g^k$ and $\sum_{\tilde{x}_{1:t+1}^k, \tilde{\gamma}_{1:t+1}^k} G_t^k(\tilde{x}_{1:t+1}^k, \tilde{\gamma}_{1:t+1}^k, h_{t+1}^0) = 1$, therefore
	\begin{equation}
		G_t^k(x_{1:t+1}^k, \gamma_{1:t+1}^k, h_{t+1}^0) = \Pr^{g^k}(x_{1:t+1}^k, \gamma_{1:t+1}^k| h_{t+1}^0).
	\end{equation}
	Hence we establish the induction step.


\subsection{Proof of Lemma \ref{lem: selfbelief}}\label{app: selfbelief}
Assume that $\overline{h}_t^i\in\overline{\mathcal{H}}_t^i$ is admissible under ${g}$. From Lemma \ref{lem: condindep}, we know that $\Pr^{g} (x_{1:t}^i, \gamma_{1:t}^i|h_t^0)$ does not depend on $g^{-i}$. As a conditional distribution obtained from $\Pr^{g} (x_{1:t}^i, \gamma_{1:t}^i|h_t^0)$, $\Pr^{g}(x_{t-d+1:t}^i|\overline{h}_t^i)$ does not depend on $g^{-i}$ either.

Therefore, we can compute the belief of coordinator $i$ by replacing $g^{-i}$ with $\hat{g}^{-i}$, which is an open-loop strategy profile that always generates the actions $u_{1:t-1}^{-i}$. 
\begin{align}
&\Pr^{g^i, g^{-i}}(x_{t-d+1:t}^i|\overline{h}_t^i)=\Pr^{g^i, \hat{g}^{-i}}(x_{t-d+1:t}^i|\overline{h}_t^i).
\end{align}
Note that we always have $\Pr^{g^i, \hat{g}^{-i}}(\overline{h}_t^i) > 0$ for all $\overline{h}_t^i$ admissible under $g$.

Furthermore, we can also introduce additional random variables into the condition that are conditionally independent according to Lemma \ref{lem: condindep}, i.e.
\begin{align}
	&\Pr^{g^i, \hat{g}^{-i}}(x_{t-d+1:t}^i|\overline{h}_t^i)=\Pr^{g^i, \hat{g}^{-i}}(x_{t-d+1:t}^i|\overline{h}_t^i, x_{t-d:t}^{-i}),
\end{align}
where $x_{t-d:t}^{-i}\in\mathcal{X}_{t-d:t}^{-i}$ is such that $\Pr^{g^i, \hat{g}^{-i}}(x_{t-d:t}^{-i}|\overline{h}_t^i) > 0$.

Let $\tau = t-d+1$. By Bayes' rule
\begin{align}
&\quad~\Pr^{g^i, \hat{g}^{-i}}(x_{\tau:t}^i|\overline{h}_t^i, x_{\tau-1:t}^{-i})\\
&=\dfrac{\Pr^{g^i, \hat{g}^{-i}}(x_{\tau:t}, y_{\tau:t-1}, u_{\tau:t-1}, \gamma_{\tau:t-1}^i|h_{\tau}^{*i} )}{ \sum_{\tilde{x}_{\tau:t}^i} \Pr^{g^i, \hat{g}^{-i}}(\tilde{x}_{\tau:t}^i, x_{\tau:t}^{-i}, y_{\tau:t-1}, u_{\tau:t-1}, \gamma_{\tau:t-1}^i|h_{\tau}^{*i} )},\label{eq: selfbeliefproof1}
\end{align}
where
\begin{align}
	h_{\tau}^{*i} = (y_{1:\tau-1}, u_{1:\tau-1}, x_{1:\tau-1}^i, x_{\tau-1}^{-i}, \gamma_{1:\tau-1}^i).
\end{align}

We have
\begin{align}
	&\quad~\Pr^{g^i, \hat{g}^{-i}}(x_{\tau:t}, y_{\tau:t-1}, u_{\tau:t-1}, \gamma_{\tau:t-1}^i|h_{\tau}^{*i} )=\prod_{l=1}^{d-1} \Big[\\
	&\quad~\Pr^{g^i, \hat{g}^{-i}}(x_{t-l+1}, y_{t-l}|h_\tau^{*i}, x_{\tau:t-l}, y_{\tau:t-l-1}, u_{\tau:t-l}, \gamma_{\tau:t-l}^i)\times\\
	&\times \Pr^{g^i, \hat{g}^{-i}}(u_{t-l}^i|h_\tau^{*i}, x_{\tau:t-l}, y_{\tau:t-l-1}, u_{\tau:t-l-1}, \gamma_{\tau:t-l}^i)\times\\
	&\times \Pr^{g^i, \hat{g}^{-i}}(\gamma_{t-l}^i|h_\tau^{*i}, x_{\tau:t-l}, y_{\tau:t-l-1}, u_{\tau:t-l-1}, \gamma_{\tau:t-l-1}^i)\Big]\times\\
	&\times \Pr^{g^i, \hat{g}^{-i}}(x_{\tau}|h_{\tau}^{*i}).\label{eq: selfbeliefproof2}
\end{align}

The first three terms in the above product are
\begin{align}
	&\quad~\Pr^{g^i, \hat{g}^{-i}}(x_{t-l+1}, y_{t-l}|h_\tau^{*i}, x_{\tau:t-l}, y_{\tau:t-l-1}, u_{\tau:t-l}, \gamma_{\tau:t-l}^i)\!\\
	&=\prod_{k\in\mathcal{I}}[\Pr(x_{t-l+1}^k|x_{t-l}^k, u_{t-l})\Pr(y_{t-l}^k|x_{t-l}^k, u_{t-l}) ],\label{eq: selfbeliefproof3}\\
	&\quad~\Pr^{g^i, \hat{g}^{-i}}(u_{t-l}^i|h_\tau^{*i}, x_{\tau:t-l}, y_{\tau:t-l-1}, u_{\tau:t-l-1}, \gamma_{\tau:t-l}^i)\\
	&=\prod_{(i, j)\in \mathcal{N}_i }\bm{1}_{\{u_{t-l}^{i, j} = \gamma_{t-l}^{i, j}(x_{\tau-l:t-l}^{i, j}) \} } \\
	&= \prod_{(i, j)\in \mathcal{N}_i } \bm{1}_{\{u_{t-l}^{i, j} = \phi_{t-l, l}^i(x_{\tau:t-l}^i) \} },\label{eq: selfbeliefproof4}\\
	&\quad~\Pr^{g^i, \hat{g}^{-i}}(\gamma_{t-l}^i|h_\tau^{*i}, x_{\tau:t-l}, y_{\tau:t-l-1}, u_{\tau:t-l-1}, \gamma_{\tau:t-l-1}^i)\\
	&=g_{t-l}^i(\gamma_{t-l}^i|y_{1:t-l-1}, u_{1:t-l-1}, x_{1:t-l-d}^i, \gamma_{1:t-l-1}^i),\quad\label{eq: selfbeliefproof5}
\end{align}
respectively.

The last term satisfies
\begin{align}
	&\Pr^{g^i, \hat{g}^{-i}}(x_{\tau}|h_{\tau}^{*i}) = \prod_{k\in\mathcal{I}}\Pr(x_{\tau}^k|x_{\tau-1}^k, u_{\tau-1}).
\end{align}

Substituting \eqref{eq: selfbeliefproof2} - \eqref{eq: selfbeliefproof5} into \eqref{eq: selfbeliefproof1} we obtain
\begin{align}
	&\quad~\Pr^{g^i, \hat{g}^{-i}}(x_{\tau:t}^i|\overline{h}_t^i, x_{\tau-1:t}^{-i})\\
	&=\dfrac{F_t^i(x_{\tau:t}^i, y_{\tau:t-1}^i, u_{\tau-1:t-1}, x_{\tau-1}^i, \phi_{t}^i)}{\sum_{\tilde{x}_{\tau:t}^i } F_t^i(\tilde{x}_{\tau:t}^i, y_{\tau:t-1}^i, u_{\tau-1:t-1}, x_{\tau-1}^i, \phi_{t}^i)}
\end{align}
where
\begin{align}
	&\quad~F_t^i(x_{\tau:t}^i, y_{\tau:t-1}^i, u_{\tau-1:t-1}, \phi_{t}^i)\\
	&:= \Pr(x_{\tau}^i|x_{\tau-1}^i, u_{\tau-1})\prod_{l=1}^{d-1}\Big[\Pr(x_{t-l+1}^i|x_{t-l}^i, u_{t-l})\times\\
	&\times\Pr(y_{t-l}^i|x_{t-l}^i, u_{t-l})  \left(\prod_{(i, j)\in \mathcal{N}_i } \bm{1}_{\{u_{t-l}^{i, j} = \phi_{t-l, l}^{i, j}(x_{\tau:t-l}^{i, j}) \} }\right) \Big] 
\end{align}

Therefore we have proved that
\begin{align}
	&\quad~\Pr^{g}(x_{t-d+1:t}^i|\overline{h}_t^i)\\
	&=P_t^i(x_{t-d+1:t}^i|y_{t-d+1:t-1}^i, u_{t-d:t-1}, x_{t-d}^i, \phi_{t}^i)\\
	&:=\dfrac{F_t^i(x_{t-d+1:t}^i, y_{t-d+1:t-1}^i, u_{t-d:t-1}, x_{t-d}^i, \phi_{t}^i)}{\sum_{\tilde{x}_{t-d+1:t}^i } F_t^i(\tilde{x}_{t-d+1:t}^i, y_{t-d+1:t-1}^i, u_{t-d:t-1}, x_{t-d}^i, \phi_{t}^i)}
\end{align}
where $P_t^i$ is independent of $g$.

\subsection{Proof of Lemma \ref{lem: suffprivated3}}\label{app: suffprivated3}
Let $\tilde{g}^i$ denote coordinator $i$'s behavioral coordination strategy. Because of Lemma \ref{lem: condindep} we have
\begin{align*}
&~\quad \Pr^{\tilde{g}^i, g^{-i}}(x_{t-d+1:t}, \gamma_{t}^{-i}|\overline{h}_t^i, \gamma_t^i)\\
&=\Pr^{\tilde{g}^i, g^{-i}}(x_{t-d+1:t}, \gamma_{t}^{-i}|h_t^0, x_{1:t-d}^i, \gamma_{1:t}^i)\\
&=\Pr^{\tilde{g}^i}(x_{t-d+1:t}^i|h_t^0, x_{1:t-d}^i, \gamma_{1:t}^i)\prod_{k\neq i} \Pr^{g^{k}}(x_{t-d+1:t}^k, \gamma_{t}^k|h_t^0).
\end{align*}

We know that $\bm{\Gamma}_{t}^i$ and $\bX_{t-d+1:t}^i$ are conditionally independent given $\overline{H}_t^i$ since $\bm{\Gamma}_{t}^i$ is chosen as a randomized function of $\overline{H}_t^i$ at a time when $\bX_{t-d+1:t}^i$ are already realized. Therefore,
\begin{align}
&\quad~\Pr^{\tilde{g}^i, g^{-i}}(x_{t-d+1:t}^i|h_t^0, x_{1:t-d}^i, \gamma_{1:t}^i)\\
&=\Pr^{\tilde{g}^i, g^{-i}}(x_{t-d+1:t}^i|h_t^0, x_{1:t-d}^i, \gamma_{1:t-1}^i)\\
&=P_t^i(x_{t-d:t}^i| y_{t-d+1:t-1}^i, u_{t-d:t-1}, x_{t-d}^i, \phi_t^i),
\end{align}
where $P_t^i$ is the belief function defined in Eq. \eqref{eq: belieffunction}.

We conclude that
\begin{align}
&~\quad\Pr^{\tilde{g}^i, g^{-i}}(x_{t-d+1:t}, \gamma_{t}^{-i}|h_t^0, x_{1:t-d}^i, \gamma_{1:t-1}^i)\\
&=F_t^i(x_{t-d+1:t}, \gamma_{t}^{-i}| h_t^0, x_{t-d}^i, \phi_t^i; g^{-i})\label{belieffunctiond2}
\end{align}
for some function $F_t^i$ that does not depend on $\tilde{g}^i$.

Consider the reward of coordinator $i$. By the law of iterated expectation we can write
\begin{align*}
J^i(\tilde{g}^i, g^{-i}) = \E^{\tilde{g}^i, g^{-i}}\left[\sum_{t\in\mathcal{T}} \E^{\tilde{g}^i, g^{-i}}[r_t^i(\bX_t, \bU_t)|\overline{H}_t^i, \bm{\Gamma}_t^i] \right].
\end{align*}

For each term we have
\begin{align}
&\quad~\E^{\tilde{g}^i, g^{-i}}[r_t^i(\bX_t, \bU_t)|\overline{h}_t^i, \gamma_t^i]\\
&= \sum_{\tilde{x}_{t-d+1:t}}\sum_{\tilde{\gamma}_t^{-i}} r_t^i(\tilde{x}_t, (\gamma_t^i(\tilde{x}_{t-d+1:t}^i), \tilde{\gamma}_t^{-i}(\tilde{x}_{t-d+1:t}^{-i}) )) \times\\
&\quad~\times F_t^i(\tilde{x}_{t-d+1:t}, \tilde{\gamma}_{t}^{-i}| h_t^0, x_{t-d}^i, \phi_{t}^i; g^{-i}) \\
&=: \overline{r}_t^{i}(h_t^0, x_{t-d}^i, \phi_t^i, \gamma_t^i; g^{-i})\label{eq: instacost}
\end{align}
where $F_t^i$ is the belief function described in \eqref{belieffunctiond2}, and $\overline{r}_t^{i}$ is a function that does not depend on $\tilde{g}^i$.

We claim that $(H_t^0, \bX_{t-d}^i, \bm{\Phi}_{t}^i)$ is a controlled Markov process controlled by coordinator $i$'s prescriptions fixing the other coordinators' strategies. We need to prove that
\begin{align}
&\quad~\Pr^{\tilde{g}^i, g^{-i}}(h_{t+1}^0, x_{t-d+1}^i, \phi_{t+1}^i |h_{1:t}^0, x_{1:t-d}^i, \phi_{1:t}^i, \gamma_{1:t}^i) \\
&= G_t^i(h_{t+1}^0, x_{t-d+1}^i, \phi_{t+1}^i| h_{t}^0, x_{t-d}^i, \phi_{t}^i, \gamma_{t}^i)
\end{align}
for some function $G_t^i$ independent of $\tilde{g}^i$.

We know that $H_{t+1}^0 = (H_t^0, \bY_t, \bU_t)$ and
\begin{align*}
Y_t^k&= \ell_t^k(\bX_t^k, \bU_t, W_t^{k, Y})\quad\forall k\in \mathcal{I},\\
U_t^{k, j}&= \Gamma_t^{k, j}(X_{t-d+1:t}^{k, j})\quad\forall (k, j)\in \mathcal{N},\\
\bm{\Phi}_{t+1}^i &= (\Phi_{t+1-s, s}^{i, j})_{(i, j)\in \mathcal{N}_i, 1\leq s\leq d-1 }  \\
\Phi_{t, 1}^{i, j}&= \Gamma_{t}^{i, j}(X_{t-d+1}^{i, j}, \cdot)\quad\forall (i, j)\in\mathcal{N}_i,\\
\Phi_{t+1-s, s}^{i, j}&=\Phi_{t+1-s, s-1}^{i, j}(X_{t-d+1}^{i, j}, \cdot)\quad\forall (i, j)\in\mathcal{N}_i, s\geq 2,
\end{align*}
hence $(H_{t+1}^0, \bX_{t-d+1}^i, \bm{\Phi}_{t+1}^i)$ is a function of $(H_t^0, \bX_{t-d+1:t}, \bm{\Gamma}_t, \bm{\Phi}_t^i)$, and $\mathbf{W}_t^{Y}$. As $\mathbf{W}_t^{Y}$ is a primitive random vector independent of $(H_{1:t}^0, \bX_{1:t-d}^i, \bm{\Phi}_{1:t}^i, \bm{\Gamma}_{1:t}^i)$, it suffices to prove that
\begin{equation}\label{aha1}
\begin{split}
&\quad~\Pr^{\tilde{g}^i, g^{-i}}(x_{t-d+1:t}, \gamma_{t}^{-i}|h_{1:t}^0, x_{1:t-d}^i, \phi_{1:t}^i, \gamma_{1:t}^i) \\
&= G_t^i(x_{t-d+1:t}, \gamma_{t}^{-i}| h_{t}^0, x_{t-d}^i, \phi_{t}^i, \gamma_{t}^i)
\end{split}
\end{equation}
for some function $G_t^i$ independent of $\tilde{g}^i$.

Since $(H_{1:t}^0, \bX_{1:t-d}^i, \bm{\Phi}_{1:t}^i, \bm{\Gamma}_{1:t}^i)$ is a function of $(\overline{H}_t^i, \bm{\Gamma}_{t}^i)$, applying smoothing property of conditional expectation to both sides of \eqref{belieffunctiond2} we obtain
\begin{align*}
&\quad~\Pr^{\tilde{g}^i, g^{-i}}(x_{t-d+1:t}, \gamma_{t}^{-i}|h_{1:t}^0, x_{1:t-d}^i, \phi_{1:t}^i, \gamma_{1:t}^i)\\
&= F_t^i(x_{t-d+1:t}, \gamma_{t}^{-i}| h_{t}^0, x_{t-d}^i, \phi_{t}^i; g^{-i}).
\end{align*}

Hence, we conclude that coordinator $i$ faces a Markov Decision Problem where the state process is $(H_t^0, \bX_{t-d}^i, \bm{\Phi}_{t}^i)$, the control action is $\bm{\Gamma}_t^i$, and the total reward is $$\E^{\tilde{g}^i, g^{-i}} \left[\sum_{t\in\mathcal{T}}\overline{r}_t^{i}(H_t^0, \bX_{t-d}^i, \bm{\Phi}_{t}^i, \bm{\Gamma}_t^i; g^{-i})\right].$$

By standard MDP theory, coordinator $i$ can form a best response by choosing $\bm{\Gamma}_t^i$ based on $(H_t^0, \bX_{t-d}^i, \bm{\Phi}_{t}^i)$.

\subsection{Proof of Theorem \ref{thm: spibexist}}\label{app: spibexist}
The idea to prove the theorem is to apply Kakutani's fixed point theorem on a special best response correspondence defined through Bellman equations.

Define $\varXi_t^i\subset \mathcal{H}_t^0 \times \mathcal{S}_{t}^i$ to be the set of admissible $(h_t^0, s_{t}^i)$'s, i.e. $(h_t^0, s_{t}^i)$'s with strictly positive probability under at least one strategy profile of the coordinators.

For $\epsilon \geq 0$, define $\mathcal{R}^{\epsilon, i}$ be the set of SPIB strategy profiles for coordinator $i$ where each prescription has probability at least $\epsilon$ to be chosen at any information set. Specifically, it suffices to consider the prescription choices for $(h_t^0, s_t^i)\in \varXi_t^i$ for each $t\in\mathcal{T}$, and we can write
\begin{equation}
\mathcal{R}^{\epsilon, i} = \prod_{t\in\mathcal{T}}\prod_{\xi_t^i\in \varXi_t^i} \Delta^\epsilon(\prescription_t^i)
\end{equation}
where 
\begin{equation}
\Delta^\epsilon(\prescription_t^i) = \{\eta \in \Delta(\prescription_t^i) : \eta(\gamma_t^i) \geq \epsilon ~~\forall \gamma_t^i\in\prescription_t^i \}.
\end{equation}

We also define $\mathcal{R}^{\epsilon} = \prod_{i\in\mathcal{I}} \mathcal{R}^{\epsilon, i}$. $\mathcal{R}^{0}$ is then the set of all SPIB strategy profiles.

Recall that in the proof of Lemma \ref{lem: suffprivated3}, we have shown that fixing a behavioral strategy coordination profile $g^{-i}$, coordinator $i$ faces an MDP problem with state $\Xi_t^i:=(H_t^0, S_t^i)$ and control action $\bm{\Gamma}_t^i$ and total reward
\begin{equation}
\E\left[\sum_{t\in\mathcal{T}}\overline{r}_t^{i}(H_t^0, S_t^i, \bm{\Gamma}_t^i; g^{-i})\right],
\end{equation}
where $\overline{r}_t^{i}$ is defined in \ref{eq: instacost}.

With some abuse of notation, let $\overline{r}_t^{i}(\Xi_t^i, \bm{\Gamma}_t^i; \rho^{-i})$ denote the instantaneous cost when all coordinators except $i$ play SPIB strategy profile $\rho^{-i}$.

Hence we can define a subset of the best response correspondence through the following construction: For each $\xi_t^i\in \varXi_t^i$, define the correspondence $\mathrm{BR}_t^{\epsilon, i}[\xi_t^i]: \mathcal{R}^{\epsilon, -i} \mapsto \Delta^\epsilon(\prescription_t^i)$ sequentially through
\begin{align*}
Q_{T}^{\epsilon, i}(\xi_T^i, \gamma_T^i; \rho^{-i}) := \overline{r}_{T}^{i}(\xi_T^i, \gamma_T^i; \rho^{-i}),
\end{align*}
and for each $t\in\mathcal{T}$ and each $\xi_t^i\in\varXi_t^i$,
\begin{align*}
\mathrm{BR}_t^{\epsilon, i}[\xi_{t}^i](\rho^{-i}) &:= \underset{\eta \in \Delta^\epsilon(\prescription_t^i)}{\arg\max} \sum_{\gamma_t^i} \eta(\gamma_t^i) Q_t^{\epsilon, i}(\xi_{t}^i, \gamma_t^i; \rho^{-i}),\\
V_t^{\epsilon, i}(\xi_t^i; \rho^{-i}) &:= \max_{\eta \in \Delta^\epsilon(\prescription_t^i)} \sum_{\gamma_t^i} \eta(\gamma_t^i) Q_t^{\epsilon, i}(\xi_{t}^i, \gamma_t^i; \rho^{-i}),\\
Q_{t-1}^{\epsilon, i}(\xi_{t-1}^i, \gamma_{t-1}^i; \rho^{-i}) &:= \overline{r}_{t-1}^{i}(\xi_{t-1}^i, \gamma_{t-1}^i; \rho^{-i}) \\
& + \sum_{\xi_t^i} V_t^{\epsilon, i}(\xi_t^i; \rho^{-i})\Pr^{\rho^{-i}}(\xi_t^i| \xi_{t-1}^i, \gamma_t^i).
\end{align*}

Define $\mathrm{BR}^{\epsilon}: \mathcal{R}^\epsilon \mapsto \mathcal{R}^\epsilon$ by
\begin{align}
&\quad~\mathrm{BR}^{\epsilon}(\rho) \\
&= \{\tilde{g}\in \mathcal{R}^\epsilon: \tilde{g}_t^i(\xi_t^i) \in \mathrm{BR}^{\epsilon, i}[\xi_{t}^i](\rho^{-i})~~\forall \xi_t^i\in\varXi_t^i, \forall i\in\mathcal{I}  \}\\
&= \prod_{i\in\mathcal{I}} \prod_{t\in\mathcal{T}} \prod_{\xi_t^i\in \varXi_t^i} \mathrm{BR}_t^{\epsilon, i}[\xi_{t}^i](\rho^{-i}).
\end{align}

\textbf{Claim}: 
\begin{enumerate}[(a)]
	\item $\overline{r}_{t}^{i}(\xi_{t}^i, \gamma_{t}^i; \rho^{-i})$ is continuous in $\rho^{-i}$ on $\mathcal{R}^{\epsilon, -i}$ for all $t\in\mathcal{T}$ and all $\xi_t^i\in\varXi_t^i, \gamma_t^i\in\prescription_{t}^i$
	\item $\Pr^{\rho^{-i}}(\xi_{t+1}^i|\xi_{t}^i, \gamma_{t}^i)$ is continuous in $\rho^{-i}$ on $\mathcal{R}^{\epsilon, -i}$ for all $t\in \mathcal{T}\backslash\{T\}$ and all $\xi_{t+1}^i\in\varXi_{t+1}^i, \xi_{t}^i\in\varXi_{t}^i, \gamma_t^i\in\prescription_{t}^i$.
\end{enumerate}

Given the claims, we prove by induction that $Q_t^{\epsilon, i}(\xi_t^i, \gamma_t^i; \cdot)$ is continuous on $\mathcal{R}^{\epsilon, -i}$ for each $\xi_t^i\in \varXi_t^i$ and $\gamma_t^i\in \prescription_t^i$.

\textbf{Induction Base}: $Q_{T}^{\epsilon, i}(\xi_{T}^i, \gamma_{T}^i; \cdot)$ is continuous on $\mathcal{R}^{\epsilon, -i}$ since $\overline{r}_{T}^{i}(\xi_T^i, \gamma_{T}^i; \rho^{-i})$ is continuous in $\rho^{-i}$ on $\mathcal{R}^{\epsilon, -i}$ for all $\xi_T^i\in\varXi_T^i$ and all $\gamma_T^i\in \varGamma_T^i$.

\textbf{Induction Step}: Suppose that the induction hypothesis is true for $t$. Then $V_t^{\epsilon, i}(\xi_t^i; \cdot)$ is continuous on $\mathcal{R}^{\epsilon, -i}$ due to Berge's Maximum Theorem. Then for all $\xi_{t-1}^i\in \varXi_{t-1}^i$ and $\gamma_{t-1}^i\in \prescription_{t-1}^i$, $Q_{t-1}^{\epsilon, i}(\xi_{t-1}^i, \gamma_{t-1}^i; \cdot)$ is continuous on $\mathcal{R}^{\epsilon, -i}$ since $\overline{r}_{t-1}^{i}(\xi_{t-1}^i, \gamma_{t-1}^i; \rho^{-i})$ is continuous in $\rho^{-i}$ on $\mathcal{R}^{\epsilon, -i}$, and the transition probability $\Pr^{\rho^{-i}}(\xi_t^i|\xi_{t-1}^i, \gamma_{t-1}^i)$ is also continuous in $\rho^{-i}$ on $\mathcal{R}^{\epsilon, -i}$.\\

Because of Berge's Maximum Theorem, we conclude that $\mathrm{BR}^{\epsilon, i}[\xi_{t}^i]$ is upper hemicontinuous on $\mathcal{R}^{\epsilon, -i}$ for each $\xi_t^i\in\varXi_t^i$. $\mathrm{BR}^{\epsilon, i}[\xi_{t}^i](\rho^{-i})$ is also non-empty and convex for each $\rho^{-i}\in \mathcal{R}^{\epsilon, -i}$ since it is a solution set of a linear program.

As a product of compact-valued upper hemicontinuous correspondences, we know that $\mathrm{BR}^{\epsilon}$ is upper hemicontinuous. Furthermore, $\mathrm{BR}^{\epsilon}(\rho)$ is non-empty and convex for each $\rho\in \mathcal{R}^\epsilon$. By Kakutani's fixed point theorem, $\mathrm{BR}^{\epsilon}$ has a fixed point.

Let $\epsilon_n \searrow 0$. Let $\rho^{(n)}\in \mathcal{R}^{\epsilon_n}$ be a fixed point of $\mathrm{BR}^{\epsilon_n}$. Then for each $i\in\mathcal{I}$ we have
\begin{equation}
\rho^{(n), i}\in \underset{\rho^i\in\mathcal{R}^{\epsilon_n, i}}{\arg\max}~ J^i(\rho^i, \rho^{(n), -i})
\end{equation}
where
\begin{equation}
J^i(\rho) = \E^{\rho}\left[\sum_{t\in\mathcal{T}} r_t^i(\bX_t, \bU_t) \right] = \E^{\rho^i}\left[\sum_{t\in\mathcal{T}}\overline{r}_t^{i}(S_t^i, \bm{\Gamma}_t^i; \rho^{-i})\right].
\end{equation}

Let $\rho^{(\infty)}\in \mathcal{R}^0$ be the limit of some sub-sequence of $(\rho^{(n)})_{n\in\mathbb{N}}$. Since $J^i(\cdot)$ is continuous on $\mathcal{R}^0$ and $\epsilon \mapsto \mathcal{R}^{\epsilon, i}$ is a continuous correspondence with compact, non-empty value (for small enough $\epsilon$), by Berge's Maximum Theorem, we conclude that for each $i$,
\begin{equation}
\rho^{(\infty), i}\in \underset{\rho^i\in\mathcal{R}^{0, i}}{\arg\max}~ J^i(\rho^i, \rho^{(\infty), -i}).
\end{equation}
i.e. $\rho^{(\infty), i}$ is one of the optimal strategies among SPIB strategies to respond to $\rho^{(\infty), -i}$. Combining with Lemma \ref{lem: suffprivated3} which states that there always exists best response strategies that are SPIB strategies, we conclude that $\rho^{(\infty)}$ forms a CNE, proving the result.

\begin{proof}[Proof of Claim]
	We first notice that, by the proof of Lemma \ref{lem: suffprivated3}, both $\overline{r}_{t}^{i}(\xi_{t}^i, \gamma_{t}^i; \rho^{-i})$ and $\Pr^{\rho^{-i}}(\xi_{t+1}^i|\xi_{t}^i, \gamma_{t}^i)$ are linear functions of $F_t^i(x_{t-d+1:t}, \gamma_t^{-i}|\xi_t^i; \rho^{-i})$ (defined in \eqref{belieffunctiond2}). We have
	\begin{align}
		&\quad~F_t^i(x_{t-d+1:t}, \gamma_t^{-i}|\xi_t^i; \rho^{-i})\\
		&= \Pr^{\rho^{-i}}(x_{t-d+1:t}, \gamma_t^{-i}|\xi_t^i)
		= \dfrac{\Pr^{\hat{\rho}^i, \rho^{-i}}(x_{t-d+1:t}, \gamma_t^{-i}, \xi_t^i) }{\Pr^{\hat{\rho}^i, \rho^{-i}}(\xi_t^i) }
	\end{align}
	where $\hat{\rho}^i \in \mathcal{R}^{\epsilon, i}$ is a fixed, arbitrary SPIB strategy. We know that both $\Pr^{\hat{\rho}^i, \rho^{-i}}(x_{t-d+1:t}, \gamma_t^{-i}, \xi_t^i)$ and $\Pr^{\hat{\rho}^i, \rho^{-i}}(\xi_t^i)$ are sums of products of components of $\rho^{-i}$ and $\hat{\rho}^i$, hence both are continuous in $\rho^{-i}$. Furthermore, we have $\Pr^{\hat{\rho}^i, \rho^{-i}}(\xi_t^i) > 0$ for all $\rho^{-i}\in\mathcal{R}^{\epsilon, -i}$ since $\xi_t^i\in \varXi_t^i$ has strictly positive probability under some strategy profile, and $(\hat{\rho}^i, \rho^{-i})$ is a strategy profile that chooses strictly mixed prescriptions. Therefore $F_t^i(x_{t-d+1:t}, \gamma_t^{-i}|\xi_t^i; \rho^{-i})$ is continuous in $\rho^{-i}$ on $\mathcal{R}^{\epsilon, -i}$.
\end{proof}

\subsection{Proof of Lemma \ref{lem: piistruebelief0}}\label{app: piistruebelief0}
We will prove a stronger result which we need in the proof of Proposition \ref{prop: pbe}.

\begin{lemma}\label{lem: piistruebelief}
	Let $(\lambda^{*k}, \psi^*)$ be a CIB strategy such that $\psi^{*, k}$ is consistent with $\lambda^{*k}$. Let $g^{*k}$ be the behavioral strategy profile generated from $(\lambda^{*k}, \psi^*)$. Let $\pi_t^k$ represent the belief on $S_t^k$ generated by $\psi^*$ at time $t$ based on $h_t^0$. Let $t< \tau$. Consider a fixed $h_{\tau}^0\in \mathcal{H}_{\tau}^0$ and some $\tilde{g}_{1:t-1}^k$ (not necessarily equal to $g_{1:t-1}^{*k}$). Assume that $h_{\tau}^0$ is admissible under $(\tilde{g}_{1:t-1}^{k}, g_{t:\tau-1}^{*k})$.
	Suppose that 
	\begin{align}
	&\quad~\Pr^{\tilde{g}_{1:t-1}^k}(s_t^k, x_{t-d+1:t}^k|h_t^0) \\
	&=\pi_t^k(s_t^k)P_t^k(x_{t-d+1:t}^k|y_{t-d+1:t-1}^k, u_{t-d:t-1}, s_t^k) \\
	&\qquad\qquad \forall s_t^k\in\mathcal{S}_t^k ~\forall x_{t-d+1:t}^k\in\mathcal{X}_{t-d+1:t}^k.\label{beliefequation}
	\end{align}
Then
	\begin{align}
	&\quad~\Pr^{\tilde{g}_{1:t-1}^k, g_{t:\tau-1}^{*k}}(s_\tau^k, x_{\tau-d+1:\tau}^k|h_\tau^0) \\
	&=\pi_\tau^k(s_\tau^k)P_\tau^k(x_{\tau-d+1:\tau}^k|y_{\tau-d+1:\tau-1}^k, u_{\tau-d:\tau-1}, s_\tau^k) \\
	&\qquad\qquad \forall s_\tau^k\in\mathcal{S}_\tau^k~ \forall x_{\tau-d+1:\tau}^k\in\mathcal{X}_{\tau-d+1:\tau}^k.
	\end{align}
\end{lemma}

The assertion of Lemma \ref{lem: piistruebelief0} follows from Lemma \ref{lem: piistruebelief} and the fact that \eqref{beliefequation} is true for $t=1$.

\begin{proof}[Proof of Lemma \ref{lem: piistruebelief}]
	We only need to prove the result for $\tau = t + 1$.
	
	Since $h_{t+1}^0$ is admissible under $(\tilde{g}_{1:t-1}^k, g_{t}^{*k})$, we have
	\begin{equation}
	\Pr^{\tilde{g}_{1:t-1}^k, g_{t}^{*k}, \hat{g}_{1:t}^{-k}}(h_{t+1}^0) > 0\label{probnonzero}
	\end{equation}
	where $\hat{g}_{1:t}^{-k}$ is the open-loop strategy where all coordinators except $k$ choose prescriptions that generate the actions $u_{1:t}^{-k}$.
	
	From Lemma \ref{lem: condindep} we know that $\Pr^{\tilde{g}_{1:t-1}^k, g_{t}^{*k}, g^{-k}}(s_{t+1}^k|h_{t+1}^0)$ is independent of $g^{-k}$. Therefore
	\begin{align}
	&\quad~\Pr^{\tilde{g}_{1:t-1}^k, g_{t}^{*k}}(s_{t+1}^k|h_{t+1}^0)\\
	&=\dfrac{\Pr^{\tilde{g}_{1:t-1}^k, g_{t}^{*k}, \hat{g}_{1:t}^{-k}}(s_{t+1}^k, y_t, u_t| h_{t}^0) }{\sum_{\tilde{s}_{t+1}^k}\Pr^{\tilde{g}_{1:t-1}^k, g_{t}^{*k}, \hat{g}_{1:t}^{-k}}(\tilde{s}_{t+1}^k, y_t, u_t| h_{t}^0)},\label{chopin}
	\end{align}
	and the denominator of \eqref{chopin} is non-zero due to \eqref{probnonzero}.
	
	We have
	\begin{align}
	&\quad~\Pr^{\tilde{g}_{1:t-1}^k, g_{t}^{*k}, \hat{g}_{1:t}^{-k}}(s_{t+1}^k, y_t, u_t|h_{t}^0)\\
	&=\sum_{\tilde{s}_t^k}\sum_{\tilde{x}_{t-d+1:t}^k }\sum_{\tilde{x}_t^{-k}} \sum_{ \tilde{\gamma}_t^k: \tilde{\gamma}_t^k(\tilde{x}_{t-d+1:t}^k) = u_t^k } \Big[\Pr(y_t^k|\tilde{x}_{t}^k, u_{t}) \times \\
	&\times \Pr(y_t^{-k}|\tilde{x}_{t}^{-k}, u_{t})\bm{1}_{ \{s_{t+1}^k = \iota_t^k(\tilde{s}_t^k, \tilde{x}_{t-d+1}^k, \tilde{\gamma}_t^k) \} } \lambda_{t}^{*k}(\tilde{\gamma}_{t}^k|b_{t}, \tilde{s}_{t}^k) \times\\ &\times\Pr^{\tilde{g}_{1:t-1}^k, g_{t}^{*k}, \hat{g}_{1:t}^{-k}}(\tilde{x}_{t-d+1:t}^k, \tilde{x}_t^{-k}, \tilde{s}_t^k|h_t^0)\Big] \\
	&=\sum_{\tilde{s}_t^k}\sum_{\tilde{x}_{t-d+1:t}^k }\sum_{\tilde{x}_t^{-k}} \sum_{ \tilde{\gamma}_t^k: \tilde{\gamma}_t^k(\tilde{x}_{t-d+1:t}^k) = u_t^k } \Big[\Pr(y_t^k|\tilde{x}_{t}^k, u_{t}) \times \\
	&\times \Pr(y_t^{-k}|\tilde{x}_{t}^{-k}, u_{t})\bm{1}_{ \{s_{t+1}^k = \iota_t^k(\tilde{s}_t^k, \tilde{x}_{t-d+1}^k, \tilde{\gamma}_t^k) \} } \lambda_{t}^{*k}(\tilde{\gamma}_{t}^k|b_{t}, \tilde{s}_{t}^k) \times\\ &\times\Pr^{\tilde{g}_{1:t-1}^k, g_{t}^{*k}, \hat{g}_{1:t}^{-k}}(\tilde{x}_{t-d+1:t}^k, \tilde{s}_t^k|h_t^0)\Pr^{\tilde{g}_{1:t-1}^k, g_{t}^{*k}, \hat{g}_{1:t}^{-k}}(\tilde{x}_t^{-k}| h_t^0) \Big] \\
	&= \left(\sum_{\tilde{x}_t^{-k}}  \Pr(y_t^{-k}|\tilde{x}_{t}^{-k}, u_{t}) \Pr^{\tilde{g}_{1:t-1}^k, g_{t}^{*k}, \hat{g}_{1:t}^{-k}}(\tilde{x}_t^{-k}| h_t^0) \right)\times\\
	&\times\sum_{\tilde{s}_t^k}\sum_{\tilde{x}_{t-d+1:t}^k } \sum_{ \tilde{\gamma}_t^k: \tilde{\gamma}_t^k(\tilde{x}_{t-d+1:t}^k) = u_t^k } \Big[\Pr(y_t^k|\tilde{x}_{t}^k, u_{t}) \times \\
	&\times \bm{1}_{ \{s_{t+1}^k = \iota_t^k(\tilde{s}_t^k, \tilde{x}_{t-d+1}^k, \tilde{\gamma}_t^k) \} } \lambda_{t}^{*k}(\tilde{\gamma}_{t}^k|b_{t}, \tilde{s}_{t}^k) \times\\ &\times\Pr^{\tilde{g}_{1:t-1}^k, g_{t}^{*k}, \hat{g}_{1:t}^{-k}}(\tilde{x}_{t-d+1:t}^k, \tilde{s}_t^k|h_t^0) \Big].\label{tchaikovsky}
	\end{align}
	
	Using \eqref{chopin} and \eqref{tchaikovsky} we obtain
	\begin{align}
	&\quad~\Pr^{\tilde{g}_{1:t-1}^k, g_{t}^{*k}}(s_{t+1}^k|h_{t+1}^0)\\
	&=\dfrac{\Upsilon_t^k(b_t, y_t^k, u_t, s_{t+1}^k)}{\sum_{ \tilde{s}_{t+1}^{k} }\Upsilon_t^k(b_t, y_t^k, u_t, \tilde{s}_{t+1}^k)}\label{provingbeliefupdate}
	\end{align}
	where
	\begin{align}
	&\quad~ \Upsilon_t^k(b_t, y_t^k, u_t, s_{t+1}^k)\\
	&=\sum_{\tilde{s}_t^k}\sum_{\tilde{x}_{t-d+1:t}^k } \sum_{ \tilde{\gamma}_t^k: \tilde{\gamma}_t^k(\tilde{x}_{t-d+1:t}^k) = u_t^k } \Big[\Pr(y_t^k|\tilde{x}_{t}^k, u_{t}) \times \\
	&\times \bm{1}_{ \{s_{t+1}^k = \iota_t^k(\tilde{s}_t^k, \tilde{x}_{t-d+1}^k, \tilde{\gamma}_t^k) \} } \lambda_{t}^{*k}(\tilde{\gamma}_{t}^k|b_{t}, \tilde{s}_{t}^k) \times\\ &\times P_t^k(\tilde{x}_{t-d+1:t}^k| y_{t-d+1:t-1}^k, u_{t-d:t-1}, s_{t}^k) \pi_t^k(s_t^k) \Big],
	\end{align}
	Therefore by the definition of consistency of $\psi^{*, k}$ with respect to $\lambda^{*k}$, we conclude that
	\begin{equation}
	\Pr^{\tilde{g}_{1:t-1}^k, g_{t}^{*k}}(s_{t+1}^k|h_{t+1}^0) = \pi_{t+1}^k(s_{t+1}^k).
	\end{equation}
	
	Now consider $\Pr^{\tilde{g}_{1:t-1}^k, g_{t}^{*k}}(\tilde{x}_{t-d+2:t+1}^k, s_{t+1}^k | h_{t+1}^0)$. 
	\begin{itemize}
		\item If $\Pr^{\tilde{g}_{1:t-1}^k, g_{t}^{*k}}(s_{t+1}^k | h_{t+1}^0) = 0$ then we have $\pi_{t+1}^k(s_{t+1}^k) = 0$ and
		\begin{align}
		&\quad~\Pr^{\tilde{g}_{1:t-1}^k, g_{t}^{*k}}(\tilde{x}_{t-d+2:t+1}^k, s_{t+1}^k | h_{t+1}^0) = 0.
		\end{align}
		
		\item If $\Pr^{\tilde{g}_{1:t-1}^k, g_{t}^{*k}}(s_{t+1}^k | h_{t+1}^0) > 0$ then
		\begin{align}
		&\quad~\Pr^{\tilde{g}_{1:t-1}^k, g_{t}^{*k}}(\tilde{x}_{t-d+2:t+1}^k, s_{t+1}^k | h_{t+1}^0) \\
		&=  \Pr^{\tilde{g}_{1:t-1}^k, g_{t}^{*k}}(\tilde{x}_{t-d+1:t}^k | h_{t+1}^0, s_{t+1}^k) \pi_{t+1}^k(s_{t+1}^k).
		\end{align}
		
		We have shown in Lemma \ref{lem: selfbelief} that
		\begin{align*}
		&\quad~\Pr^{\tilde{g}_{1:t-1}^k, g_{t}^{*k}}(\tilde{x}_{t-d+2:t+1}^k | \overline{h}_{t+1}^k)\\
		&=P_{t+1}^k(\tilde{x}_{t-d+2:t+1}^k| y_{t-d+2:t}^k, u_{t-d+1:t}, s_{t+1}^k)
		\end{align*}
		and $(h_{t+1}^0, s_{t+1}^k)$ is a function of $\overline{h}_{t+1}^k$. By the law of iterated expectation we have
		\begin{align*}
		&\quad~\Pr^{\tilde{g}_{1:t-1}^k, g_{t}^{*k}, \hat{g}_{1:t}^{-k}}(\tilde{x}_{t-d+2:t+1}^k | h_{t+1}^0, s_{t+1}^k)\\
		&=P_{t+1}^k(\tilde{x}_{t-d+2:t+1}^k| y_{t-d+2:t}^k, u_{t-d+1:t}, s_{t+1}^k).
		\end{align*}
	\end{itemize}
	
	We conclude that
	\begin{align}
	&\quad~\Pr^{\tilde{g}_{1:t-1}^k, g_{t}^{*k}}(\tilde{x}_{t-d+2:t+1}^k, s_{t+1}^k | h_{t+1}^0)\\
	&=P_t^k(\tilde{x}_{t-d+2:t+1}^k| y_{t-d+2:t}^k, u_{t-d+1:t}, s_{t+1}^k) \pi_{t+1}^k(s_{t+1}^k)
	\end{align}
	for all $s_{t+1}^k\in \mathcal{S}_{t+1}^k$ and all $x_{t-d+2:t+1}^k\in \mathcal{X}_{t-d+2:t+1}^k$. 
	
\end{proof}

\subsection{Proof of Lemma \ref{lem: closenessd1}}\label{app: closenessd1}
Let $g^{-i}$ denote the behavioral strategy profile of all coordinators other than $i$ generated from the CIB strategy profile $(\lambda^k, \psi^k)_{k\in\mathcal{I}\backslash\{ i\} }$. Let $(\overline{h}_t^i, \gamma_{t}^i)$ be admissible under $g^{-i}$. We have shown in the proof of Lemma \ref{lem: suffprivated3} that
\begin{align}
&\quad~\Pr^{g^{-i}}(x_{t-d+1:t}, \gamma_t^{-i}|\overline{h}_t^i, \gamma_{t}^i)\\
&=P_t^i(x_{t-d:t}^i| y_{t-d+1:t-1}^i, u_{t-d:t-1}, x_{t-d}^i, \phi_t^i) \times\\&\times\prod_{k\neq i}\Pr^{g^{k}}(x_{t-d+1:t}^{k}, \gamma_{t}^{k}|h_t^0).\label{bigproduct}
\end{align}
where $P_t^i$ is the belief function defined in Eq. \eqref{eq: belieffunction}.

Since all coordinators other than coordinator $i$ are using the same belief generation systems, we have $B_t^j=B_t^k$ for $j, k\neq i$. Denote $B_t=B_t^k$ for all $k\in\mathcal{I}\backslash\{i \}$. Let $b_t=\left(\left(\pi_t^{*, l}\right)_{l\in\mathcal{I}}, y_{t-d+1:t-1}, u_{t-d:t-1}\right)$ be a realization of $B_t$. Also define $\psi^*=\psi^k$ for all $k\neq i$.

Consider $k\neq i$. Coordinator $k$'s strategy $g^{k}$ is a self-consistent CIB strategy. We also have $h_t^0$ admissible under $g^{k}$ since $(\overline{h}_t^i, \gamma_{t}^i)$ is admissible under $g^{-i}$. Hence applying Lemma \ref{lem: piistruebelief0} we have
\begin{align}
&\quad~\Pr^{g^{k}}( \tilde{s}_{t}^k, x_{t-d+1:t}^{k}|h_t^0) \\
&= \pi_t^{*, k}( \tilde{s}_{t}^k)P_t^k(x_{t-d+1:t}^k| y_{t-d+1:t-1}^k, u_{t-d:t-1},  \tilde{s}_{t}^{k})
\end{align}

Hence the second term of the right hand side of \eqref{bigproduct} satisfies
\begin{align}
	&\quad~ \Pr^{g^{k}}(x_{t-d+1:t}^{k}, \gamma_{t}^{k}|h_t^0)= \sum_{ \tilde{s}_{t}^{k} } \Pr^{g^{k}}(\tilde{s}_{t}^{k}, x_{t-d+1:t}^k, \gamma_{t}^{k}|h_t^0) \\
	&= \sum_{ \tilde{s}_{t}^{k} } \Big[ \pi_t^{*, k}(\tilde{s}_{t}^k) P_t^k(x_{t-d+1:t}^k| y_{t-d+1:t-1}^k, u_{t-d:t-1}, \tilde{s}_{t}^{k}) \times\\
	&\quad\times \lambda_t^k(\gamma_{t}^k| b_t, \tilde{s}_{t}^{k}) \Big],\label{rhsofbigproduct}
\end{align}
where $P_t^k$ is the belief function defined in Eq. \eqref{eq: belieffunction}.

Recall that $b_t=\left(\left(\pi_t^{*, l}\right)_{l\in\mathcal{I}}, y_{t-d+1:t-1}, u_{t-d:t-1}\right)$. From \eqref{bigproduct} and \eqref{rhsofbigproduct} We conclude that
\begin{align}
&\Pr^{g^{-i}}(x_{t-d+1:t}, \gamma_{t}^{-i}|\overline{h}_t^i, \gamma_t^i )=F_t^i(x_{t-d+1:t}, \gamma_{t}^{-i}| b_t, s_{t}^i)\\
&\label{eq: ftid2}
\end{align}
for some function $F_t^i$ for all $(\overline{h}_t^i, \gamma_t^i)$ admissible under $g^{-i}$.

Consider the total reward of coordinator $i$. By the law of iterated expectation we can write
\begin{align}
J^{i}(\tilde{g}^i, g^{-i}) =\E^{\tilde{g}^i, g^{-i}}\left[\sum_{t\in\mathcal{T}}\E^{g^{-i}}[r_t^i(\bX_t, \bU_t)|\overline{H}_t^i, \bm{\Gamma}_t^i] \right].
\end{align}

For $(\overline{h}_t^i, \gamma_{t}^i)$ admissible under $g^{-i}$,
\begin{align}
&\quad~\E^{g^{-i}}[r_t^i(\bX_t, \bU_t)|\overline{h}_t^i, \gamma_t^i]\\
&= \sum_{\tilde{x}_{t-d+1:t}}\sum_{\tilde{\gamma}_t^{-i}} r_t^i(\tilde{x}_t, (\gamma_t^i(\tilde{x}_{t-d+1:t}^i), \tilde{\gamma}_t^{-i}(\tilde{x}_{t-d+1:t}^{-i}) ))\times \\
&\quad\times F_t^i(\tilde{x}_{t-d+1:t}, \tilde{\gamma}_t^{-i}| b_t, s_{t}^i)\\
&= \overline{r}_t^{i}(b_t, s_{t}^i, \gamma_t^i),\label{eq: inscost}
\end{align}
for some function $\overline{r}_t^{i}$ that depends on $g^{-i}$ (specifically, on $\lambda_t^{-i}$) but not on $\tilde{g}^i$. 

We claim that $(B_t, S_{t}^i)$ is a controlled Markov process controlled by coordinator $i$'s prescriptions, given that other coordinators are using the strategy profile $g^{-i}$. Let $\tilde{g}^i$ denote an arbitrary strategy for coordinator $i$ (not necessarily a CIB strategy). We need to prove that
\begin{align}
	&\quad~\Pr^{\tilde{g}^i, g^{-i}}(b_{t+1}, s_{t+1}^i| b_{1:t}, s_{1:t}^i, \gamma_{1:t}^i) \\
	&= \Xi_t^i(b_{t+1}, s_{t}^i| b_{t}, s_{t}^i, \gamma_{t}^i)\\
	&\quad \forall (b_{1:t}, s_{1:t}^i, \gamma_{1:t}^i) ~\mathrm{s.t.}~\Pr^{\tilde{g}^i, g^{-i}}(b_{1:t}, s_{1:t}^i, \gamma_{1:t}^i) > 0
\end{align}
for some function $\Xi_t^i$ independent of $\tilde{g}^i$.

We know that
\begin{align*}
B_{t+1} &= (\bm{\Pi}_{t+1}, \bY_{t-d+2:t}, \bU_{t-d+1:t}),\\
\bm{\Pi}_{t+1} &= \psi_t^*(B_t, \bY_{t}, \bU_t),\\
Y_t^k &= \ell_t^k(\bX_{t}^k, \bU_t, W_t^{k, Y})\quad\forall k\in\mathcal{I}, \\
U_t^{k, j} &= \Gamma_t^{k, j}(X_{t-d+1:t}^{k, j})\quad\forall (k, j)\in\mathcal{N},\\
S_{t+1}^{i}&= \iota_t^i(S_t^i, \bX_{t-d+1}^i, \bm{\Gamma}_t^i).
\end{align*}

Hence $(B_{t+1}, S_{t}^i)$ is a fixed function of $(B_t, S_t^i, \bX_{t-d+1:t}, \bm{\Gamma}_t, \mathbf{W}_t^Y)$, where $\mathbf{W}_t^Y$ is a primitive random vector independent of $(B_{1:t}, S_{1:t}^i, \bm{\Gamma}_{1:t}^i,  \bX_{t-d+1:t})$. Therefore, it suffices to prove that
\begin{equation}\label{jumpman}
\begin{split}
&~\quad\Pr^{\tilde{g}^i, g^{-i}}(x_{t-d+1:t}, \gamma_{t}^{-i}|b_{1:t}, s_{1:t}^i, \gamma_{1:t}^i) \\
&= \Xi_t^i(x_{t-d+1:t}, \gamma_{t}^{-i}| b_{t}, s_{t}^i, \gamma_{t}^i)
\end{split}
\end{equation}
for some function $\Xi_t^i$ independent of $\tilde{g}^i$.

$(B_{1:t}, S_{1:t}^i, \bm{\Gamma}_{1:t}^i)$ is a function of $(\overline{H}_t^i, \bm{\Gamma}_{t}^i)$. Therefore, by applying smoothing property of conditional expectations to both sides of \eqref{eq: ftid2} we obtain 
\begin{align}
&\quad~\Pr^{\tilde{g}^i, g^{-i}}(x_{t-d+1:t}, \gamma_{t}^{-i}|b_{1:t}, s_{1:t}^i, \gamma_{1:t}^i)\\
&=F_t^i(x_{t-d+1:t}, \gamma_{t}^{-i}| b_t, s_{t}^i),
\end{align}
where we know that $F_t^i$, as defined in \eqref{eq: ftid2}, is independent of $\tilde{g}^i$.

We conclude that coordinator $i$ faces a Markov Decision Problem where the state process is $(B_t, S_{t}^i)$, the control action is $\bm{\Gamma}_t^i$, and the total reward is $$\E\left[\sum_{t\in\mathcal{T}} \overline{r}_t^{i}(B_t, S_{t}^i, \bm{\Gamma}_t^i)\right].$$ By standard MDP theory, coordinator $i$ can form a best response by choosing $\bm{\Gamma}_t^i$ as a function of $(B_t, S_{t}^i)$.

\subsection{Proof of Theorem \ref{thm: sd1}}\label{app: sd1}
Let $(\lambda^*, \psi^*)$ be a pair that solves the dynamic program defined in the statement of the theorem. Let $g^{*k}$ denote the behavioral coordination strategy corresponding to $(\lambda^{*k}, \psi^*)$ for $k\in\mathcal{I}$.
We only need to show the following: Suppose that the coordinators other than coordinator $i$ play $g^{*-i}$, then $g^{*i}$ is a best response to $g^{*-i}$.

Let $h_t^0\in\mathcal{H}_t^0$ be admissible under $g^{*-i}$. Then
\begin{align}
&\quad~\Pr^{g^{*k}}(s_{t}^k, x_{t-d+1:t}^k|h_t^0)\\
&=\pi_t^{k}(s_{t}^k) P_t^k(x_{t-d+1:t}^k|y_{t-d+1:t-1}^k, u_{t-d:t-1}, s_{t}^k)\label{sufjanstevens}
\end{align}
for all $k\neq i$ by Lemma \ref{lem: piistruebelief0}, where $\pi_t^{k}$ is the belief generated by $\psi^*$ when $h_t^0$ occurs.

By Lemma \ref{lem: selfbelief} we also have
\begin{align}
	&\quad~\Pr(\tilde{s}_t^i, \tilde{x}_{t-d+1:t}^i |h_{t}^0, s_t^i)\\
	&=P_t^i(\tilde{x}_{t-d+1:t}^i|y_{t-d+1:t-1}^i, u_{t-d:t-1}, \tilde{s}_{t}^i)\label{sandyalexg}
\end{align}

Combining \eqref{sufjanstevens} and \eqref{sandyalexg}, the belief for coordinator $i$ defined in the stage game according to Definition \ref{def: stagegamed2} satisfies
\begin{align}
&\quad~\beta_t^i(\tilde{z}_t|s_{t}^i) \\
&= \bm{1}_{\{\tilde{s}_{t}^i=s_{t}^i \} } \prod_{k\neq i} \pi_t^{k}(\tilde{s}_{t}^k) \times\\
&\times\left(\prod_{k\in \mathcal{I}} P_t^k(\tilde{x}_{t-d+1:t}^k|y_{t-d+1:t-1}^k, u_{t-d:t-1}, \tilde{s}_{t}^k)\right)\Pr(\tilde{w}_t^{Y})\\
&=\Pr(\tilde{s}_t^i, \tilde{x}_{t-d+1:t}^i |h_{t}^0, s_t^i) \left(\prod_{k\neq i} \Pr^{g^{*k}}(\tilde{s}_t^k, \tilde{x}_{t-d+1:t}^k|h_t^0)\right)\Pr(\tilde{w}_t^{Y})\\
&=\Pr^{g^{*-i}}(\tilde{s}_t ,\tilde{x}_{t-d+1:t} |h_{t}^0, s_t^i)\Pr(\tilde{w}_t^{k, Y}) = \Pr^{g^{*-i}}(\tilde{z}_t|h_t^0, s_{t}^i)
\end{align}
for all $(h_t^0, s_{t}^i)$ admissible under $g^{*-i}$, i.e. the belief represents a true conditional distribution. Since $\beta_t^i(\cdot|s_{t}^i)$ is a fixed function of $(b_t, s_{t}^i)$, by applying smoothing property on both sides of the above equation we can obtain
\begin{align*}
\beta_t^i(\tilde{z}_t|s_{t}^i) = \Pr^{ g^{*-i}}(\tilde{z}_t|b_t, s_{t}^i).
\end{align*}
for all $(b_t, s_{t}^i)$ admissible under $g^{*-i}$. \footnote{Note that $\Pr^{ g^{-i}}(\tilde{z}_t|b_t, s_{t}^i)$ is different from $\beta_t^i(\tilde{z}_t|s_{t}^i)$. Since $B_t$ is just a compression of the common information based on an predetermined update rule $\psi$, which may or may not be consistent with the actually played strategy, $B_t$ may not represent the true belief. $\Pr^{ g^{-i}}(\tilde{z}_t|b_t, s_{t}^i)$ is the belief an agent \emph{inferred from} the event $B_t=b_t, S_t^i=s_t^i$. The agent knows that $b_t$ might not contain the true belief, but it is useful anyway in inferring the true state. $\beta_t^i(\tilde{z}_t|s_{t}^i)$ is a conditional distribution \emph{computed with} $b_t$, \emph{pretending} that $b_t$ contains the true belief.}

Then the interim expected utility considered in the definition of IBNE correspondences (Definition \ref{def: ibe}) can be written as
\begin{align*}
&\quad~\sum_{\tilde{z}_t, \tilde{\gamma}_t} \eta(\tilde{\gamma}_{t}^i) Q_t^i(\tilde{z}_t, \tilde{\gamma}_t) \beta_t^i(\tilde{z}_t|x_{t-1}^i)\prod_{k\neq i}\lambda_t^{*k}(\tilde{\gamma}_t^k|b_t, \tilde{s}_{t}^k) \\
&=\sum_{\tilde{\gamma}_t^i} \eta(\tilde{\gamma}_{t}^i) \E^{g_{1:t}^{*-i}}[Q_t^i(\mathbf{Z}_t, \bm{\Gamma}_{t})|b_t, s_{t}^i, \tilde{\gamma}_{t}^i].
\end{align*}
for all $(b_t, s_{t}^i)$ admissible under $g^{*-i}$.

The condition of Theorem \ref{thm: sd1} then implies
\begin{align}
\lambda_t^{*i}&(b_t, s_{t}^i) \in \underset{\eta\in \Delta(\prescription_{t}^i)}{\arg\max} \sum_{\tilde{\gamma}_t} \eta(\tilde{\gamma}_{t}^i) \E^{g^{*-i}}[r_t^i(\bX_t, \bU_t) + \\
&+V_{t+1}^i(B_{t+1}, S_{t+1}^i)|b_t, s_{t}^i,\tilde{\gamma}_{t}^i];\label{eq: stageopt}\\
V_t^i&(b_t, s_{t}^i)= \sum_{\tilde{\gamma}_t^i} \Big[\lambda_t^{*i}(\tilde{\gamma}_t^i|b_t, s_{t}^i) \times\\ &\times\E^{g_{1:t}^{*-i}}[r_t^i(\bX_t, \bU_t) +V_{t+1}^i(B_{t+1}, S_{t+1}^i)|b_t, s_{t}^i, \tilde{\gamma}_{t}^i]\Big]~~\label{eq: valueupdate}
\end{align}
for all $(b_t, s_{t}^i)$ admissible under $g^{*-i}$. 

Recall that in the proof of Lemma \ref{lem: closenessd1}, we have already proved that fixing $(\lambda^{*-i}, \psi^*)$, $(B_t, S_{t}^i)$ is a controlled Markov process controlled by $\bm{\Gamma}_{t}^i$. Hence \eqref{eq: stageopt} and \eqref{eq: valueupdate} show that $\lambda_t^{*i}$ is a dynamic programming solution of the MDP with instantaneous reward
	\begin{equation}
	\overline{r}_t^i(B_t, S_{t}^i, \bm{\Gamma}_{t}^i):= \E^{g^{*-i}}[r_t^i(\bX_t, \bU_t)|B_t, S_{t}^i, \bm{\Gamma}_{t}^i].
	\end{equation}

Therefore, $\lambda^{*i}$ maximizes 
\begin{equation}
\E^{\lambda^i, \lambda^{*-i}}\left[\sum_{t\in\mathcal{T}} \overline{r}_t^i(B_t, S_{t}^i, \bm{\Gamma}_t^i) \right]
\end{equation}
over all $\lambda^i=(\lambda_t^i)_{t\in\mathcal{T}}, \lambda_t^i: \mathcal{B}_t\times\mathcal{S}_t^i \mapsto \Delta(\prescription_t^i)$. 

Notice that for any $\lambda^i$, if $g^i$ is the behavioral coordination strategy corresponding to the CIB strategy $(\lambda^i, \psi_t^*)$, then by Law of Iterated Expectation
\begin{align*}
\E^{\lambda^i, \lambda^{*-i}}\left[\sum_{t\in\mathcal{T}} \overline{r}_t^i(B_t, S_{t}^i, \bm{\Gamma}_t^i) \right] &= \E^{g^i, g^{*-i}}\left[\sum_{t\in\mathcal{T}} r_t^i(\bX_t, \bU_t) \right].
\end{align*}

Hence we know that $g^{*i}$ maximizes
\begin{equation}\label{geq325}
\E^{g^{i}, g^{*-i}}\left[\sum_{t\in\mathcal{T}} r_t^i(\bX_t, \bU_t) \right]
\end{equation}
over all $g^i$ generated from a CIB strategy with the belief generation system $\psi^*$.

By the closedness property of CIB strategies (Lemma \ref{lem: closenessd1}), we conclude that $g^{*i}$ is a best response to $g^{*-i}$ over all behavioral coordination strategies of coordinator $i$, proving the result.

\subsection{Proof of Proposition \ref{prop: nonexistenceexample}}\label{app: nonexistenceexample}
We will characterize all the Bayes-Nash Equilibria of Example \ref{ex: nonexistence} in terms of individual players' behavioral strategies. Then we will show that none of the BNE correspond to a CIB-CNE. 

Let $p=(p_1, p_2)\in [0, 1]^2$ describe Alice's behavioral strategy: $p_1$ is the probability that Alice plays $U_1^A=-1$ given $X_1^A=-1$; $p_2$ is the probability that Alice plays $U_1^A=+1$ given $X_1^A=+1$. Let $q=(q_1, q_2)\in [0, 1]^2$ denote Bob's behavioral strategy: $q_1$ is the probability that Bob plays $U_3^B=\mathrm{L}$ when observing $U_1^A=-1$, $q_2$ is the probability that Bob plays $U_3^B=\mathrm{L}$ when observing $U_1^A=+1$.

\textbf{Claim:}
\begin{equation}
p^*=\left(\frac{1}{3}, \frac{1}{3}\right),\quad q^*=\left(\frac{1}{3}+\varepsilon, \frac{1}{3}-\varepsilon\right)
\end{equation}
is the unique BNE of Example \ref{ex: nonexistence}.

Given the claim, one can conclude that a CIB-CNE does not exist in this game: Suppose that $(\lambda^*, \psi^*)$ forms a CIB-CNE, Then by the definition of CIB strategies, at $t=1$ the team of Alice chooses a prescription (which maps $\mathcal{X}_1^A$ to $\mathcal{U}_1^A$) based on no information. At $t=3$, the team of Bob chooses a prescription (which is equivalent to an action since Bob has no state) based solely on $B_3$. Define the induced behavioral strategy of Alice and Bob through
\begin{align*}
p_1 &= \lambda_1^{*A}(\mathbf{id}|\varnothing) + \lambda_1^{*A}(\mathbf{cp}_{-1}|\varnothing),\\
p_2 &= \lambda_1^{*A}(\mathbf{id}|\varnothing) + \lambda_1^{*A}(\mathbf{cp}_{+1}|\varnothing),\\
q_1 &= \lambda_3^{*B}(\mathbf{L}|b_3[-1]),\\
q_2 &= \lambda_3^{*B}(\mathbf{L}|b_3[+1]),
\end{align*}
where $b_3[u]$ is the CCI under belief generation system $\psi^*$ when $U_1^A=u$. $\mathbf{id}$ is the prescription that chooses $U_1^A=X_1^A$; $\mathbf{cp}_{u}$ is the prescription that chooses $U_1^A=u$ irrespective of $X_1^A$; $\mathbf{L}$ is Bob's prescription that chooses $U_3^B=\mathrm{L}$.

The consistency of $\psi_1^*$ with respect to $\lambda_1^*$ implies that
\begin{align}
\Pi_2(-1) &= \dfrac{p_1}{p_1+1-p_2}\quad\text{ if }p\neq (0, 1), U_1=-1,\\
\Pi_2(+1) &= \dfrac{p_2}{p_2+1-p_1}\quad\text{ if }p\neq (1, 0), U_1=+1,
\end{align}

The consistency of $\psi_2^*$ with respect to $\lambda_2^*$ implies that
\begin{align}
\Pi_3(+1) &= \Pi_2(U_1^A).
\end{align}

If a CIB-CNE induces behavioral strategy $p^*=\left(\frac{1}{3}, \frac{1}{3}\right)$, then the CIB belief $\Pi_3\in \Delta(\mathcal{X}_2)$ will be the same for both $U_1=+1$ and $U_1=-1$ under any consistent belief generation system $\psi^*$. Then $B_3=(\Pi_3, \bU_2)$ will be the same for both $U_1=+1$ and $U_1=-1$ since $\bU_2$ only takes one value. Hence Bob's induced stage behavioral strategy $q$ should satisfy $q_1=q_2$. However $q^*=\left(\frac{1}{3}+\varepsilon, \frac{1}{3}-\varepsilon\right)$ is such that $q_1^*\neq q_2^*$, hence $(p^*, q^*)$ cannot be induced from any CIB-CNE. 

Since the induced behavioral strategy of any CIB-CNE should form a BNE in the game among individuals, we conclude that a CIB-CNE does not exist in Example \ref{ex: nonexistence}.\\ 

\textbf{Proof of Claim:}
Denote Alice's total expected payoff to be $J(p, q)$. Then 
\begin{align*}
J(p, q) &= \frac{1}{2}\varepsilon(1-p_1+p_2) + \frac{1}{2} \left((1-p_1)(1-q_2) + p_1 \cdot 2q_1 \right) +\\
&+ \frac{1}{2} \left((1-p_2)(1-q_1) + p_2\cdot 2q_2 \right)\\
&=\frac{1}{2}\varepsilon(1-p_1+p_2) + \frac{1}{2} (2 - p_1 - p_2) +\\
&+ \frac{1}{2}(2p_1 + p_2 - 1)q_1 + \frac{1}{2}(2p_2 + p_1 - 1)q_2.
\end{align*}

Since this is a zero-sum game, Alice's expected payoff at equilibrium can be characterized as
\begin{align*}
J^* = \max_{p} \min_q J(p, q)
\end{align*}

Alice plays $p$ at some equilibrium if and only if $\min_q J(p, q) = J^*$. Define $J^*(p) = \min_q J(p, q)$. We compute
\begin{align*}
J^*(p) &= \frac{1}{2}\varepsilon(1-p_1+p_2) + \frac{1}{2} (2 - p_1 - p_2) + \\
&+\begin{cases}
\frac{1}{2}(3p_1 + 3p_2) - 1& 2p_1 + p_2 \leq 1, 2p_2 + p_1 \leq 1\\
\frac{1}{2}(2p_2 + p_1 - 1) & 2p_1 + p_2 > 1, 2p_2 + p_1 \leq 1\\
\frac{1}{2}(2p_1 + p_2 - 1) & 2p_1 + p_2 \leq 1, 2p_2 + p_1 > 1\\
0& 2p_1 + p_2 > 1, 2p_2 + p_1 > 1
\end{cases}
\end{align*}

The set of equlibrium strategies for Alice is the set of maximizers of $J^*(p)$. Since $J^*(p)$ is a continuous piecewise linear function, the set of maximizers can be found by comparing the values at the extreme points of the pieces.

We have
\begin{align*}
J^*(0, 0) &= \frac{1}{2}\varepsilon + 1 - 1 = \frac{1}{2}\varepsilon;\\
J^*\left(\frac{1}{2}, 0\right) &= \frac{1}{2}\varepsilon \cdot \frac{1}{2} + \frac{1}{2} \cdot \frac{3}{2} + \frac{1}{2} \cdot \frac{3}{2} - 1= \frac{1}{4}\varepsilon + \frac{1}{2};\\
J^*\left(0, \frac{1}{2}\right) &= \frac{1}{2}\varepsilon\cdot \frac{3}{2} + \frac{1}{2} \cdot \frac{3}{2} + \frac{1}{2} \cdot \frac{3}{2} - 1= \frac{3}{4}\varepsilon + \frac{1}{2};\\
J^*(1, 0) &= \frac{1}{2}\varepsilon \cdot 0+ \frac{1}{2}\cdot 1 + \frac{1}{2}\cdot 0 = \frac{1}{2};\\
J^*(0, 1) &= \frac{1}{2}\varepsilon \cdot 2 + \frac{1}{2}\cdot 1 + \frac{1}{2} \cdot 0 = \varepsilon + \frac{1}{2};\\
J^*\left(\frac{1}{3}, \frac{1}{3}\right) &= \frac{1}{2}\varepsilon + \frac{1}{2}\cdot \frac{4}{3} + \frac{1}{2} \cdot 0  = \frac{1}{2}\varepsilon + \frac{2}{3};\\
J^*(1, 1) &= \frac{1}{2}\varepsilon + \frac{1}{2}\cdot 0 + 0 = \frac{1}{2}\varepsilon.
\end{align*}

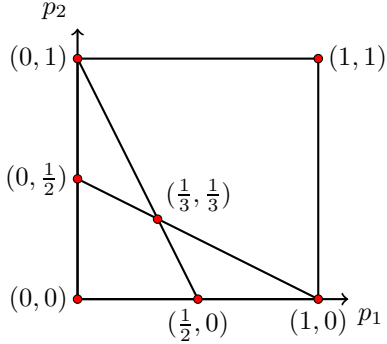
\begin{figure}[!ht]
	\centering
	\begin{tikzpicture}[scale=0.8]
	\draw[thick, ->] (0, 0) -- (4.5, 0) node[anchor=north west] {$p_1$};
	\draw[thick, ->] (0, 0) -- (0, 4.5) node[anchor=south east] {$p_2$};
	
	\draw[thick] (4, 0) -- (4, 4) node[anchor=west] {$(1, 1)$};
	\draw[thick] (4, 4) -- (0, 4) node[anchor=east] {$(0, 1)$};
	\draw[thick] (0, 4) -- (0, 0) node[anchor=east] {$(0, 0)$};
	\draw[thick] (0, 0) -- (4, 0) node[anchor=north] {$(1, 0)$};
	
	\draw[thick] (0, 4) -- (2, 0) node[anchor=north] {$(\frac{1}{2}, 0)$};
	\draw[thick] (4, 0) -- (0, 2) node[anchor=east] {$(0, \frac{1}{2})$};
	\draw[thick] (1.3, 1.3) node[anchor=south west] {$(\frac{1}{3}, \frac{1}{3})$};
	
	\foreach \point in {(0, 0), (4, 0), (0, 4), (2, 0), (0, 2), (1.33, 1.33), (4, 4)}
	\draw [fill=red] \point circle [radius=0.07];
	\end{tikzpicture}
	\caption{The pieces (polygons) for which $J^*(p)$ is linear on. The extreme points of the pieces are labeled.} \label{fig: extremepoints}
\end{figure}

Since $\varepsilon < \frac{1}{3}$, we have $(\frac{1}{3}, \frac{1}{3})$ to be the unique maximum among the extreme points. Hence we have $\arg\max_p J^*(p) = \{(\frac{1}{3}, \frac{1}{3}) \}$, i.e. Alice always plays $p^*=(\frac{1}{3}, \frac{1}{3})$ in any BNE of the game.

Now, consider Bob's equilibrium strategy. $q^*$ is an equilibrium strategy of Bob only if $p^* \in \arg\max_{p} J(p, q^*)$.

For each $q$, $J(p, q)$ is a linear function of $p$ and
\begin{align*}
\nabla_p J(p, q) = \left(-\frac{1}{2}\varepsilon - \frac{1}{2} + q_1 + \frac{1}{2}q_2, \frac{1}{2}\varepsilon - \frac{1}{2} + \frac{1}{2}q_1 + q_2 \right)&\\
\quad \forall p\in (0, 1)^2.&
\end{align*}

We need $\nabla_p J(p, q^*)\Big|_{p=p^*} = (0, 0)$. Hence
\begin{align*}
-\frac{1}{2}\varepsilon - \frac{1}{2} + q_1^* + \frac{1}{2}q_2^* &= 0;\\
\frac{1}{2}\varepsilon - \frac{1}{2} + \frac{1}{2}q_1^* + q_2^* &= 0,
\end{align*}
which implies that $q^*=(\frac{1}{3}+\varepsilon, \frac{1}{3}-\varepsilon)$, proving the claim.

\subsection{Proof of Theorem \ref{thm: case1}}\label{app: case1}
We use Theorem \ref{thm: sd1} to establish the existence of CIB-CNE: We show that for each $t$ there always exists a pair $(\lambda_t^*, \psi_t^*)$ such that $\lambda_t^*$ forms an equilibrium at $t$ given $\psi_t^*$, and $\psi_t^*$ is consistent with $\lambda_t^*$. We provide a constructive proof of existence of CIB-CNE by proceeding backwards in time.

Since $d=1$ we have $S_t^i = \bX_{t-1}^i$. The CCI consists of the beliefs along with $\bU_{t-1}$.

Consider the condensation of the information graph into a directed acyclic graph (DAG) whose nodes are strongly connected components. Each node may contain multiple teams. Consider one topological ordering of this DAG. Denote the nodes by $[1], [2], \cdots$ ($[j]$ is reachable from $[k]$ only if $k < j$.) We use the notation $X_t^{[k]}, \Pi_t^{[k]}$ to denote the vector of the system variables of the teams in a node. In particular, following Definition \ref{def: stagegamed2}, we define $\mathbf{Z}_t^{[k]} = (\bX_{t-1:t}^{[k]}, \mathbf{W}_t^{[k], Y})$. We also use $[1:k]$ as a short hand for the set $[1]\cup [2]\cup \cdots\cup [k]$. Define $B_t^{[1:k]} = (\Pi_t^{[1:k]}, \bU_{t-1}^{[1:k]})$. (Note that the usage of superscript here is different from the CCI $B_t^i$ defined in Definition \ref{def: cci2}.)

We construct the solution first backwards in time, then in the order of the node for each stage. For that matter, we need some induction invariant on the value functions $V_t^i$ (as defined in Theorem \ref{thm: sd1}) for the solution we are going to construct.

\textbf{Induction Invariant:} For each time $t$ and each node index $k$,
\begin{itemize}
	\item $V_t^{i}(b_t, x_{t-1}^{i})$ depends on $b_t$ only through $(b_t^{[1:k-1]}, u_{t-1}^{i})$ for all teams $i\in [k]$, if $[k]$ consists of only one team. (With some abuse of notation, we write $V_t^{i}(b_t, x_{t-1}^{i}) = V_t^{i}(b_t^{[1:k-1]}, u_{t-1}^{i}, x_{t-1}^{i})$ in this case.)
	
	\item $V_t^{i}(b_t, x_{t-1}^{i})$ depends on $b_t$ only through $b_t^{[1:k]}$ for all teams $i\in [k]$, if $[k]$ consists of multiple public teams. (We write $V_t^{i}(b_t, x_{t-1}^{i}) = V_t^{i}(b_t^{[1:k]}, x_{t-1}^{i})$ in this case.)
\end{itemize}

\textbf{Induction Base:} For $t=T+1$ we have $V_{T+1}^i(\cdot)\equiv 0$ for all coordinators $i\in\mathcal{I}$ hence the induction invariant is true.

\textbf{Induction Step:} Suppose that the induction invariant is true at time $t+1$ for all nodes. We construct the solution so that it is also true at time $t$. 

To complete this step we provide a procedure to solve the stage game. We argue that one can solve a series of optimization problems or finite games following the topological order of the nodes through an inner induction step. 

\textbf{Inner Induction Step:} Suppose that the first $k-1$ nodes has been solved, and the equilibrium strategy $\lambda_{t}^{*[1:k-1]}$ uses only $b_t^{[1:k-1]}$ along with private information. Suppose that the update rules $\psi_t^{*,[1:k-1]}$ have also been determined, and they use only $(b_t^{[1:k-1]}, y_t^{[1:k-1]}, u_t^{[1:k-1]})$. We now establish the same property for $(\lambda_{t}^{[k]}, \psi_t^{[k]})$.

\begin{itemize}
	\item If the $k$-th node contains a single coordinator $i$, the value to go is $V_{t+1}^i(B_{t+1}^{[1:k-1]}, \bU_t^i, \bX_t^i)$ by the induction hypothesis. The instantaneous reward for a coordinator $i$ in the $k$-th node can be expressed by $r_t^i(\bX_t^{[1:k]}, \bU_{t}^{[1:k]})$ by the information graph. In the stage game, coordinator $i$ chooses a prescription to maximize the expected value of
	\begin{align*}
		&\quad~ Q_t^i(b_{t}^{[1:k-1]}, \mathbf{Z}_t^{[1:k]}, \bm{\Gamma}_t^{[1:k]})\\
		&:=r_t^i(\bX_t^{[1:k]}, \bU_{t}^{[1:k]}) + V_{t+1}^i(B_{t+1}^{[1:k-1]}, \bU_{t}^i, \bX_t^i),
	\end{align*}
	where
	\begin{align*}
	B_{t+1}^{[1:k-1]} &= (\Pi_{t+1}^{[1:k-1]}, \bU_t^{[1:k-1]} ),\\
	\Pi_{t+1}^{j}&=\psi_t^{*,j}(b_t^{[1:k-1]}, \bY_t^{j}, \bU_t^{[1:k-1]})\quad\forall j\in [1:k-1], \\
	\bY_{t}^{j} &= \ell_t^j(\bX_{t}^j, \bU_t^{[1:k-1]}, \mathbf{W}_t^{j, Y})\quad\forall j\in [1:k-1],\\
	\bU_{t}^j &= \bm{\Gamma}_{t}^j(\bX_t^j)\quad\forall j\in [1:k].
	\end{align*}
	
	The expectation is computed using the belief $\beta_t^i$ (defined through Eq. \eqref{beliefstagegame2} in Definition \ref{def: stagegamed2}) along with $\lambda_{t}^{*[1:k-1]}$ that has already been determined. It can be written as
	\begin{align*}
	&\quad\sum_{\tilde{s}_t, \tilde{\gamma}_t^{[1:k-1]}} \beta_t^i(\tilde{s}_t|x_{t-1}^i) Q_t^i(b_t^{[1:k-1]}, \tilde{s}_t^{[1:k]}, (\tilde{\gamma}_t^{[1:k-1]}, \gamma_{t}^i) )\times \\
	&\times\prod_{j\in [1:k-1]} \lambda_{t}^j(\tilde{\gamma}_t^j|b_t^{[1:k-1]}, \tilde{x}_{t-1}^j) \\
	&=\sum_{\tilde{s}_{t}^{[1:k]}, \tilde{\gamma}_t^{[1:k-1]}} \bm{1}_{ \{\tilde{x}_{t-1}^i=x_{t-1}^i \} } \Pr(\tilde{w}_t^{[1:k], Y})\times \\
	&\times\left(\prod_{j\in [1:k-1]}  \pi_t^{j}(\tilde{x}_{t-1}^j)\Pr(\tilde{x}_t^j|\tilde{x}_{t-1}^j, u_{t-1}^{[1:k-1]}) \right)\times\\
	&\times\left(\prod_{j\in [1:k-1]} \lambda_{t}^{*j}(\tilde{\gamma}_t^j|b_t^{[1:k-1]}, x_{t-1}^j) \right)\times\\
	&\times \Pr(\tilde{x}_t^i|x_{t-1}^i, u_{t-1}^{[1:k]})  Q_t^i(b_t^{[1:k-1]}, \tilde{s}_t^{[1:k]}, (\tilde{\gamma}_t^{[1:k]}, \gamma_{t}^i)).
	\end{align*}
	Therefore, the expected reward of coordinator $i$ depends on $b_t$ through $(b_t^{[1:k-1]}, u_{t-1}^i)$. Coordinator $i$ can choose the optimal prescription based on $(b_t^{[1:k-1]}, u_{t-1}^i, x_{t-1}^i)$, i.e. $\lambda_t^{*i}(b_t, x_{t-1}^i)=\lambda_t^{*i}(b_t^{[1:k-1]}, u_{t-1}^i, x_{t-1}^i)$. We then have $V_t^i(b_t, x_{t-1}^i) = V_t^i(b_t^{[1:k-1]}, u_{t-1}^i, x_{t-1}^i)$. The update rule $\psi_t^{*, [k]} =\psi_t^{*, i}$ is then determined to be an arbitrary update rule consistent with $\lambda_{t}^{*, i}$, which can be chosen as a function from $\mathcal{B}_t^{[1:k]}\times \mathcal{Y}_t^{[k]}\times \mathcal{U}_t^{[1:k]}$ (instead of $\mathcal{B}_t\times \mathcal{Y}_t^{[k]}\times \mathcal{U}_t$) to $\varPi_{t+1}^{[k]}$.
	
	\item If the $k$-th node contains a group of public teams, then update rules $\hat{\psi}_t^{*, [k]}$ are fixed, irrespective of the stage game strategies, i.e. there exist a unique update rule $\hat{\psi}_t^{*, i}$ that is compatible with any $\lambda_{t}^{*, i}$ for a public team $i$. This update rule is a map from $\mathcal{Y}_t^{[k]}\times \mathcal{U}_t^{[1:k]}$ to a vector of delta measures on $\prod_{i\in [k]} \Delta(\mathcal{X}_{t-1}^i)$, i.e. the map to recover $\bX_{t-1}^{[k]}$ from the observations (see Definition \ref{def: publicteam}). The function takes $\bU_{t}^{[1:k]}$ as its argument due to the fact that the observations of the $k$-th node depends on $\bU_t$ only through $\bU_t^{[1:k]}$.
	
	The value to go for each coordinator $i$ can be expressed as $V_{t+1}^i(B_t^{[1:k]}, \bX_{t-1}^i)$ by induction hypothesis. The instantaneous reward can be written as $r_t^i(\bX_t^{[1:k]}, \bU_{t}^{[1:k]})$ by the definition of the information dependency graph.
	
	In the stage game, coordinator $i$ in the $k$-th node chooses a distribution $\eta_{t}^i$ on prescriptions to maximize the expected value of
	\begin{align*}
	&\quad~Q_t^i(b_{t}^{[1:k]}, \mathbf{Z}_t^{[1:k]}, \bm{\Gamma}_{t}^{[1:k]}) \\
	&:=r_t^i(\bX_t^{[1:k]}, \bU_{t}^{[1:k]}) + V_{t+1}^i(B_{t+1}^{[1:k]}, \bX_t^i),
	\end{align*}
	where
	\begin{align*}
	B_{t+1}^{[1:k]} &= (\Pi_{t+1}^{[1:k]}, \bU_t^{[1:k]} ),\\
	\Pi_{t+1}^{j}&=\psi_t^{*, j}(b_t^{[1:k-1]}, \bY_t^{j}, \bU_t^{[1:k-1]})\quad\forall j\in [1:k-1], \\
	\Pi_{t+1}^{[k]}&=\hat{\psi}_t^{*, [k]}(b_{t}^{[1:k]}, \bY_{t}^{[1:k]}, \bU_{t}^{[1:k]} ),\\
	\bY_{t}^{j} &= \ell_t^j(\bX_{t}^j, \bU_t^{[1:k]}, \mathbf{W}_t^{j, Y})\quad\forall j\in [1:k],\\
	\bU_{t}^j &= \bm{\Gamma}_{t}^j(\bX_t^j)\quad\forall j\in [1:k].
	\end{align*}
	The expectation is taken with respect to the belief $\beta_t^i$ (defined through Eq. \eqref{beliefstagegame2} in Definition \ref{def: stagegamed2}) and the strategy prediction $\lambda_t^{[1:k]}$. This expectation can be written as
	\begin{align*}
	&\quad\sum_{\tilde{s}_t, \tilde{\gamma}_t^{[1:k]}} \beta_t^i(\tilde{s}_t|x_{t-1}^i) Q_t(b_t^{[1:k]}, \tilde{s}_t^{[1:k]}, \tilde{\gamma}_t^{[1:k]} ) \eta_{t}^{i}(\tilde{\gamma}_t^i) \times \\
	&\times \prod_{\substack{ j\in [1:k]\\ j\neq i}} \lambda_{t}^j(\tilde{\gamma}_t^j|b_t^{[1:k-1]}, \tilde{x}_{t-1}^j) \\
	&=\sum_{\tilde{s}_{t}^{[1:k]}, \tilde{\gamma}_t^{[1:k]}} \bm{1}_{ \{\tilde{x}_{t-1}^i = x_{t-1}^i \} } \Pr(\tilde{w}_t^{[1:k], Y}) \times\\
	&\left(\prod_{\substack{ j\in [1:k]\\ j\neq i}} \pi_{t}^j(\tilde{x}_{t-1}^j)\Pr(\tilde{x}_t^j|\tilde{x}_{t-1}^j, u_{t-1}^{[1:k]})\lambda_{t}^{*j}(\tilde{\gamma}_t^j|b_t^{[1:k]}, \tilde{x}_{t-1}^j) \right)\times \\
	&\times \Pr(\tilde{x}_t^i|x_{t-1}^i, u_{t-1}^{[1:k]}) \eta_{t}^{i}(\tilde{\gamma}_t^i) Q_t^i(b_t^{[1:k]}, \tilde{s}_t^{[1:k]}, \tilde{\gamma}_t^{[1:k]}),
	\end{align*}
	which dependents only on $b_t$ only through $b_t^{[1:k]}$.
	Therefore, the stage game defined in Definition \ref{def: stagegamed2} induces a finite game between the coordinators in the $k$-th node (instead of all coordinators) with parameter $(b_t^{[1:k]}, (\psi_t^{*, [1:k-1]}, \hat{\psi}_t^{*,[k]}))$ (instead of $(b_t, \psi_t)$), where $\lambda_t^{*[1:k-1]}$ has been fixed. Teams in the $k$-th node form/play a stage game where the first $k-1$ nodes act like nature, while the coordinators after $k$-th node have no effect in the payoffs of the coordinators in the $k$-th node. Hence, a coordinator $i$ in the $k$-th node can based their decision on $(b_t^{[1:k]}, x_{t-1}^i)$, i.e. $\lambda_t^{*i}(b_t, x_{t-1}^i)=\lambda_t^{*i}(b_t^{[1:k]}, x_{t-1}^i)$. We also have $V_t^i(b_t, x_{t-1}^i) = V_t^i(b_t^{[1:k]}, x_{t-1}^i)$. The update rule is determined by $\psi_t^{*,[k]} = \hat{\psi}_t^{*,[k]}$, which is guaranteed to be consistent with $\lambda_t^{*[k]}$.
\end{itemize}

In summary, we determine $(\lambda_{t}^*, \psi_t^*)$ using a node-by-node approach. If the $k$-th node consists of one team, then we first determine $\lambda_{t}^{*[k]}$ from an optimization problem dependent on $(\lambda_{t}^{*[1:k-1]}, \psi_t^{*,[1:k-1]})$, and then determine $\psi_t^{*,[k]}$. If the $k$-th node consists of multiple public players, then we first determine $\psi_t^{*,[k]}$ and then solve $\lambda_{t}^{*[k]}$ from a finite game dependent on $(\lambda_{t}^{*[1:k-1]}, \psi_t^{*, [1:k]})$. Hence we have constructed the solution and established both inner and outer induction steps, proving the theorem.

\subsection{Proof of Theorem \ref{thm: signalfreegame}}\label{app: signalfreegame}
We prove the Theorem for $d=1$. The proof idea for $d>1$ is similar. 

We will prove a stronger result. 
For each $\Pi_t^i\in \Delta(\mathcal{X}_{t-1}^i)$, define the corresponding $\hat{\Pi}_t^i\in\Delta(\mathcal{X}_t)$ by
\begin{equation}
	\overline{\Pi}_t^i(x_t^i):= \sum_{\tilde{x}_{t-1}^i}\Pi_t^i(\tilde{x}_{t-1}^i) \Pr(x_t^i|\tilde{x}_{t-1}^i).
\end{equation}

Define $\hat{\psi}_t^i$ to be the signaling-free update function, i.e. the belief update function such that
\begin{align*}
\Pi_{t+1}^i(x_t^i) &=\hat{\psi}_t^i(\overline{\Pi}_t^i, \bY_t^i)= \dfrac{\overline{\Pi}_t^i(x_{t}^i)\Pr(\bY_t^i|x_t^i) }{\sum_{\tilde{y}_t^i} \overline{\Pi}_t^i(x_{t}^i)\Pr(\tilde{y}_t^i|x_t^i)}.
\end{align*}

Define \emph{open-loop} prescriptions as the prescriptions that simply instruct members of a team to take a certain action irrespective their private information. We will show that there exist an equilibrium where each team plays a common information based signaling-free (CIBSF) strategy, i.e. the common belief generation system for all coordinators is given by the signaling-free update functions $\hat{\psi}$, and coordinator $i$ chooses randomized open-loop prescriptions based on $\overline{\bm{\Pi}}_t=(\overline{\Pi}_t^i)_{i\in\mathcal{I}}$ instead of $(B_t, \bX_{t-1}^i)$.

\textbf{Induction Invariant:} $V_t^i(B_t, \bX_{t-1}^i) = V_t^i(\overline{\bm{\Pi}}_t, \bX_{t-1}^i)$.

\textbf{Induction Base:} The induction variant is true for $t=T+1$ since $V_{T+1}^i(\cdot) \equiv 0$ for all $i\in\mathcal{I}$.

\textbf{Induction Step:} Suppose that the induction variant is true for $t+1$, prove it for time $t$.

Let $\hat{\psi}_t$ be the signaling-free update rule. We solve the stage game $G_t(V_{t+1}, \hat{\psi}_t, b_t)$. In the stage game, coordinator $i$ chooses a prescription to maximize the expectation of
\begin{align*}
r_t^i(\bX_t^{-i}, \bU_{t}) + V_{t+1}^i(\overline{\bm{\Pi}}_{t+1}, \bX_t^i),
\end{align*}
where
\begin{align*}
\overline{\Pi}_{t+1}^k(x_{t+1}^k)&= \sum_{\tilde{x}_t^k} \Pi_{t+1}^k(\tilde{x}_t^k)\Pr(x_{t+1}^k|\tilde{x}_t^k)\quad\forall x_{t+1}^k\in\mathcal{X}_{t+1}^k, \\
\Pi_{t+1}^k&=\hat{\psi}_t^k(\overline{\Pi}_{t}^k, \bY_{t}^k)\quad\forall k\in\mathcal{I},\\
\bY_{t}^{k} &= \ell_t^k(\bX_{t}^k, \mathbf{W}_t^{k, Y})\quad\forall k\in\mathcal{I},\\
U_{t}^{k, j} &= \Gamma_{t}^{k, j}(X_t^{k, j})\quad\forall (k, j)\in\mathcal{N}.
\end{align*}

Since $V_{t+1}^i(\overline{\bm{\Pi}}_{t+1}, \bX_t^i)$ does not depend on coordinator $i$'s prescriptions, coordinator $i$ only need to maximize the expectation of $r_t^i(\bX_t^{-i}, \bU_{t})$, which is
\begin{align*}
	&\sum_{\tilde{x}_{t-1:t}^{-i}, \tilde{\gamma}_t^{-i} } \left(\prod_{j\neq i} \pi_t^{j}(\tilde{x}_{t-1}^{j}) \Pr(\tilde{x}_t^{j}|\tilde{x}_{t-1}^{j}) \lambda_{t}^{j}(\tilde{\gamma}_t^{j}|b_t, \tilde{x}_{t-1}^{j}) \right) \times \\
	&\quad\times r_t^i(\tilde{x}_t^{-i}, (\tilde{\gamma}_t^{-i}(\tilde{x}_t^{-i}), \gamma_{t}^i(x_t^i) ) ). 
\end{align*}

\textbf{Claim:} In the stage game, if all coordinators $-i$ use CIBSF strategy, then coordinator $i$ can respond with a CIBSF strategy.

\begin{proof}[Proof of Claim:]
	Let $\eta_{t}^k: \overline{\varPi}_t \mapsto \Delta(\mathcal{U}_t^k)$ be the CIBSF strategy of coordinator $k\neq i$. Then coordinator $i$'s expected payoff given $\gamma_{t}^i$ can be written as
	 \begin{align*}
	 &\quad~\sum_{\tilde{x}_{t-1:t}^{-i}, \tilde{u}_t^{-i} } \left(\prod_{j\neq i} \pi_t^{j}(\tilde{x}_{t-1}^{j}) \Pr(\tilde{x}_t^{j}|\tilde{x}_{t-1}^{j}) \eta_{t}^{j}(\tilde{u}_t^{j}|\overline{\pi}_t) \right) \times \\
	 &\quad\times r_t^i(\tilde{x}_t^{-i}, (\tilde{u}_t^{-i}, \gamma_{t}^i(x_t^i) ) ) \\
	 &=\sum_{\tilde{x}_{t}^{-i}, \tilde{u}_t^{-i} } \left(\prod_{j\neq i} \left(\sum_{\tilde{x}_{t-1}^j}\pi_t^{j}(\tilde{x}_{t-1}^{j}) \Pr(\tilde{x}_t^{j}|\tilde{x}_{t-1}^{j})\right) \eta_{t}^{j}(\tilde{u}_t^{j}|\overline{\pi}_t) \right) \times \\
	 &\quad\times r_t^i(\tilde{x}_t^{-i}, (\tilde{u}_t^{-i}, \gamma_{t}^i(x_t^i) ) ) \\
	 &=\sum_{\tilde{x}_{t}^{-i}, \tilde{u}_t^{-i} } \left(\prod_{j\neq i} \overline{\pi}_t^{j}(\tilde{x}_{t}^{j}) \eta_{t}^{j}(\tilde{u}_t^{j}|\overline{\pi}_t) \right) r_t^i(\tilde{x}_t^{-i}, (\tilde{u}_t^{-i}, \gamma_{t}^i(x_t^i) ) ) \\
	 &=:\overline{r}_t^i(\overline{\pi}_t, \eta_{t}^{-i}, \gamma_{t}^i(x_t^i)).
	 \end{align*}
	 
	 Hence coordinator $i$ can respond with a prescription $\gamma_{t}^i$ such that $\gamma_{t}^i(x_t^i) = u_t^i$ for all $x_t^i$, where
	 \begin{align*}
	 	u_t^i \in \arg\max_{\tilde{u}_t^i} \overline{r}_t^i(\overline{\pi}_t, \eta_{t}^{-i}, \tilde{u}_t^i),
	 \end{align*} 
	 can be chosen based on $(\overline{\pi}_t, \eta_{t}^{-i})$, proving the claim.
\end{proof}

Given the claim, we conclude that there exist a stage game equilibrium where all coordinators play CIBSF strategies: Define a new stage game where we restrict each coordinator to CIBSF strategies. A best response in the restricted stage game will be also a best response in the original stage game due to the claim. The restricted game is a finite game (It is a game of symmetrical information with parameter $\overline{\pi}_t$ where coordinator $i$'s action is $u_t^i$ and its payoff is a function of $\overline{\pi}_t$ and $u_t$.) that always has an equilibrium. The equilibrium strategy will be consistent with $\hat{\psi}_t$ due to Lemma \ref{lem: sigfreeupdate}.

\begin{lemma}\label{lem: sigfreeupdate}
	The signaling-free update rule $\hat{\psi}_t^i$ is consistent with any $\lambda_t^i: \mathcal{B}_t\times \mathcal{X}_{t-1}^i \mapsto \Delta(\prescription_{t}^i)$ that corresponds to a CIBSF strategy at time $t$.
\end{lemma}

\begin{proof}
	Can be done with standard arguments for strategy independence of belief.
\end{proof}

Let $\eta_{t}^{*}=(\eta_{t}^{*j})_{j\in\mathcal{I}}, \eta_{t}^{*j}: \overline{\varPi}_t \mapsto \Delta(\mathcal{U}_t^j)$ be a CIBSF strategy profile that is a stage game equilibrium. Then the value function
\begin{align*}
	&\quad ~V_{t}^i(b_t, x_{t-1}^i) = \left(\max_{\tilde{u}_t^i} \overline{r}_t^i(\overline{\pi}_t, \eta_{t}^{*-i}, \tilde{u}_t^i) \right) + \\
	&+\sum_{\tilde{x}_t, \tilde{y}_t} V_{t+1}^i(\hat{\psi}_t(\overline{\pi}_{t}, \tilde{y}_t), \tilde{x}_{t}^i) \Pr(\tilde{y}_t|\tilde{x}_t)\Pr(\tilde{x}_t^i|x_{t-1}^i)\overline{\pi}_t^{-i}(\tilde{x}_t^{-i})
\end{align*}
depends on $(b_t, x_{t-1}^i)$ only through $(\overline{\pi}_t, x_{t-1}^i)$, establishing the induction step.

\subsection{Proof of Lemma \ref{lem: suffhidden}}\label{app: suffhidden}
For ease of illustration we prove the result for $d=2$. The result for $d=1$ is trivially true, and the result for $d>2$ can be proved following a similar logic to that of this proof.

The key idea is to apply person-by-person refinement of a team strategy. 
Let $g^{-i}$ be some behavioral coordination strategy profile for coordinators other than coordinator $i$. Let $\mu^i$ denote a pure team strategy that is a best response to $g^{-i}$.\footnote{Note that we are not considering a coordination strategy and no randomization is considered. At time $t$, agent $(i, j)$ decides on her action through $u_t^{i, j} = \mu_t^{i, j}(h_t^{i, j})$.} To proceed we first prove the following lemma.

\begin{lemma}\label{lem: refineteam}
	Fixing $g^{-i}$, for any pure team strategy profile $\mu^i$ and any $(i, j)\in\mathcal{N}_i$, there exist a pure team strategy profile $\tilde{\mu}^i$ such that (1) $\tilde{\mu}_t^{i, j}(h_t^{i, j})$ does not depend on $x_{t-1}^i$; (2) $\tilde{\mu}^{i, -j} = \mu^{i, -j}$; (3) $J^i(\tilde{\mu}^i, g^{-i})\geq J^i(\mu^i, g^{-i})$.
\end{lemma}

Given the result of Lemma \ref{lem: refineteam}, we can refine any best-response pure strategy $\mu^i$ in a person-by-person manner to obtain a pure strategy in which $\mu_t^{i, j}(h_t^{i, j})$ does not depend on $x_{t-1}^{i, j}$ for all $(i, j)\in\mathcal{N}_i$. Then, one can transform the new pure strategy $\mu^i$ into one of its equivalent pure coordination strategies $\nu^i$, where $\nu^i$ always assigns simple prescriptions.

\begin{proof}[Proof of Lemma \ref{lem: refineteam}]
	Fix the strategy $\mu^{i, -j}$ for members of team $i$ other than $(i, j)$, and also fix $g^{-i}$ for other teams. We refine agent $(i, j)$'s strategy so as to maximize team $i$'s expected reward.
	
	We argue that agent $(i, j)$ is facing a POMDP problem with:
	\begin{itemize}
		\item State: $(\bY_{1:t-1}, \bU_{1:t-1}, \bX_{1:t}^{-i}, \bm{\Gamma}_{1:t-1}^{-i}, \bX_{1:t}^{i, -j}, X_t^{i, j})$
		\item Observation: $H_t^{i, j}=(\bY_{1:t-1}, \bU_{1:t-1}, \bX_{1:t-2}^i, X_{t-1:t}^{i, j})$
		\item Action: $U_t^{i, j}$
		\item Instantaneous reward: $r_t^i(\bX_t, \bU_t)$
	\end{itemize}
	where $\bU_t^{-i}$ follows the distribution induced from the random prescriptions generated by $g_t^{-i}$ and $(H_t^0, \bX_{1:t}^{-i}, \bm{\Gamma}_{1:t-1}^{-i})$ and $\bU_t^{i, -j}$ are generated from $\mu_t^{i, -j}$.
	
	By the standard POMDP structural result, the conditional distribution of the state given observations is an information state for agent $(i, j)$. Notice that $X_{t-1}^{i, j}$ only appears in the observation but not in the state. Furthermore, $\bY_{1:t-1}, \bU_{1:t-1}$, and $X_t^{i, j}$ are perfectly observed by agent $(i, j)$. Therefore, to prove that agent $(i, j)$ does not need to use $X_{t-1}^{i, j}$, it is sufficient to prove the following claim: 
	
	\textbf{Claim:} $\Pr^{\mu^{i, -j}, g^{-i}}(x_{1:t}^{-i}, \gamma_{1:t-1}^{-i}, x_{t-1:t}^{i, -j} |h_t^{i, j})$ does not depend on $x_{t-1}^{i, j}$.
	
	\begin{proof}[Proof of Claim]
		Due to conditional independence among different teams (Lemma \ref{lem: condindep}), we have
		\begin{align}
		&~\quad\Pr^{\mu^{i, -j}, g^{-i}}(x_{1:t}^{-i}, \gamma_{1:t-1}^{-i}, x_{t-1:t}^{i, -j} |h_t^i, x_{t-1:t}^{i, j})\\
		&=\Pr^{\mu^{i, -j}}(x_{t-1:t}^{i, -j} |h_t^i, x_{t-1:t}^{i, j}) \prod_{k\neq i} \Pr^{g^{-i}}(x_{1:t}^k, \gamma_{1:t-1}^{k}|h_t^0).\label{tameimpala}
		\end{align}
		
		Note that the first conditional belief term on the right hand side of \eqref{tameimpala} does not depend on strategy $g^{-i}$ (by Lemma \ref{lem: condindep}) and $\mu_{1:t-1}^{i, j}$ (by the standard policy-independence property of belief in POMDP). Therefore, one can consider the above conditional belief term assuming that other teams play according to open-loop strategy $\hat{g}^{-i}$, which generates $u_{1:t-1}^{-i}$, and agent $(i, j)$ plays according to open-loop strategy $\hat{g}^{i, j}$, which generates $u_{1:t-1}^{i, j}$. 
		
		Consequently, letting $\hat{x}_{t-1}^{-i} \in\mathcal{X}_{t-1}^{-i}$ be arbitrary and using Bayes' rule we obtain
		\begin{equation}
		\begin{split}
		&\quad~\Pr^{\mu^{i, -j}}(x_{t-1:t}^{i, -j} |h_t^i, x_{t-1:t}^{i, j}, \hat{x}_{t-1}^{-i})\\
		&=\dfrac{\Pr^{\hat{g} }(x_{t-1:t}^{i, -j}, y_{t-1}, u_{t-1}, x_{t-1:t}^{i, j}|h_{t-1}^0, x_{1:t-2}^i, \hat{x}_{t-1}^{-i})}{\sum_{\tilde{x}_{t-1:t}^{i, -j}} \Pr^{\hat{g} }(\tilde{x}_{t-1:t}^{i, -j}, y_{t-1}, u_{t-1}, x_{t-1:t}^{i, j}|h_{t-1}^0, x_{1:t-2}^i, \hat{x}_{t-1}^{-i}) }
		\end{split}\label{dividedown}
		\end{equation}
		where $\hat{g} = (\hat{g}^{i, j}, \mu^{i, -j}, \hat{g}^{-i})$.
		
		We have
		\begin{align}
		&\quad~ \Pr^{\mu^{i, -j}, \hat{g}^{i, j}, \hat{g}^{-i}}(x_{t-1:t}^{i, -j}, y_{t-1}, u_{t-1}, x_{t-1:t}^{i, j}|h_{t-1}^0, x_{1:t-2}^i, \hat{x}_{t-1}^{-i})\!\!\!\!\!\!\!\!\\
		&=\Pr^{\mu^{i, -j}, \hat{g}^{i, j}, \hat{g}^{-i}}(x_t^{i, -j}, x_t^{i, j}| y_{1:t-1}, u_{1:t-1}, x_{1:t-1}^i, \hat{x}_{t-1}^{-i})\\
		&\quad\times \Pr^{\mu^{i, -j}, \hat{g}^{i, j}, \hat{g}^{-i}}(y_{t-1}|y_{1:t-2}, u_{1:t-1}, x_{1:t-1}^i, \hat{x}_{t-1}^{-i})\\
		&\quad\times \Pr^{\mu^{i, -j}, \hat{g}^{i, j}, \hat{g}^{-i}}(u_{t-1}^i|y_{1:t-2}, u_{1:t-2}, x_{1:t-1}^i, \hat{x}_{t-1}^{-i})\\
		&\quad \times \Pr^{\mu^{i, -j}, \hat{g}^{i, j}, \hat{g}^{-i}}(x_{t-1}^i|y_{1:t-2}, u_{1:t-2}, x_{1:t-2}^i, \hat{x}_{t-1}^{-i})\\
		&=\Pr(x_t^{i, -j}|x_{t-1}^{i, -j}, u_{t-1}) \Pr(x_t^{i, j}|x_{t-1}^{i, j}, u_{t-1}) { \Pr(y_{t-1}^{i, j}|x_{t-1}^{i, j}, u_{t-1}) }\!\!\!\!\!\!\!\!\!\!\\
		&\quad {\times \Pr(y_{t-1}^{i, -j}|x_{t-1}^{i, -j}, u_{t-1}) }
		\Pr(y_{t-1}^{-i}|\hat{x}_{t-1}^{-i}, u_{t-1})\\
		&\quad\times \bm{1}_{\{\mu_{t-1}^{i, -j}(y_{1:t-2}, u_{1:t-2}, x_{1:t-3}^i, x_{t-2:t-1}^{i, -j}) = u_{t-1}^{i, -j} \}}\\
		&\quad \times\Pr(x_{t-1}^{i, -j}|x_{t-2}^{i, -j}, u_{t-2}) \Pr(x_{t-1}^{i, j}|x_{t-2}^{i, j}, u_{t-2})\\
		&=F_t(x_{t-1:t}^{i, -j}, h_t^i) \cdot G_t(x_{t-1:t}^{i, j}, \hat{x}_{t-1}^{-i}, h_t^i)\label{sepdynproductterms}
		\end{align}
		for some functions $F_t$ and $G_t$.
		
		Combining \eqref{dividedown} and \eqref{sepdynproductterms} we obtain
		\begin{align}
		&\quad~\Pr^{\mu^{i, -j}}(x_{t-1:t}^{i, -j} |h_t^0, x_{1:t-2}^i, x_{t-1:t}^{i, j}, \hat{x}_{t-1}^{-i})\\
		&=\dfrac{F_t(x_{t-1:t}^{i, -j}, h_t^i)  \cdot G_t(x_{t-1:t}^{i, j}, \hat{x}_{t-1}^{-i}, h_t^i)}{\sum_{\tilde{x}_{t-1:t}^{i, -j}} F_t(\tilde{x}_{t-1:t}^{i, -j}, h_t^i)  \cdot G_t(x_{t-1:t}^{i, j}, \hat{x}_{t-1}^{-i}, h_t^i)}\\
		&=\dfrac{F_t(x_{t-1:t}^{i, -j}, h_t^i) }{\sum_{\tilde{x}_{t-1:t}^{i, -j}} F_t(\tilde{x}_{t-1:t}^{i, -j}, h_t^i) }
		\end{align}
		which is independent of $x_{t-1}^{i, j}$. Hence we proved the claim. 
	\end{proof}
	
	Therefore, agent $(i, j)$ can solve the POMDP problem and obtain an optimal strategy $\tilde{\mu}^{i, j}$ that does not use $X_{t-1}^{i, j}$ to choose actions. Define $\tilde{\mu}^i = (\tilde{\mu}^{i, j}, \mu^{i, -j})$. One can verify that conditions (1)-(3) of Lemma \ref{lem: refineteam} are satisfied.
\end{proof}

\subsection{Proof of Proposition \ref{prop: pbe}}\label{app: pbe}
For each $t\in\mathcal{T}$ and $h_t^0\in\mathcal{H}_t^0$, let $(\pi_t^k)_{k\in\mathcal{I}}$ be the beliefs generated from the belief generation system $\psi^*$. Let $b_t=\left(\left(\pi_t^k\right)_{k\in\mathcal{I}}, y_{t-d+1:t-1}, u_{t-d:t-1}\right)$.

We first define the belief system $\vartheta^*$. Consider coordinator $i\in\mathcal{I}$ and $\overline{h}_t^i\in\overline{\mathcal{H}}_t^i$. Let $\hat{g}_{1:t-1}^i$ be the open loop coordination strategy that generates the prescriptions $\gamma_{1:t-1}^i$. We only need to consider realizations $\overline{h}_t^i\in\overline{H}_t^i$ that are admissible under $\hat{g}_{1:t-1}^i$. We define $\vartheta_t^{*i}(\overline{h}_t^i)$ as following.
\begin{itemize}
	\item Case I: $\Pr^{\hat{g}_{1:t-1}^i, g_{1:t-1}^{*-i}}(\overline{h}_t^i) > 0$. Define
	\begin{equation}
	\vartheta_t^{*i}(h_{t}^*|\overline{h}_t^i) = \Pr^{\hat{g}_{1:t-1}^i, g_{1:t-1}^{*-i}}(h_{t}^*| \overline{h}_t^i).\label{norahjones}
	\end{equation}
	
	\item Case II: $\Pr^{\hat{g}_{1:t-1}^i, g_{1:t-1}^{*-i}}(\overline{h}_t^i) = 0$. We define
	\begin{equation}
	\vartheta_t^{*i}(h_{t}^*|\overline{h}_t^i) = \Pr^{\hat{g}_{1:t-1}^i, \tilde{g}_{1:t-1}^{-i}}(h_{t}^*| \overline{h}_t^i).\label{johnmayer}
	\end{equation}
	where $\tilde{g}_{1:t-1}$ is an arbitrary strategy profile which satisfies $\Pr^{\tilde{g}_{1:t-1}}(h_t^0) > 0$ and $\prod_{k\in \mathcal{I}}\pi_t^k(s_t^k) = \Pr^{\tilde{g}_{1:t-1}}(s_t|h_t^0)$ for all $s_t^k\in\mathcal{S}_t^k$ (since $\psi^*$ is regular there exist such $\tilde{g}_{1:t-1}^k$). 
	
	Let $g_{1:t-1}^{-i}$ be such that $\Pr^{\hat{g}_{1:t-1}^i, g_{1:t-1}^{-i}}(\overline{h}_t^i) > 0$. Then \eqref{johnmayer} is well defined since
	\begin{align}
		&\quad~\Pr^{\hat{g}_{1:t-1}^i, \tilde{g}_{1:t-1}^{-i}}(\overline{h}_t^i) \\
		&=  \Pr^{\hat{g}_{1:t-1}^i, \tilde{g}_{1:t-1}^{-i}}(x_{1:t-d}^i, \gamma_{1:t-1}^i|h_t^0) \Pr^{\hat{g}_{1:t-1}^i, \tilde{g}_{1:t-1}}(h_t^0)\\
		&= \Pr^{\hat{g}_{1:t-1}^i}(x_{1:t-d}^i, \gamma_{1:t-1}^i|h_t^0) \Pr^{\hat{g}_{1:t-1}^i, \tilde{g}_{1:t-1}^{-i}}(h_t^0)\label{jimihendrix}
	\end{align}
	where in the last equality we used the strategy independence result in Lemma \ref{lem: condindep}.
	The first term in \eqref{jimihendrix} is non-zero since $\overline{h}_t^i$ is admissible under $\hat{g}_{1:t-1}^i$. The second term in \eqref{jimihendrix} is non-zero since $\Pr^{\tilde{g}_{1:t-1}}(h_t^0) > 0$. Therefore $\Pr^{\hat{g}_{1:t-1}^i, \tilde{g}_{1:t-1}^{-i}}(\overline{h}_t^i) > 0$.
	
\end{itemize}

It is clear that $\vartheta^*$ is consistent with $g^*$.
Now, we would argue below that $g^{*i}$ is sequentially rational with respect to $\vartheta^{*i}$.

By the construction of the stage game (Definition \ref{def: stagegamed2}), $(\lambda_t^{*i})_{t\in\mathcal{T}}$ is a solution to the dynamic program of the following MDP:
\begin{itemize}
	\item The state process is $(B_t, S_{t}^i)$;
	\item The control action is $\bm{\Gamma}_{t}^i$;
	\item The transition kernel is given by
	\begin{align}
	&\quad~\Pr(\tilde{b}_{t+1}, s_{t+1}^i| b_t, s_t^i, \gamma_t^i)\\
	&=\sum_{\tilde{s}_t, \tilde{x}_{t-d+1:t}} \sum_{\tilde{\gamma}_t^{-i}: \tilde{\gamma}_t^{-i}(\tilde{x}_{t-d+1:t}^{-i}) = \tilde{u}_t^{-i} } \bm{1}_{ \{\tilde{\bm{\pi}}_{t+1} = \psi_t^*(b_t, \tilde{y}_t, \tilde{u}_t) \} } \times\\
	&\times \bm{1}_{ \{\tilde{u}_t^i = \gamma_t^i(\tilde{x}_{t-d+1:t}^i) \} } \bm{1}_{ \{\tilde{y}_{t-d+2:t-1} = y_{t-d+2:t-1}\}} \times\\
	&\times \bm{1}_{\{ \tilde{u}_{t-d+1:t-1} = u_{t-d+1:t-1} \} } \Pr(\tilde{y}_t|\tilde{x}_t, \tilde{u}_t) \times\\
	&\times \bm{1}_{\{\tilde{s}_t^i = s_t^i \}  } \left(\prod_{k\neq i} \lambda_t^{*k}(\tilde{\gamma}_t^k|b_t, \tilde{s}_{t}^k) \pi_t^k(\tilde{s}_{t}^k)\right) \times \\
	&\times \prod_{k\in \mathcal{I}} P_t^k(\tilde{x}_{t-d+1:t}^k| y_{t-d+1:t-1}^k, u_{t-d:t-1}, \tilde{s}_{t}^k)\Big]\label{bsupdate}
	\end{align}
	\item The instantaneous reward is $\overline{r}_t^i(B_t, S_{t}^i, \bm{\Gamma}_{t}^i)$, where
	\begin{align}
	&\quad~\overline{r}_t^i(b_t, s_{t}^i, \gamma_{t}^i) \\
	&= \sum_{\tilde{s}_t, \tilde{x}_{t-d+1:t}, \tilde{\gamma}_t^{-i}} \Big[r_t^i(\tilde{x}_t, (\tilde{\gamma}_t^{-i}(\tilde{x}_{t-d+1:t}^{-i}), \gamma_t^i(\tilde{x}_{t-d+1:t}^i)))\times\\
	&\times \bm{1}_{\{\tilde{s}_t^i = s_t^i \}  } \left(\prod_{k\neq i} \lambda_t^{*k}(\tilde{\gamma}_t^k|b_t, \tilde{s}_{t}^k) \pi_t^k(\tilde{s}_{t}^k)\right) \times \\
	&\times \prod_{k\in \mathcal{I}} P_t^k(\tilde{x}_{t-d+1:t}^k| y_{t-d+1:t-1}^k, u_{t-d:t-1}, \tilde{s}_{t}^k)\Big]\label{eqbelief1}
	\end{align}
\end{itemize}

Now we investigate the condition for sequential rationality of $g^{*i}$ w.r.t. $\vartheta^*$ according to Definition \ref{def: wpbe}. Fix $t$ and $\overline{h}_t^i\in\overline{\mathcal{H}}_t^i$, we have
\begin{align}
	&~\quad J_t^i(g_{t:T}^i, g_{t:T}^{*-i}; \vartheta_t^*, \overline{h}_t^i)\\
	&=\E^{\vartheta_t^{*i}(\overline{h}_t^i), g_{t:T}^i, g_{t:T}^{*-i}}\left[\sum_{\tau=t}^T r_{\tau}^i(\bX_{\tau}, \bU_{\tau})\right]\\
	&=\E^{\vartheta_t^{*i}(\overline{h}_t^i), g_{t:T}^i, g_{t:T}^{*-i}}\left[\sum_{\tau=t}^T \tilde{r}_{\tau}^i(\overline{H}_{\tau}^i, \bm{\Gamma}_{\tau}^i) \right]\label{flyinglotus}
\end{align}
where
\begin{align}
	&\quad~\tilde{r}_{\tau}^i(\overline{h}_{\tau}^i, \gamma_{\tau}^i)\\
	&=\E^{\vartheta_t^{*i}(\overline{h}_t^i), g_{t:T}^i, g_{t:T}^{*-i}} [r_{\tau}^i(\bX_{\tau}, \bU_{\tau}) |\overline{h}_{\tau}^i,  \gamma_{\tau}^i]\\
	&= \sum_{\tilde{s}_{\tau}, \tilde{x}_{\tau-d+1:\tau}, \tilde{\gamma}_{\tau}^{-i} }\Big[r_{\tau}^i(\tilde{x}_{\tau}, (\tilde{\gamma}_{\tau}^{-i}(\tilde{x}_{\tau-d+1:\tau}^{-i}), \gamma_{\tau}^i(\tilde{x}_{\tau-d+1:\tau}^i)))\times\\
	& \quad~\times\Pr^{\vartheta_t^{*i}(\overline{h}_t^i), g_{t:T}^i, g_{t:T}^{*-i}}(\tilde{s}_{\tau}, \tilde{x}_{\tau-d+1:\tau}, \tilde{\gamma}_{\tau}^{-i} |\overline{h}_{\tau}^i)
\end{align}
for all $\overline{h}_{\tau}^i\in \mathcal{H}_\tau^i$ such that 
\begin{equation}
\Pr^{\vartheta_t^{*i}(\overline{h}_t^i), g_{t:T}^i, g_{t:T}^{*-i}}(\overline{h}_{\tau}^i) > 0.\label{thundercat}
\end{equation}

To analyze \eqref{flyinglotus}, we only need to consider $\tilde{r}_\tau(\overline{h}_{\tau}^i, \gamma_\tau^i)$ for $\overline{h}_{\tau}^i$ that satisfies \eqref{thundercat}.

Let $\tilde{g}_{1:t-1}^{-i}$ be either $g_{1:t-1}^{*-i}$ (when Case I is met) or $\tilde{g}_{1:t-1}^{-i}$ in \eqref{johnmayer} (when Case II is met). Then
\begin{align}
&\quad~\Pr^{\tilde{g}_{1:t-1}^{k}}(\tilde{s}_{t}^k, \tilde{x}_{t-d+1:\tau}^k| h_{t}^0) \\
&= \pi_t^k(\tilde{s}_{t}^k) P_t^k(\tilde{x}_{t-d+1:t}^k| y_{t-d+1:t-1}^k, u_{t-d:t-1}, \tilde{s}_{t}^k) 
\end{align}
for all $k\neq i$.
Furthermore, \eqref{thundercat} implies that $\Pr^{\hat{g}_{1:t-1}^i, \tilde{g}_{1:t-1}^{-i}, g_{t:T}^i, g_{t:T}^{*-i}}(\overline{h}_{\tau}^i) > 0$. In particular, $h_{\tau}^0$ is admissible under $(\tilde{g}_{1:t-1}^{k}, g_{t:\tau-1}^{*k})$. Therefore, by Lemma \ref{lem: piistruebelief} we conclude that 
\begin{align}
&\quad~\Pr^{\tilde{g}_{1:t-1}^{k}, g_{t:T}^{*k}}(\tilde{s}_{\tau}^k, \tilde{x}_{\tau-d+1:\tau}^k| h_{\tau}^0) \\
&= \pi_\tau^k(\tilde{s}_{\tau}^k) P_\tau^k(\tilde{x}_{\tau-d+1:\tau}^k| y_{\tau-d+1:\tau-1}^k, u_{\tau-d:\tau-1}, \tilde{s}_{\tau}^k) 
\end{align}
for all $k\neq i$. Therefore
\begin{align}
&\quad~\Pr^{\vartheta_t^{*i}(\overline{h}_t^i), g_{t:T}^i, g_{t:T}^{*-i}}(\tilde{s}_{\tau}, \tilde{x}_{\tau-d+1:\tau}, \tilde{\gamma}_{\tau}^{-i}|\overline{h}_{\tau}^i ) \\
&= \Pr^{\hat{g}_{1:t-1}^i, \tilde{g}_{1:t-1}^{-i}, g_{t:T}^i, g_{t:T}^{*-i}}(\tilde{s}_{\tau}, \tilde{x}_{\tau-d+1:\tau}, \tilde{\gamma}_{\tau}^{-i}|\overline{h}_{\tau}^i )\\
&= \Pr(\tilde{s}_{\tau}^i, \tilde{x}_{\tau-d+1:\tau}^i|\overline{h}_{\tau}^i ) \prod_{k\neq i}\Pr^{\tilde{g}_{1:t-1}^{k}, g_{t:T}^{*k}}(\tilde{s}_{\tau}^k, \tilde{x}_{\tau-d+1:\tau}^k, \tilde{\gamma}_{\tau}^{k}|\overline{h}_{\tau}^0 )\\
&= \bm{1}_{\{\tilde{s}_t^i = s_t^i \}  } \left(\prod_{k\neq i}  \lambda_{\tau}^{*k}(\tilde{\gamma}_\tau^k|b_\tau, \tilde{s}_{\tau}^k) \pi_\tau^k(\tilde{s}_{\tau}^k) \right) \times \\
&\times \prod_{k\in \mathcal{I}} P_\tau^k(\tilde{x}_{\tau-d+1:\tau}^k| y_{\tau-d+1:\tau-1}^k, u_{\tau-d:\tau-1}, \tilde{s}_{\tau}^k)\label{eqbelief2}
\end{align}

Comparing \eqref{eqbelief2} with \eqref{eqbelief1}, we conclude that $\tilde{r}_{\tau}^i(\overline{h}_{\tau}^i, \gamma_{\tau}^i) = \overline{r}_\tau^i(b_\tau, s_\tau^i)$ for $\overline{h}_{\tau}^i\in \mathcal{H}_\tau^i$ such that $\Pr^{\vartheta_t^{*i}(\overline{h}_t^i), g_{t:T}^i, g_{t:T}^{*-i}}(\overline{h}_{\tau}^i) > 0$. 

Similarly, we can show that under the measure generated by $(\vartheta_t^{*i}(\overline{h}_t^i), g_{t:T}^i, g_{t:T}^{*-i})$, $(B_t, S_t^i)_{t\geq \tau}$ is a controlled Markov Chain with control action $\bm{\Gamma}_t^i$ that has the same transition kernel as described in \eqref{bsupdate}. Therefore we conclude that $g_{t:T}^{*i}$ (generated from $\lambda_{t:T}^*$ and $\psi^*$) optimizes $J_t^i(g_{t:T}^i, g_{t:T}^{*-i}; \vartheta_t^*, \overline{h}_t^i)$ over all $g_{t:T}^i$, proving sequential rationality.

\end{document}